\documentclass{elsarticle}

\usepackage{graphicx}
\usepackage{amssymb}
\usepackage{amsthm}

\biboptions{sort&compress}

\usepackage{url}
\usepackage{color}
\usepackage{soul}
\usepackage{multirow}
\usepackage{alltt}

\newtheorem{theorem}{Theorem}[section]
\newtheorem{lemma}[theorem]{Lemma}
\newtheorem{proposition}[theorem]{Proposition}
\newtheorem{corollary}[theorem]{Corollary}

\newdefinition{definition}{Definition}
\newdefinition{example}{Example}
\newdefinition{remark}{Remark}

\newcommand{\calL}{\Sigma}
\newcommand{\calF}{\Sigma^{\mathbf{f}}}
\newcommand{\ext}{\textsf{ext}}

\newcommand{\langOf}[1]{\ensuremath{[#1]}}
\newcommand{\langOfAutom}[1]{\ensuremath{[\mathcal{#1}]}}
\newcommand{\autom}[1]{\ensuremath{\mathcal{#1}}}
\newcommand{\automSet}[1]{\ensuremath{\mathbb{#1}}}
\newcommand{\automUnion}[1]{\ensuremath{\cup #1}}
\newcommand{\automInters}[1]{\ensuremath{\cap #1}}
\newcommand{\f}{\ensuremath{\mathbf{f}}}
\newcommand{\g}{\ensuremath{\mathbf{g}}}

\newcommand{\locS}{\textsc{loc}\ensuremath{_{[\mathcal{S}]}}}
\newcommand{\mlS}{\textsc{ml}\ensuremath{_{[\mathcal{S}]}}}
\newcommand{\perfS}{\textsc{perf}\ensuremath{_{[\mathcal{S}]}}}
\newcommand{\elocS}{\ensuremath{\exists}\textsc{-loc}\ensuremath{_{[\mathcal{S}]}}}
\newcommand{\emlS}{\ensuremath{\exists}\textsc{-ml}\ensuremath{_{[\mathcal{S}]}}}
\newcommand{\eperfS}{\ensuremath{\exists}\textsc{-perf}\ensuremath{_{[\mathcal{S}]}}}

\newcommand{\locR}{\textsc{loc}\ensuremath{_{[\mathcal{R}]}}}
\newcommand{\mlR}{\textsc{ml}\ensuremath{_{[\mathcal{R}]}}}
\newcommand{\perfR}{\textsc{perf}\ensuremath{_{[\mathcal{R}]}}}
\newcommand{\elocR}{\ensuremath{\exists}\textsc{-loc}\ensuremath{_{[\mathcal{R}]}}}
\newcommand{\emlR}{\ensuremath{\exists}\textsc{-ml}\ensuremath{_{[\mathcal{R}]}}}
\newcommand{\eperfR}{\ensuremath{\exists}\textsc{-perf}\ensuremath{_{[\mathcal{R}]}}}

\newcommand{\elocRB}{\ensuremath{\exists}\textsc{-loc}\ensuremath{_{[\mathcal{R}]}^B}}
\newcommand{\emlRB}{\ensuremath{\exists}\textsc{-ml}\ensuremath{_{[\mathcal{R}]}^B}}
\newcommand{\eperfRB}{\ensuremath{\exists}\textsc{-perf}\ensuremath{_{[\mathcal{R}]}^B}}

\newcommand{\loc}[1]{\textsc{loc}\ensuremath{_{[\scriptsize\texttt{#1}\normalsize]}}}
\newcommand{\ml}[1]{\textsc{ml}\ensuremath{_{[\scriptsize\texttt{#1}\normalsize]}}}
\newcommand{\perf}[1]{\textsc{perf}\ensuremath{_{[\scriptsize\texttt{#1}\normalsize]}}}
\newcommand{\eloc}[1]{\ensuremath{\exists}\textsc{-loc}\ensuremath{_{[\scriptsize\texttt{#1}\normalsize]}}}
\newcommand{\eml}[1]{\ensuremath{\exists}\textsc{-ml}\ensuremath{_{[\scriptsize\texttt{#1}\normalsize]}}}
\newcommand{\eperf}[1]{\ensuremath{\exists}\textsc{-perf}\ensuremath{_{[\scriptsize\texttt{#1}\normalsize]}}}

\newcommand{\elocB}[1]{\ensuremath{\exists}\textsc{-loc}\ensuremath{_{[\scriptsize\texttt{#1}\normalsize]}^B}}
\newcommand{\emlB}[1]{\ensuremath{\exists}\textsc{-ml}\ensuremath{_{[\scriptsize\texttt{#1}\normalsize]}^B}}
\newcommand{\eperfB}[1]{\ensuremath{\exists}\textsc{-perf}\ensuremath{_{[\scriptsize\texttt{#1}\normalsize]}^B}}

\newcommand{\equivalence}[1]{\textsc{equiv}\ensuremath{_{[\scriptsize\texttt{#1}\normalsize]}}}
\newcommand{\equivalenceR}{\textsc{equiv}\ensuremath{_{[\mathcal{R}]}}}
\newcommand{\equivalenceS}{\textsc{equiv}\ensuremath{_{[\mathcal{S}]}}}

\newcommand{\concatUnivR}{\textsc{concat-univ}\ensuremath{_{[\mathcal{R}]}}}

\newcommand{\consistent}[1]{\textsc{cons}\ensuremath{_{[\scriptsize\texttt{#1}\normalsize]}}}
\newcommand{\consistentS}{\textsc{cons}\ensuremath{_{[\mathcal{S}]}}}

\newcommand{\oneUnamb}[1]{\textsc{one-unamb}\ensuremath{_{[\scriptsize\texttt{#1}\normalsize]}}}
\newcommand{\oneUnambR}{\textsc{one-unamb}\ensuremath{_{[\mathcal{R}]}}}

\newcommand{\aFA}{\texttt{aFA}}
\newcommand{\aFAs}{\texttt{aFAs}}

\newcommand{\nFA}{\texttt{nFA}}
\newcommand{\nFAs}{\texttt{nFAs}}

\newcommand{\dFA}{\texttt{dFA}}
\newcommand{\dFAs}{\texttt{dFAs}}

\newcommand{\nRE}{\texttt{nRE}}
\newcommand{\nREs}{\texttt{nREs}}

\newcommand{\dRE}{\texttt{dRE}}
\newcommand{\dREs}{\texttt{dREs}}

\newcommand{\rDTD}{\ensuremath{\mathcal{R}}\textrm{-}\texttt{DTD}}
\newcommand{\rDTDs}{\ensuremath{\mathcal{R}}\textrm{-}\texttt{DTDs}}

\newcommand{\rSDTD}{\ensuremath{\mathcal{R}}\textrm{-}\texttt{SDTD}}
\newcommand{\rSDTDs}{\ensuremath{\mathcal{R}}\textrm{-}\texttt{SDTDs}}

\newcommand{\rEDTD}{\ensuremath{\mathcal{R}}\textrm{-}\texttt{EDTD}}
\newcommand{\rEDTDs}{\ensuremath{\mathcal{R}}\textrm{-}\texttt{EDTDs}}

\newcommand{\xDTD}[1]{\ensuremath{\texttt{#1\textrm{-}DTD}}}
\newcommand{\xDTDs}[1]{\ensuremath{\texttt{#1\textrm{-}DTDs}}}

\newcommand{\xSDTD}[1]{\ensuremath{\texttt{#1\textrm{-}SDTD}}}
\newcommand{\xSDTDs}[1]{\ensuremath{\texttt{#1\textrm{-}SDTDs}}}

\newcommand{\xEDTD}[1]{\ensuremath{\texttt{#1\textrm{-}EDTD}}}
\newcommand{\xEDTDs}[1]{\ensuremath{\texttt{#1\textrm{-}EDTDs}}}

\newcommand{\nUTA}{\texttt{nUTA}}

\newcommand{\dUTA}{\texttt{dUTA}}
\newcommand{\dUTAs}{\texttt{dUTAs}}

\newcommand{\ancstrt}[2]{\ensuremath{\textsf{anc-str}_{#1}(#2)}}
\newcommand{\ancstr}[1]{\ensuremath{\textsf{anc-str}(#1)}}

\newcommand{\labt}[2]{\ensuremath{\textsf{lab}_{#1}(#2)}}
\newcommand{\lab}[1]{\ensuremath{\textsf{lab}(#1)}}

\newcommand{\parentt}[2]{\ensuremath{\textsf{parent}_{#1}(#2)}}

\newcommand{\childstrt}[2]{\ensuremath{\textsf{child-str}_{#1}(#2)}}
\newcommand{\childstr}[1]{\ensuremath{\textsf{child-str}(#1)}}

\newcommand{\childrent}[2]{\ensuremath{\textsf{children}_{#1}(#2)}}
\newcommand{\children}[1]{\ensuremath{\textsf{children}(#1)}}

\newcommand{\treet}[2]{\ensuremath{\textsf{tree}_{#1}(#2)}}

\newcommand{\rootnode}[1]{\ensuremath{\textsf{root}(#1)}}

\newcommand{\class}[1]{\ensuremath{{\scriptsize{\textbf{#1}}}}}
\newcommand{\classc}[1]{\ensuremath{{\scriptsize{\textbf{#1}}}\textrm{-complete}}}
\newcommand{\classh}[1]{\ensuremath{{\scriptsize{\textbf{#1}}}\textrm{-hard}}}

\newcommand{\dual}[1]{\ensuremath{\textsl{dual}(#1)}}

\newcommand{\ttn}{\ensuremath{t_1..t_n}}
\newcommand{\ffn}{\ensuremath{\f_1,\ldots,\f_n}}

\newcommand{\extTn}{\ensuremath{\ext_T(\tau_n)}}
\newcommand{\extTt}{\ensuremath{\ext_T(t_1..t_n)}}
\newcommand{\xextTt}[1]{\ensuremath{\ext_{#1}(t_1..t_n)}}

\newcommand{\typeTn}{\ensuremath{\textsf{type}_T(\tau_n)}}

\newcommand{\Tn}{\ensuremath{T(\tau_n)}}
\newcommand{\Tt}{\ensuremath{T_{[t_1..t_n]}}}
\newcommand{\Tf}{\ensuremath{T_{[\f_1,\ldots,\f_n]}}}

\newcommand{\extwn}{\ensuremath{\ext_w(\tau_n)}}
\newcommand{\wn}{\ensuremath{w(\tau_n)}}

\newcommand{\mysetminus}{\ensuremath{-}}
\newcommand{\mymid}{\ensuremath{-}}

\newcommand{\ONLINE}[1]{#1}

\begin{document}

\begin{frontmatter}

\title{Distributed XML Design}

\author[inr]{S. Abiteboul}
\ead{serge.abiteboul at inria.fr}
\address[inr]{INRIA Saclay -- \^{I}le-de-France \& University Paris Sud, FR}

\author[oxf]{G. Gottlob}
\ead{georg.gottlob@comlab.ox.ac.uk}
\address[oxf]{Oxford University Computing Laboratory \&\\ Oxford-Man Institute of Quantitative Finance, University of Oxford, UK}

\author[cal]{M. Manna}
\ead{manna@mat.unical.it}
\address[cal]{Department of Mathematics, University of Calabria, IT}

\begin{abstract}
A \emph{distributed XML document} is an XML document that spans several machines.
We assume that a distribution design of the document tree is given,
consisting of an \emph{XML kernel-document} $\Tf$ where some leaves are ``docking points'' for external resources providing XML subtrees ($\f_1,\ldots,\f_n$, standing, e.g., for Web services or peers at remote locations). The top-down design problem consists in, given a \emph{type} (a schema document that may vary from a DTD to a tree automaton) for the distributed document, ``propagating'' locally this type into a collection of types, that we call \emph{typing}, while preserving desirable properties.  We also consider the bottom-up design which consists in, given a type for each external resource, exhibiting a global type that is enforced by the local types, again with natural desirable properties. In the article, we lay out the fundamentals of a theory of distributed XML design, analyze problems concerning typing issues in this setting, and study their complexity.
\end{abstract}


\end{frontmatter}


\section{Introduction}

\paragraph{Context and Motivation}

With the Web, information tends to be more and more distributed.  In
particular, the distribution of XML data is essential in many areas
such as e-commerce (shared product catalog), collaborating editing
(e.g., based on WebDAV \citep{HernandezPegah03}), or network
directories \citep{JagadishLakshmananMilo99}. (See also the W3C XML
Fragment Interchange Working Group \cite{GrossoVeillard01}.)  It
becomes often cumbersome to verify the validity, e.g., the type, of
such a hierarchical structure spanning several machines. In this
paper, we consider typing issues raised by the distribution of XML
documents. We introduce ``nice'' properties that the distribution should
obey to facilitate type verification based on locality conditions.
We propose an automata-based study of the problem. Our theoretical
investigation provides a starting point for the distributed
validation of tree documents (verification) and for selecting a
distribution for Web data (design). In general, it provides new
insights in the typing of XML documents.

A {\em distributed XML document} $\Tt$ is given by an \emph{XML kernel-document} $\Tf$, that is stored locally at some site, some of which leaves (the \emph{docking points}) refer to external resources, here denoted by $\f_1,\ldots,\f_n$, that provide the additional XML data $\ttn$ to be attached, respectively, to $T$. For simplicity, each node playing the role of docking point is called a \emph{function-node} and it is labeled with the resource that it refers.

The {\em extension} $\extTt$ of $T$ is the whole XML document obtainable from the distributed document $\Tt$ by replacing the node referring resource $\f_i$ with the forest of XML trees (in left-to-right order) directly connected to the root of $t_i$, for each $i$ in $[1..n]$.

\begin{figure}[htbp] \centering
\fbox{\includegraphics[scale=0.6]{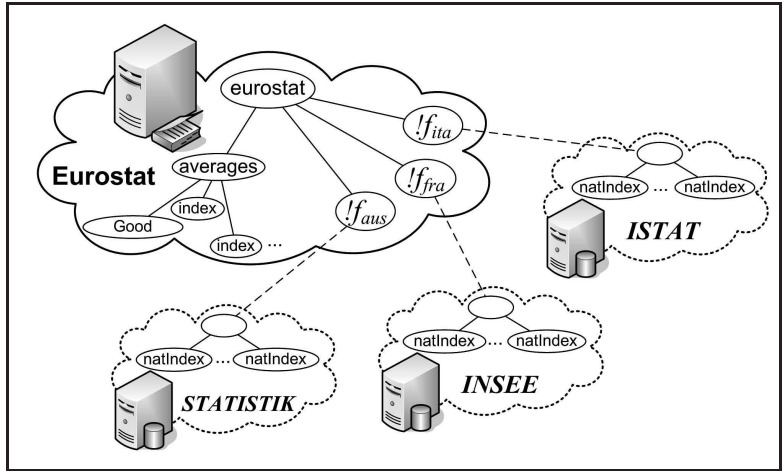}}
\caption{A distributed XML document for the \emph{National Consumer
Price Index} where both its kernel and some of its remote XML (sub)documents have been hightailed.}
\label{figDistribArchit}
\end{figure}

Figure~\ref{figDistribArchit} shows a (drastically simplified) possible distributed XML document for the {\em National Consumer Price Index (NCPI)}~\footnote{See \url{http://epp.eurostat.ec.europa.eu}} maintained by the {\em Eurostat}~\footnote{See \url{http://ec.europa.eu/eurostat}}. This example is detailed further in this section.

Typically, a global designer first chooses a specific language for constraining the documents of interest. The focus in this paper is on ``structural constraints''. Clearly, one could also consider other constraints such as key and referential constraints.
So, say the designer has to specify documents using DTDs.
Then he specifies a kernel document $\Tf$ together with either:
\begin{enumerate}[ \ $\centerdot$]
  \item[] \emph{bottom-up design}: types $\tau_i$ for each $\f_i$;

  \item[] \emph{top-down design}: a global type $\tau$.
\end{enumerate}

In the bottom-up case, we are interested by the global type that
results from each local source enforcing its local type. \emph{Can such
typing be described by specific type languages?}

In the top-down case, we would like the extension of $T$ to satisfy
$\tau$. The issue is ``\emph{Is it possible to enforce it using only local
control?}'' In particular, we would like to break down $\tau$ into
local types $\tau_i$ that could be enforced locally. More precisely,
we would like to provide each $\f_i$ with a typing $\tau_i$ guaranteeing that
(i) if each $\f_i$ verifies its type, then the global type is verified
(soundness), and (ii) the typing $\tau_1..\tau_n$ is not more restrictive
than the global type (completeness). We call such a typing local
typing. We both study (maximal) local typings and an even more
desirable notion, namely ``perfect typings'' (to be defined).

To conclude this introduction, we next detail the Eurosat
example. We then present a formal overview of the paper (which may
be skipped in a first reading.) Finally, we survey related works.

\paragraph{Working Example}
Before mentioning some related works and concluding this section, we
further illustrate these concepts by detailing our {\em Eurostat}
example.

The NCPI is a document containing consumer price data for each EC country. We assume that the national data are maintained in local XML repositories by each country's national statistics bureau (INSEE for France, Statistik for Austria, Istat for Italy, UK Statistics Authority, and so on). Each national data set is under the strict control of its respective statistics bureau. The kernel document $T_0$ is maintained by Eurostat in Luxembourg and has a docking point for each resource $\f_i$ located in a particular country. In addition, $T_0$ contains average data for the entire EU zone. Figure~\ref{figExtT0}
shows a possible extension of $T_0$, where the actual data values are omitted.

\begin{figure}[h] \centering
\fbox{\includegraphics[scale=0.86]{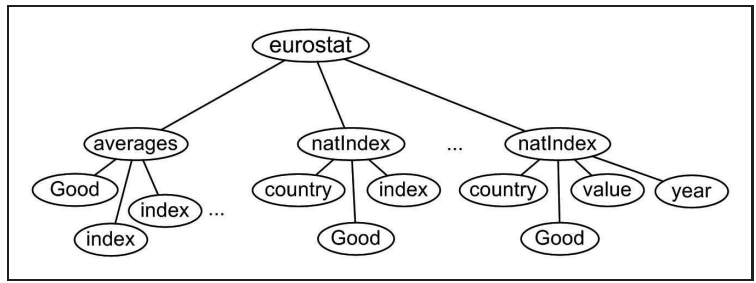}} \caption{The
extension of a possible distributed document having kernel $T_0$, and complying with the whole structure showed in Figure \ref{figDistribArchit}.} \label{figExtT0}
\end{figure}

We first assume that Eurostat specifies the global type
$\tau$ for the distributed NCPI document, where $\tau$ is given
by the DTD document shown in Figure~\ref{DTD_OK_Fig}. (In the following, we adopt a more succinct notation for types where the content model of an element name is either left undefined if it is solely {``\footnotesize \texttt{\#PCDATA}''}, or defined by a rule of the form {``\footnotesize \texttt{index \(\rightarrow\) value, year}''}, otherwise.) Briefly, DTD $\tau$ requires that each possible extension $\xextTt{T_0}$ consists of a subtree containing average data for \emph{Goods} (such as food, energy, education, and so on). Each \emph{Good} item is evaluated in different years by means of an \emph{index}.  Moreover, $\xextTt{T_0}$ may contain a forest of \emph{nationalIndex}, namely indexes associated to goods in precise countries.

\begin{figure}[h]
\hrule \vspace{2mm}
\footnotesize \begin{alltt}
  <!ELEMENT eurostat (averages, nationalIndex\(\sp{*}\))>
    <!ELEMENT averages (Good, index\(\sp{+}\))\(\sp{+}\)>
    <!ELEMENT nationalIndex (country, Good, (index | value, year))>
       <!ELEMENT index (value, year)>
  <!ELEMENT country (#PCDATA)>
  <!ELEMENT Good (#PCDATA)>
  <!ELEMENT value (#PCDATA)>
  <!ELEMENT year (#PCDATA)>
\end{alltt} \vspace{-1.3mm}
\normalsize \hrule \caption{W3C DTD $\tau$}\label{DTD_OK_Fig}
\end{figure}

To comply with different national databases, two different formats are allowed: \emph{(country, Good, index)} or \emph{(country, Good, value, year)}. It is easy to see that the pair $\langle \tau, T_0 \rangle$ allows a local typing (see
Figure~\ref{Perfect_Fig}) that is even perfect (so, can be obtained by the algorithm shown in Section \ref{perfect}), as we will
clarify in the next section.

\begin{figure}[htbp]
\hrule \vspace{2mm} \footnotesize \begin{alltt}
  root\(\sb{i}\) \(\rightarrow\) nationalIndex\(\sp{*}\)
  nationalIndex  \(\rightarrow\) country, Good, (index | value, year)
  index  \(\rightarrow\) value, year
\end{alltt} \vspace{-1.3mm}
\normalsize \hrule \caption{Type $\tau_i$ ($1\leq i \leq n$) in the perfect typing for the top-down design $\langle \tau, T_0 \rangle$}\label{Perfect_Fig}
\end{figure}

Suppose now that a designer defined instead the DTD $\tau'$ shown
in Figure~\ref{DTD_BAD_Fig} as global type. The pair $\langle \tau',T_0\rangle$ would
be a bad design since $\tau'$ imposes to all countries to adopt the
same format for their indexes (\emph{natIndA} or \emph{natIndB}).
But this represents a constraint that cannot be controlled locally.
Indeed, this new design does not admit any local typing.  The nice
locality properties of designs are obvious in such simplistic examples.
However, when dealing with a large number of peers with very different desires and complex documents,
the problem rapidly starts defeating human expertise.
\begin{figure}[h]
\hrule \vspace{2mm} \footnotesize \begin{alltt}
  eurostat \(\rightarrow\) averages, (natIndA\(\sp{*}\) | natIndB\(\sp{*}\))
  averages \(\rightarrow\) (Good, index\(\sp{+}\))\(\sp{+}\)
  natIndA  \(\rightarrow\) country, Good, index
  natIndB  \(\rightarrow\) country, Good, value, year
  index    \(\rightarrow\) value, year
\end{alltt} \vspace{-1.mm}
\normalsize \hrule \caption{Type $\tau'$}\label{DTD_BAD_Fig}
\end{figure}
Consider, for instance, the type $\tau''$ defined in Figure \ref{XSD_Fig_1} and the kernel $T_1 = \emph{eurostat}(\f_1, \ \emph{nationalIndex}(\f_2), \ \f_3)$ containing only three
function calls. Even if this design is as small as $\langle \tau, T_0 \rangle$, it already starts to become hard to manage with no automatic technique. Here, \emph{natIndA} and \emph{natIndB} are different specializations of \emph{nationalIndex} elements  (note that, as detailed in Section \ref{types}, this feature requires schema languages more expressive than DTDs), while all other elements have no specialization.

\begin{figure}[h]
\hrule \vspace{2mm}
\footnotesize \begin{alltt}
  eurostat \(\rightarrow\) averages, (natIndA, natIndB)\(\sp{+}\)
  averages \(\rightarrow\) (Good, index\(\sp{+}\))\(\sp{+}\)
  natIndA  \(\rightarrow\) country, Good, index
  natIndB  \(\rightarrow\) country, Good, value, year
  index    \(\rightarrow\) value, year
\end{alltt} \vspace{-1.3mm}
\normalsize \hrule \caption{Type $\tau''$}\label{XSD_Fig_1}
\end{figure}

In this case, it is not as easy as before to state that
the new design has no perfect typing and exactly the two maximal local typings shown below (only the content models of the roots are specified). This is mainly because the functions in $T_1$ have different depth, but also due to specializations.

\footnotesize \begin{alltt}
  \(\tau\sb{1.1}''\):  root\(\sb{1}\) \(\rightarrow\) averages, (natIndA, natIndB)\(\sp{*}\)
  \(\tau\sb{2.1}''\):  root\(\sb{2}\) \(\rightarrow\) country, Good, index
  \(\tau\sb{3.1}''\):  root\(\sb{3}\) \(\rightarrow\) country, Good, value, year, (natIndA, natIndB)\(\sp{*}\)

  \(\tau\sb{1.2}''\):  root\(\sb{1}\) \(\rightarrow\) averages, (natIndA, natIndB)\(\sp{*}\), natIndA
  \(\tau\sb{2.2}''\):  root\(\sb{2}\) \(\rightarrow\) country, Good, value, year
  \(\tau\sb{3.2}''\):  root\(\sb{3}\) \(\rightarrow\) (natIndA, natIndB)\(\sp{*}\)
\end{alltt} \normalsize

The techniques developed in this paper are meant to support
experts in designing such distributed document schemas.

\paragraph{Overview of Results}

We next precise the formal setting of the paper and its results.
{From} a formal viewpoint, we use Active XML terminology and notation
for describing distributed documents~\citep{Abiteboul08}.

Not surprisingly, our results depend heavily of the nature of the
typing that is considered. For types, we consider abstract versions of
the conventional typing languages~\citep{PapakonstantinouVianu00,
BalminPapakonstantinou04,
MurataLeeMani05,
MartensNeven06}, namely $\rDTDs$ (for \emph{W3C DTDs}),
$\rSDTDs$ (for \emph{W3C XSD}), and $\rEDTDs$ (for regular tree grammars such as \emph{Relax-NG}) where $\mathcal{R}$ (varying among $\nFAs$, $\dFAs$, $\nREs$, and $\dREs$, namely automata and regular expressions both nondeterministic and deterministic) denotes the formalism for specifying content models.

As a main contribution, we initiate a theory of local typing. We
introduce and study three main notions of locality: local typing,
maximal local typing, and perfect typing. For a given XML schema language $\mathcal{S}$, we study the following verification problems:
\begin{enumerate}[ \ $\centerdot$]
  \item Given an $\mathcal{S}$-typing for a top-down $\mathcal{S}$-design, determine whether the former is local, maximal local, or perfect. We call these problems $\locS$, $\mlS$ and $\perfS$, respectively;

  \item Given a top-down $\mathcal{S}$-design, establish whether a local, maximal local, or perfect $\mathcal{S}$-typing does exist (and, of course, find them). We call these problems $\elocS$, $\emlS$, and $\eperfS$, respectively;

  \item Given a bottom-up $\mathcal{S}$-design, establish whether it defines an $\mathcal{S}$-type. The problem is called $\consistentS$.
\end{enumerate}

The analysis carried out in this paper provides tight complexity
bounds for some of these problems. In particular, for bottom-up designs, we prove that $\consistentS$ is:
\begin{enumerate}[ \ $\centerdot$]
  \item decidable in constant time for $\rEDTDs$, for each $\mathcal{\mathcal{R}}$;

  \item $\classc{PSPACE}$ both for $\rDTDs$ and $\rSDTDs$, in general;

  \item $\classh{PSPACE}$ with an $\class{EXPTIME}$ upper bound for $\xDTDs{\dRE}$ and $\xSDTDs{\dRE}$.
\end{enumerate}
For top-down designs, after showing that the problems for trees can be reduced to problems on words, we specialize the analysis to the case of $\mathcal{R} = \nFA$. In particular:
\begin{enumerate}[ \ $\centerdot$]
  \item $\locS$, $\mlS$, $\perfS$, and $\eperfS$ are $\classc{PSPACE}$ when $\mathcal{S}$ stands for $\xDTD{\nFA}$ or $\xSDTD{\nFA}$, and $\locS$ is $\classc{EXPTIME}$ for $\xEDTDs{\nFA}$;

  \item $\elocS$ and $\emlS$ are $\classh{PSPACE}$ with an $\class{EXPSPACE}$ upper bound when $\mathcal{S}$ stands for $\xDTD{\nFA}$ or $\xSDTD{\nFA}$;

  \item the remaining problems are $\classh{EXPTIME}$ with either $\class{coNEXPTIME}$ or 2-$\class{EXPSPACE}$ upper bounds.
\end{enumerate}

\paragraph{Related Work}

Distributed data design has been studied quite in depth, in particular for relational databases \citep{CeriPerniciWiederhold84,OzsuValduriez91}. Some previous works
have considered the design of Web applications
\citep{CeriFraternaliBongio2000}.  They lead to the design of Web
sites. The design there is guided by an underlying process.  It
leads to a more dynamic notion of typing, where part of the content
evolves in time, e.g., creating a cart for a customer. For obvious
reasons, distributed XML has raised a lot of attention recently.
Most works focused on query optimization, e.g.,
\citep{AbiteboulManolescuTaropa06}. The few that consider design
typically assume no ordering or only limited one
\citep{BremerGertz03}. This last work would usefully
complement the techniques presented here. Also, works on relational
database and LDAP~\footnote{Lightweight Directory Access Protocol
(LDAP) is a set of open protocols used to access centrally stored
information over a network.}~design focus on unordered collections.
Even the W3C goes in this direction with a working group on
\emph{XML Fragment Interchange} \citep{GrossoVeillard01}. The goal is
to be able to process (e.g., edit) document fragments independently.
Quoting the W3C Candidate Recommendation: ``It may be desirable to
view or edit one or more [fragments] while having no interest, need,
or ability to view or edit the entire document.'' This is clearly
related to the problem we study here. Finally, the concept of distributed documents, as defined in this paper, is already implemented in Active XML, a declarative framework that harnesses web services for data integration, and is put to work in a peer-to-peer architecture \citep{AbiteboulBonifati03,Abiteboul08}. Moreover, XML documents, XML schemas, and formal languages have been extensively studied and, although all the problems treated in this paper are essentially novel,%
\footnote{Consider that, as highlighted in \cite{MartensNiewerthSchwentick10}, an interesting problem in Formal Language Theory open for more than ten years, named \emph{Language Primality}, is essentially a special case of our problem $\eloc{\dFA}$. The complexity of Primality has been also settled in \cite{MartensNiewerthSchwentick10}.}
the theoretical analysis has got benefit from a number of existing works. Classical results about formal (string and regular) languages come from \citep{StockmeyerMeyer73,Jones75,Seidl90,JiangRavikumar93,RozenbergSalomaaYu97,
HromkovicSeibert97,HagenahMuscholl98,Suciu02} and in particular, those about state complexity of these languages can be found in \cite{Yu01,HolzerKutrib02}, those about one-unambiguous regular languages in \cite{BruggemannWood98,BexGeladeMartensNeven09}, and those about alternating finite state machines in \cite{RozenbergSalomaaYu97,FellahJurgensenYu90}. Finally, regarding XML documents and schemas, our abstract presentation builds on ten years of research in this field.
In particular, it has been strongly influenced by document typings
studied in \cite{PapakonstantinouVianu00,
BruggemannMurata01,MurataLeeMani05,ClarkMurata01,
BalminPapakonstantinou04,ThompsonBeech04,MartensNeven04,
MurataLeeMani05,
MartensNeven06,BrayPaoli06,
MartensNiehren07,ComonDauchetGilleron07,
MartensNevenSchwentick09}  and
results on them presented there.

\paragraph{Structure of the Paper}
This concludes the introduction. The remaining
of the paper is organized as follows. Section \ref{generalSetting} fixes some preliminary notation, formally introduces our notions of type, distributed XML document, and defines the decision problems studied. It also provides an overview of the results. Section  \ref{BottomUpDesign} considers the bottom-up design. Section \ref{TopDownDesign} presents basic results regarding the top-down design. Sections \ref{word} and \ref{perfect} present the main results for the word case. Section \ref{ComplexityForTrees} completes the complexity analysis. Section \ref{conclusion} concludes and mentions possible areas for further research.

\section{General setting}\label{generalSetting}

In this paper, we use a widespread abstraction of XML documents and XML Schemas focusing on document structure \citep{PapakonstantinouVianu00,BalminPapakonstantinou04,MurataLeeMani05,MartensNeven06}, and Active XML terminology and notation for describing distributed documents \citep{AbiteboulBonifati03,Abiteboul08}. In particular, for XML Schemas we will consider families of \emph{tree grammars} (called $\rDTDs$, $\rSDTDs$, and $\rEDTDs$) each of which allows different formalisms for specifying content models ($\mathcal{R}$ may vary among $\nFAs$, $\dFAs$, $\nREs$, and $\dREs$, respectively, nondeterministic automata, deterministic automata, regular expressions, and deterministic regular expressions). This, could be surprising at first sight because the W3C standards impose stricter limitations. However, as we will informally motivate later,  (and has been formally proved in \cite{MartensNiewerthSchwentick10}), some of the problems we define and analyze here, have the same complexity independently of whether we use deterministic or nondeterministic string-automata, or even deterministic regular expressions. Informally, it can be observed that the document distribution often erases the benefits of determinism. For this reason, and because this paper intends to be a first fundamental study of XML distribution, we include in our analysis different possibilities for schema languages, even if for some problem we only analyze the most general case ($\mathcal{R}$ is set to $\nFAs$) in order to delimit its complexity. Moreover, the typing problems we study hint at the possibility that there could be interesting real world applications (all distributed applications that involve the management of distributed data, such as data integration from databases and other data resources exported as Web services, or managing active views on top of data sources) where W3C recommendations are too strict and thus unsuitable in the context of distributed XML documents.

\subsection{Preliminaries}\label{preliminaries}

In this paper, we use also the following notation. We always denote, by $\Sigma$, a (finite) \emph{alphabet}; by $\varepsilon$, the empty string; by $\emptyset$, the empty language; by $\cdot$, the binary relation of concatenation on $\Sigma^*$ and by $\circ$, its extension on $2^{\Sigma^*}$; by $\autom{A}$, an \emph{automaton} for defining a string-language or tree-language over $\Sigma$; by $r$ a \emph{regular expression} over $\Sigma$;
by $\mathcal{R}$, a formalism for defining string languages; by $\mathcal{S}$, a formalism for defining tree languages; by $\tau$, an $\mathcal{R}$-\emph{type} or an $\mathcal{S}$-\emph{type} (a concrete formal structure defining, respectively, a string languages or a tree languages, such as a regular expression or an XML schema document) over $\Sigma$; by $\langOf{\tau}$, the language defined by $\tau$.

\subsubsection{XML Documents}

An \emph{XML document} can be viewed, from a structural point of view, as a \emph{finite ordered, unranked tree} (hereafter just a tree) $t$ with nodes labeled over a given alphabet $\Sigma$. The \emph{topmost node} in $t$ is denoted by $\rootnode{t}$, while for any node $x$ of $t$, we denote by
\begin{enumerate}[ \ $\centerdot$]
  \item $\parentt{t}{x}$ the (unique) parent node of $x$ (if node $x$ is not the root);

  \item $\childrent{t}{x}$ is the sequence of children (possibly empty) of $x$ in left-to-right order;

  \item $\treet{t}{x}$ the subtree of $t$ rooted at $x$;

  \item $\labt{t}{x} \in \Sigma$ the label of $x$;

  \item $\ancstrt{t}{x} \in \Sigma^+$ is the sequence of labels of the path from the root of $t$ to $x$;

  \item $\childstrt{t}{x} \in \Sigma^*$ the labels of the children of $x$ in left-to-right order.
\end{enumerate}
In particular, if $\childstrt{t}{x} = \varepsilon$, then $x$ is called a \emph{leaf node}. The size of $t$, denoted by $\|t\|$, is the number of its nodes. Also in these predicates we may omit the subscript $t$ when it is clear from the context.

\subsubsection{Regular String Languages}

A \emph{nondeterministic finite state machine} ($\nFA$) over $\Sigma$ is a quintuple $\autom{A} = \langle K, \Sigma, \Delta, q_s, F \rangle$ where $K$ are the \emph{states}, $q_s \in K$ is the \emph{initial state}, $F \subseteq K$ are the \emph{final states}, and $\Delta \subseteq K \times (\Sigma \cup \{\varepsilon\}) \times K$ is the \emph{transition relation}. Each triple $(q,\alpha,q') \in \Delta$ is called a \emph{transition} of $\autom{A}$. Sometimes the notation $q' \in \Delta(q,\alpha)$, where $\Delta$ is seen as a function from $K \times (\Sigma \cup \{\varepsilon\})$ to $2^K$, is more convenient. By $\Delta^* \subseteq K \times \Sigma^* \times K$ we denote the \emph{extended transition relation} defined as the reflexive-transitive closure of $\Delta$, in such a way that $(q, w, q') \in \Delta^*$ iff there is a sequence of transitions from $q$ to $q'$ recognizing string $w$. The set of strings $\langOfAutom{A} = \{w \in \Sigma^{*}: \Delta^*(q_s,w) \in F\}$ is the \emph{language} defined by $\autom{A}$. Such machines can be combined in various ways (see \cite{HolzerKutrib02} for a comprehensive analysis). In particular, $\bar{\autom{A}}$ denotes the complement of $\autom{A}$, and defines the language $\Sigma^* - \langOf{\autom{A}}$. Given two $\nFAs$ $\autom{A}_1$ and $\autom{A}_2$, we denote by $\autom{A}_1 \cdot \autom{A}_2$, $\autom{A}_1 \cup \autom{A}_2$, $\autom{A}_1 \cap \autom{A}_2$, and $\autom{A}_1 \mymid \autom{A}_2$ the $\nFA$ defining $\langOf{\autom{A}_1} \circ \langOf{\autom{A}_2}$, $\langOf{\autom{A}_1} \cup \langOf{\autom{A}_2}$, $\langOf{\autom{A}_1} \cap \langOf{\autom{A}_2}$, and $\langOf{\autom{A}_1} - \langOf{\autom{A}_2}$, respectively (operators $\cdot$ and $\circ$ are often omitted). Also, for a set $\automSet{A} = \{\autom{A}_1,\ldots,\autom{A}_m\}$ of $\nFAs$, we often write $\automInters{\automSet{A}}$ (or $\automUnion{\automSet{A}}$) instead of $\autom{A}_1 \cap \ldots \cap \autom{A}_m$ (or $\autom{A}_1 \cup \ldots \cup \autom{A}_m$.)

A \emph{deterministic finite automaton} ($\dFA$) over $\Sigma$ is an $\nFA$ where $\Delta$ is a function from $K \times \Sigma$ to $K$.

\medskip

A (possibly nondeterministic) \emph{regular expression} ($\nRE$ or also \emph{regex}, for short) $r$ over $\Sigma$ is generated by the following abstract syntax:
\[
    r \ \textsf{::=} \ \varepsilon \ | \ \emptyset \ | \ a \ | \ (r \cdot r) \ | \ (r + r) \ | \ r? \ | \ r^+ \ | \ r^*
\]
where $a$ stands generically for the elements of $\Sigma$. When it is clear from the context, we avoid unnecessary brackets or the use of $\cdot$ for concatenation. The language $\langOf{r}$ is defined as usual.

A \emph{deterministic regular expression} ($\dRE$) $r$ is an $\nRE$ with the following restriction. Let us consider the regex $\tilde{r}$ built from $r$ by replacing each symbol $a \in \Sigma$ with $\tilde{a}^i$ where $i$ is the position from left-to-right of $a$ in $r$. By definition, $r$ is a $\dRE$ if there are no strings $w\tilde{a}^i u$ and $w\tilde{a}^j v$ in $\langOf{\tilde{r}}$ such that $i \neq j$. The language $\langOf{r}$ of a $\dRE$ $r$ is called \emph{one-unambiguous} \cite{BruggemannWood98}.

\medskip

A cartesian product of $n$ finite sets is called a ``box'' \citep{Winkler04}. More precisely, fix a positive number $n$. Let $\Sigma$ be an alphabet. A \emph{box} $B$ over $\Sigma$ is any language of the form $\Sigma_1 \ldots \Sigma_n$ where $n$ is its \emph{width}, and $\Sigma_i \subseteq \Sigma$ for each $i$ in $[1..n]$. Clearly, each box is a regular language as it is a finite one.

\subsubsection{Regular Tree Languages}

A \emph{nondeterministic Unranked Tree Automaton} (\nUTA) is a quadruple $\autom{A} = \langle K, \Sigma, \Delta, F \rangle$ where $\Sigma$ is the alphabet, $K$ is a finite set of \emph{states}; $F \subseteq K$ is the set of \emph{final states}; $\Delta$ is a function mapping pairs from $(K \times \Sigma)$ to $\nFAs$ over $K$. A tree $t$ belongs to $\langOfAutom{A}$ if and only if there is a mapping $\mu$ from the nodes of $t$ to $K$ such that (i) $\mu(\rootnode{t}) \in F$, and (ii) for each node $x$ of $t$, either $\varepsilon$ or $\mu(\children{x})$ belongs to $\langOf{\Delta(\mu(x),\lab{x})}$ according to whether $x$ is a leaf-node or not, respectively.

A \emph{bottom-up-deterministic Unranked Tree Automaton} ($\dUTA$) over $\Sigma$ is an $\nUTA$ where $\Delta$ is a function from $(K \times \Sigma)$ to $\dFAs$ over $K$ in such a way that $\langOf{\Delta(q,a)} \cap \langOf{\Delta(q',a)} = \emptyset$ for each $q \neq q'$.

\subsubsection{Known decision problems}
In this section we recall some well known decision problems.

\begin{definition}
$\equivalenceS$ is the following decision problem. Given two $\mathcal{S}$-types, do they define the same language? \qed
\end{definition}
In particular, whenever we consider two $\mathcal{R}$-types instead of $\mathcal{S}$-types, we still denote by $\equivalenceR$ the equivalence problem defined exactly as above.

\begin{definition}
$\oneUnambR$ is the following decision problems. Given a regular language $L$ specified by an $\mathcal{R}$-type, is $L$ one-unambiguous? \qed
\end{definition}

\subsection{Types}\label{types}
As already mentioned, we consider abstractions of the most common XML Schemas by allowing regular languages, specified by possibly different formalisms for defining content models. More formally, let $\mathcal{R}$ be a mechanism for describing regular languages ($\nFAs$, $\dFAs$, $\nREs$, $\dREs$, or even others). We want to define and computationally characterize the problems regarding \emph{Distributed XML design} in a comparative analysis among the three main actual formalisms for specifying XML schema documents: \emph{W3C DTDs}, \emph{W3C XSD} and \emph{Regular Tree Grammars} (like \emph{Relax-NG}). For each of these schema languages, we adopt a class of abstractions that we call $\rDTDs$, $\rSDTDs$, and $\rEDTDs$, respectively, where $\mathcal{R}$ is the particular mechanism  for defining content models. We show that a number of properties do not depend on the choice of $\mathcal{R}$ (or even of $\mathcal{S}$) and for some complexity results we focus our analysis to the case of $\nFAs$. Before that, we summarize in Table \ref{tableOurAbstractions} the relevance of the different tree grammars.

\begin{table}[h]
\caption{Comparison between our abstractions of XML Schemas and existing formalisms.} \label{tableOurAbstractions}\centering \footnotesize
\begin{tabular}{c|c|c}
    \noalign{\smallskip} \noalign{\smallskip}\noalign{\smallskip}

    \vspace{-1mm}

    \ & \ & \ \\

    \vspace{-1mm}

    Schema language & Previously introduced formalism & Our abstraction \\

    \vspace{-1mm}

    \ & \ & \ \\

    \hline\hline

    \ & \ & \ \\

    \vspace{-1mm}

    \emph{W3C DTDs} & \emph{DTD}s and \emph{ltd}s & $\xDTDs{\dRE}$\\

    \vspace{-1mm}

    \ & \ & \ \\

    \hline

    \multirow{4}{*}{\emph{W3C XSD}} & \ & \multirow{4}{*}{$\xSDTDs{\dRE}$} \\

    \vspace{-1mm}

    \ & \emph{Single-Type Tree Grammars} & \ \\

    \vspace{-1mm}

    \ & \ & \ \\

    \vspace{-1mm}

    \ & and \emph{single-type EDTD}s & \ \\

    \vspace{-1mm}

    \ & \ & \ \\

    \hline

    \multirow{4}{*}{\emph{Relax NG}} & \ & \multirow{4}{*}{$\xEDTDs{\nRE}$} \\

    \vspace{-1mm}

    \ & \emph{unranked regular tree languages} & \ \\

    \vspace{-1mm}

    \ & \ & \ \\

    \vspace{-1mm}

    \ & (\emph{specialized ltd}s and \emph{EDTD}s) & \ \\

    \ & \ & \ \\
\end{tabular}
\end{table}
\normalsize

\subsubsection{$\rDTD$ types}
The following definition generalizes definitions considered in the literature such as \emph{ltd}s \cite{PapakonstantinouVianu00,BalminPapakonstantinou04} or \emph{DTD}s \cite{MartensNeven06,MartensNiehren07}, and defined for analyzing the properties of \emph{W3C Document Type Definitions}. As we marry these views, we define the following class of abstractions capturing all of them.

\begin{definition}
An $\rDTD$ is formalized as a triple $\tau = \langle \Sigma, \pi, s \rangle$ where
\begin{enumerate}[ \ $\centerdot$]
  \item $\Sigma$ is an alphabet (the \textbf{element names});

  \item $\pi$ is a function mapping the symbols of $\Sigma$ to $\mathcal{R}$-types still over $\Sigma$;

  \item $s \in \Sigma$ is the \textbf{start symbol}.
\end{enumerate}
A tree $t$, having labels over $\Sigma$ belongs to $\langOf{\tau}$ if and only if:
$\lab{\rootnode{t}} = s$ and $\childstr{x} \in \langOf{\pi(\lab{x})}$, for each node $x$ of $t$.
For a given element name $a$, the regular language $\langOf{\pi(a)}$, associated to $a$, is usually called the \textbf{content model} of $a$. \qed
\end{definition}

Notice that, due to the above definition, $\rDTDs$ with useless element names, or even defining the empty language, do exist. This is because the above definition allows to specify $\rDTDs$ that are, in a sense, ``not reduced'' (think about finite automata with unreachable states). Since it is much more convenient to deal with types that are not effected by these drawbacks, after giving some more definition, we formalize the notion of \textbf{reduced} types.

We introduce the $\dFA$ $\dual{\tau}$. It is the language consisting of the set of paths from the root to a leaf in trees in $\langOf{\tau}$ and it is in some sense the \emph{vertical language} of $\tau$.
\begin{definition}
Let $\tau = \langle \Sigma, \pi, s \rangle$ be an $\rDTD$. We build from $\tau$ the \textbf{dual} $\dFA$ $\dual{\tau} = \langle K, \Sigma, \delta, q_0, F \rangle$ as follows:
\begin{enumerate}[ \ $\centerdot$]
    \item $K = \{q_0\} \cup \{q_a : a \in \Sigma\}$;

    \item $\delta(q_0,s) = q_s$;

    \item for each $a,b \in \Sigma$, $\delta(q_a,b) = q_b$ iff $b$ appears in the alphabet of $\pi(a)$;

    \item $q_a \in F$ iff $\varepsilon \in \langOf{\pi(a)}$. \qed
\end{enumerate}
\end{definition}
Before defining a set of conditions ensuring that all the content models of a given $\rDTD$ $\tau$ are well defined and have no redundancy w.r.t. the language $\langOf{\tau}$, we mark the states of $\dual{\tau}$ (in a bottom-up style) as follows:
\begin{enumerate}
  \item Mark each final state of $\dual{\tau}$ as \emph{bound};
  \item For each non-\emph{bound} state $q_b$, consider the set $\Sigma_b \subseteq \Sigma$ where $\delta(q_b,a)$ is \emph{bound} iff $a \in \Sigma_b$. If $\langOf{\pi(b)} \cap \Sigma_b^+ \neq \emptyset$, then mark also $q_b$ as \emph{bound};
  \item Repeat step 2 until no more states can be marked.
\end{enumerate}

\begin{definition} Let $\tau$ be an $\rDTD$. We say that $\tau$ is \textbf{reduced} iff
\begin{enumerate}[ \ $\centerdot$]
  \item Each state of $\dual{\tau}$ is in at least a path from $q_0$ to a final state in $F$;

  \item Each state of $\dual{\tau}$ is bound;

  \item $\langOf{\dual{\tau}}$ is nonempty. \qed
\end{enumerate}
\end{definition}
We consider only \emph{reduced} $\rDTDs$ where, by the previous definition, it is clear that $\langOf{\tau} \neq \emptyset$. Note that for a given $\rDTD$ $\tau$, it is very easy to build $\dual{\tau}$ and for each ``unprofitable'' state $q_a$
\begin{enumerate}[ \ $\centerdot$]
  \item remove the element name $a$ from $\Sigma$;

  \item remove the rule $\pi(a)$ from $\pi$;

  \item modify the rules containing $a$ in their content models (using standard regular language manipulation) to produce only words not containing $a$ (see \cite{MartensNevenSchwentick09}, for more details.)
\end{enumerate}
Finally, we notice that only the last step of the reducing algorithm may depend on the choice of $\mathcal{R}$. Clearly, an $\rDTD$ and its reduced version describe the same language.

From a theoretical point of view, $\rDTDs$ do not express more than the \emph{local tree languages} \cite{MurataLeeMani05}. In particular, $\xDTDs{\nFA}$, $\xDTDs{\dFA}$ and $\xDTDs{\nRE}$ exactly capture this class of languages while $\xDTDs{\dRE}$ are less expressive \citep{PapakonstantinouVianu00,MurataLeeMani05}.
Nevertheless, the last class of types (using deterministic regular expressions \cite{BruggemannWood98} and that does not capture all the local tree languages) is, from a structural point of view, the closest to \emph{W3C DTDs}.

In this paper, for a given $\rDTD$ where $\mathcal{R}$ stands for $\dFAs$ or $\nFAs$ (for shortness, w.l.o.g., and only in examples) we often specify $\pi$ as a function that maps $\Sigma$-symbols to $\Sigma$-$\nREs$ (recall that any regular expression of size $n$ can be transformed into an equivalent $\varepsilon$-free \nFA\ with $\mathcal{O}(n \log^2 n)$ transitions in time $\mathcal{O}(n \log^2 n)$ \citep{HagenahMuscholl98,HromkovicSeibert97}.)

Finally, an example of $\xDTD{\dRE}$ is $\tau_1 = \langle \{s_1,c\},\pi_1,s_1\rangle$ with $\pi_1(s_1) = c^*$ and $\pi_1(c) = \varepsilon$. In the rest of the paper, we often omit to specify rules such as $\pi_1(c) = \varepsilon$; i.e., if no rule is given for a label, nodes with this label are assumed to be (solely) leaves.

\subsubsection{$\rSDTD$ types}

The following definition generalizes definitions considered in the literature such as \emph{Single-Type Tree Grammar}s \cite{MurataLeeMani05} or \emph{single-type EDTD}s \cite{MartensNeven06}, and defined for analyzing the properties of \emph{W3C XML Schema Definitions}. Also here, we define a class of abstractions capturing all of them.

\begin{definition}
An $\rSDTD$ (standing for single-type extended $\rDTD$) is a quintuple $\tau = \langle \Sigma, \tilde{\Sigma}, \pi, \tilde{s}, \mu \rangle$ where
\begin{enumerate}[ \ $\centerdot$]
  \item $\tilde{\Sigma}$ are the \textbf{specialized element names};

  \item $\langle \tilde{\Sigma}, \pi, \tilde{s} \rangle$ is an $\rDTD$ on $\tilde{\Sigma}$ and denoted by $\textsf{dtd}(\tau)$;

  \item $\mu : \tilde{\Sigma} \rightarrow \Sigma$ is a mapping from all the specialized element names \emph{onto} the set of element names. For each $a \in \Sigma$, we denote by $\tilde{a}^1, \ldots, \tilde{a}^n$ the distinct elements in $\tilde{\Sigma}$ that are mapped to $a$. This set is denoted $\tilde{\Sigma}(a)$;

  \item Let $\dual{\textsf{dtd}(\tau)}$ be $\langle K, \tilde{\Sigma}, \tilde{\delta}, q_0, F \rangle$. Build from this $\dFA$ the possibly $\nFA$ $\dual{\tau} = \langle K, \Sigma, \delta, q_0, F \rangle$ where for each $q,q'\in K$ and $a \in \Sigma$, $\delta(q,a) = q'$ iff there is an element $\tilde{a} \in \tilde{\Sigma}$ such that $\tilde{\delta}(q,\tilde{a}) = q'$. We require that $\dual{\tau}$ is a $\dFA$ (this captures the single-type requirement). Also in this case, $\dual{\tau}$ defines the vertical language of $\tau$.
\end{enumerate}
A tree $t$, labeled over $\Sigma$, is in $\langOf{\tau}$ if and only if there exists a tree $t' \in \langOf{\textsf{dtd}(\tau)}$ such that $t = \mu(t')$ (where $\mu$ is extended to trees). Informally, we call $t'$ a witness for $t$. Finally, an $\rSDTD$ $\tau$ is \textbf{reduced} if and only if $\textsf{dtd}(\tau)$ is. \qed
\end{definition}
As for $\rDTDs$, we consider only \emph{reduced} $\rSDTDs$.

From a theoretical point of view, $\rSDTDs$ are more expressive than $\rDTDs$ but do not capture the \emph{unranked regular tree languages} yet.

\subsubsection{$\rEDTDs$ types}
The following definition generalizes definitions considered in the literature such as \emph{specialized ltd}s \cite{PapakonstantinouVianu00,BalminPapakonstantinou04} or \emph{EDTD}s \cite{MartensNeven06}. Such formalisms (like \emph{Relax-NG}), from a structural perspective, express exactly the \emph{homogeneous unranked regular tree languages} and are as expressive as \emph{unranked tree automata} or \emph{Regular Tree Grammars} \cite{BruggemannMurata01}.

\begin{definition}
An $\rEDTD$ (extended $\rDTD$) $\tau$ is an $\rSDTD$ without the single-type requirement. More formally, the automaton $\dual{\tau}$, built as for $\rSDTD$, may be here an $\nFA$. The language $\langOf{\tau}$ is defined as for $\rSDTDs$. \qed
\end{definition}

\subsection{Distributed Documents}
In the context of distributed architectures (e.g., P2P architectures), distributed documents (or distributed trees), such as AXML documents, are XML documents that may contain embedded \emph{function calls}. In particular, a \emph{distributed XML document} $\Tt$ can be viewed as a collection of (classical) XML documents $\ttn$ brought together by a unique (special) XML document $\Tf$, the \emph{kernel}, some of whose leaf-nodes, called function-nodes, play the role of ``docking points'' for the external resources $\ffn$. The ``activation'' of a node of $T$ having a function as label, say $\f_i$, consists in a call to resource (or function) $\f_i$ the result of which is still an XML document, say $t_i$. When $\f_i$ is invoked, its result is used to extend the kernel $\Tf$. Thus, each docking point connects the peer that holds the kernel and invokes the resource $\f_i$, and the peer that provides the corresponding XML document $t_i$. For simplicity of notation, for labeling a function-node we use exactly the name of the resource it refers. For instance, the tree $T_0 = s(a \ \f_1 \ b(\f_2))$ is a kernel having $s$ as root, and containing two function-nodes referring the external resource $\f_1$ and $\f_2$.

The {\em extension} $\extTt$ of $T$ is the whole XML document (without any function at all) obtained from the distributed document $\Tt$ by replacing each node referring resource $\f_i$ with the forest of XML trees (in left-to-right order) directly connected to the root of $t_i$. This process is called \emph{materialization}. For instance, the extension of kernel $T_0$ would be $s(a \ c(dd) \ b(d(ef)))$ in case of resources $\f_1$ and $\f_2$ provided trees $s_1(c(dd))$ and $s_2(d(ef))$, respectively.

An interesting task is to associate a type $\tau_i$ (e.g., a \emph{W3C XSD} document) to each resource $\f_i$ in such a way that the XML document $t_i$ returned as answer is {\em valid} w.r.t. this type and any materialization process always produces a document $\extTt$ valid w.r.t. a given global type $\tau$ (still specified by the \emph{W3C XSD} syntax). A global type and a kernel document represent the (top-down) \emph{design} of a given distributed architecture. A collection of types associated to the function calls in such a design is called a \emph{typing}. Given a distributed design, we would like to know whether either a precise typing has some properties or a typing with some properties does exist. But also, we could directly start from a kernel $T$ and a typing (bottom-up design) and analyze the properties of the tree language consisting in each possible extension $\extTt$.

More formally, let $\calL$ and $\calF$ be two alphabets, respectively, of \emph{element names} (such as $s$, $a$, $b$, $c$, etc.) and \emph{function symbols} (such as $\f$, $\g$, etc.). A \emph{kernel document} or \emph{kernel tree} $\Tf$ (or also $T(\f_n)$, with $(\f_n)$ denoting a sequence%
\footnote{We denote a finite sequence of objects $(x_1,\ldots,x_n)$ over an index set $I = \{1,\ldots,n\}$ by $(x_n)$ and we often omit the specification of the index set $I$.} of length $n$) is a tree over ($\calL \cup \calF$) where:
\begin{enumerate}[(i)]
  \item the root is an element node (say $s_0$);
  \item the function nodes $\f_1,\ldots,\f_n$ are leaf nodes;
  \item no function symbol occurs more than once.
\end{enumerate}
In particular, for each non-leaf node of $T$, say $x$, the \emph{kernel string} $\childstr{x}$, with $k \geq 0$ functions, is of the form $w_{h} \f_{h+1} w_{h+1} \ldots \f_{h+k} w_{h+k}$ (for some $h$ in $[1..n]$) where $w_i \in \calL^*$ for each $i \in \{h,\ldots,h+k\}$, $\f_i \in \calF$ for each $i \in \{h+1,\ldots,k\}$, and $\f_i \neq \f_j$ for each $i \neq j$.

We next consider its semantics. It is defined by providing a tree for each function-node. In particular, an {\em extension} $\ext$ maps each $i$ in $[1..n]$ to a tree $t_i = \ext(\f_i)$. The {\em extension} $\extTt$ of a kernel $\Tf$ is obtained by replacing each $\f_i$ with the forest of trees (in left-to-right order) directly connected to the root of $t_i$.

A \emph{type} $\tau$ for a kernel tree $T$ is one of an $\rDTD$, $\rSDTD$, or $\rEDTD$. Given an extension $\ttn$, we say that tree $\Tt$ satisfies type $\tau$ if and only if $\extTt$ does. This motivates requirement (iii) to avoid irregularities: For instance, in the kernel $T_1 = s(\f \ \f)$ the children of $s$ in any extension of $T_1$ are of the form $ww$ for some word $w$. But since this is not a regular language, the type of $T_1$ cannot be defined by none of the three adopted formalisms. Although we disallow the same function to appear twice, several functions may share the same type. Also, even if for labeling a function-node we use exactly the name of the resource it refers (for simplicity of notation), this does not prohibit a resource to provide two XML subtrees to be attached to the kernel. In fact, different names (function symbols) can be associated to the same resource still preserving extensions from irregularities.

We introduce {\em typings} to constrain the types of the function calls of a kernel document.  A \emph{typing} for a kernel tree $T(\f_n)$ is a positional mapping from the functions in $(\f_n)$ to a sequence $(\tau_n)$ of types (schema documents).  Now, as we replace each $\f_i$ (in the extensions of $T$) with a forest of XML documents then, for each type $\tau_i$ associated to $\f_i$, we actually use a schema document containing an ``extra'' element name, say $s_i$, being only the label of the root in all the trees in $\langOf{\tau_i}$.

\begin{definition}
We denote by $\extTn$ the tree language consisting of all possible extensions $\extTt$ where $t_i \models \tau_i$ ($t_i$ is valid w.r.t $\tau_i$) for each $i$. \qed
\end{definition}

\begin{definition}\label{DefTn}
We denote by $\Tn$ the $\xEDTD{\nFA}$ (or $\xEDTD{\nRE}$) constructed from $T$ and $(\tau_n)$ in the obvious way such that $\langOf{\Tn} = \extTn$. \qed
\end{definition}

In Section \ref{rEDTDsTyping} we will show precisely how to build $\Tn$ in polynomial time, prove that the construction is semantically correct, and establish that the size of $\Tn$ is purely linear in the size of $T$ and $(\tau_n)$.
Let us illustrate for now the issues with an example. Observe, for instance, that for the tree $T = s_0(a(b) \f_1 a(c))$, no matter which type $\tau_1$ is, there is no $\rDTD$-typing expressing the language $\ext_T(\tau_1)$. Indeed, this is even the case for $T = s_0(a(b) a(c))$ with no function at all. If we consider the tree $T = s_0(a(\f_1) a(\f_2))$, then the typing $\langOf{\tau_1} = \{s_1(b)\}$, $\langOf{\tau_2} = \{s_2(c)\}$ prohibits that $\ext_T(\tau_1,\tau_2)$ is expressible by an $\rDTD$-type because $\langOf{\Tn} = \{s_0(a(b) a(c))\}$ entailing that the content model of $b$ is non-regular; while the typing $\langOf{\tau_1} = \{s_1(b)\}$, $\langOf{\tau_2} = \{s_2(b)\}$ allows that, because $\langOf{\Tn} = \{s_0(a(b) a(b))\}$ entailing that all the content models of $s_0$, $a$ and $b$ are regular languages, $\{aa\}$, $\{b\}$ and $\emptyset$, respectively. Such situations motivated Definition \ref{DefTn}.

Before concluding this section, we adapt the previous definitions to strings in the straightforward way. (We will often use reductions to strings problems in the paper.) Let $w(\f_n) = w_0 \f_1 w_1 \ldots \f_n w_n$ be a \emph{kernel string}. For typing strings, we use $\mathcal{R}$-types where $\mathcal{R} \in \{\nFA, \dFA, \nRE, \dRE\}$.
A \emph{typing} for $w(\f_n)$ is still a positional mapping from the functions in $(\f_n)$ to a sequence $(\tau_n)$ of $\mathcal{R}$-types. By $\extwn$ we still denote the string language consisting of all possible extensions of $w$, and by $\wn$ the $\nFA$ (or $\nRE$) constructed from $w$ and $(\tau_n)$ is such a way that $\langOf{\wn} = \extwn$.

We will use in our proofs a generalization to ``Boxes''. A \emph{kernel box} $B(\f_n) = B_0 \f_1 B_1 \ldots \f_n B_n$ is, here, a finite regular language over $(\calL \cup \calF)$ where $\f_1,\ldots,\f_n$ are as above, and each $B_i$ is a box (of a fixed width) over $\calL$. With $B(\tau_n)$ we denote the $\nFA$ (or $\nRE$) constructed from $B$ and $(\tau_n)$ is such a way that $\langOf{B(\tau_n)} = \textsf{ext}_B(\tau_n)$.

\subsection{The Typing Problems}\label{typing}

In this section, we introduce the notion of distributed XML design, define the design problems that are central to the present work, and give the overview of the complexity results. We consider two different approaches, \emph{bottom-up} and \emph{top-down}, according to whether the distributed design, other than a \emph{kernel tree}, consists of a \emph{typing} or a \emph{target type}, respectively.

\begin{definition}
Let $\mathcal{S}$ be a schema language, and $\Tf$ be a kernel document. We call $\mathcal{S}$-\textbf{design} (or just design) one of the following:
\begin{enumerate}[ \ $\centerdot$]
  \item $D=\langle (\tau_n), \Tf \rangle$ where $(\tau_n)$ is an $\mathcal{S}$-typing. This is \textbf{bottom-up design}.

  \item $D=\langle \tau, \Tf \rangle$ where $\tau$ is a (target) $\mathcal{S}$-type. This is \textbf{top-down design}. \qed
\end{enumerate}
\end{definition}

Intuitively, given a bottom-up design, one would like to find a global type that captures the typing of the global document. On the other hand, given a top-down design, one would like to find types for the local documents that will guarantee the global type.

With the following definition, we start the bottom-up analysis. Notice that the concepts used for bottom-up design will be also useful when we consider top-down design.

\begin{definition}\label{DefConsistS}
Given an $\mathcal{S}$-design $D = \langle (\tau_n), T \rangle$, the $\mathcal{S}$-typing $(\tau_n)$ is \textbf{$\mathcal{S}$-consistent} with $T$ (simply consistent when $\mathcal{S}$ is understood) if there exists an $\mathcal{S}$-type $\tau$ such that $\langOf{\tau} = \extTn$, in other words, if $\extTn$ is definable by some $\mathcal{S}$-type. This problem (deciding whether an $\mathcal{S}$-typing is $\mathcal{S}$-consistent with a kernel tree) is called $\consistentS$. \qed
\end{definition}

We will denote by $\textsf{type}_{T,\mathcal{S}}(\tau_n)$, or $\typeTn$ when $\mathcal{S}$ is understood, the $\mathcal{S}$-type when it exists such that $\langOf{\typeTn} = \extTn$. Notice that if both $\cal S$ and $T$ are fixed, then $\typeTn$ plays the role of a function from  the set of all possible $\mathcal{S}$-typings of length $n$ to a set of certain $\mathcal{S}$-types. According to every possible decision-answer of $\consistentS$ (where $T$ is now fixed), such a function might be always definable, never, or only for some $\mathcal{S}$-typing. Finally, the complexity of deciding $\consistentS$ or computing $\typeTn$ (with an estimation, w.r.t. $\Tn$, of its possible size), may vary considerably due to $\mathcal{S}$.

Table \ref{tableConsistResults} summarizes the complexity results of $\consistentS$.
We vary $\mathcal{S}$
among $\rDTDs$, $\rSDTDs$ and $\rEDTDs$, for various kinds of $\mathcal{R}$. In all cases
but $\dRE$, we get tight bounds. For \texttt{DTDs} and \texttt{SDTDs} with $\dRE$, we provide nonmatching lower and upper bounds. The table also shows the size that
$\typeTn$ may have in the worst case. Again this is given precisely
for all cases but $\dRE$. For \texttt{DTDs} and \texttt{SDTDs} with $\dRE$, we provide
nonmatching bounds.

In the next sections, we systematically analyze the complexity of this
problem by varying $\mathcal{S}$ among $\rDTDs$, $\rSDTDs$ and $\rEDTDs$, and we will
consider $\typeTn$ for each of these schema languages. We next give
an example to illustrate some of the main concepts introduced.

\begin{table}[h]
\caption{Complexity results of $\consistentS$ compared with the  worst-case-optimal size of $\typeTn$ with respect to $m = \| \Tn \|$.} \label{tableConsistResults}\centering \scriptsize
\begin{tabular}{c||ccc}
    \noalign{\smallskip} \noalign{\smallskip}\noalign{\smallskip}

    \vspace{-2mm}

    \ & \ & \ & \ \\

    \vspace{-1mm}

    \ & \ & \ & \ \\

    \vspace{-1mm}

    \ & \texttt{-DTDs} & \texttt{-SDTDs} & \texttt{-EDTDs}\\

    \vspace{-1mm}

    \ & \ & \ & \ \\

    \hline\hline

    \ & \ & \ & \ \\

    \vspace{-1mm}

    \ & \classc{PSPACE} & \classc{PSPACE} & $\class{DTIME}(\mathcal{O}(1))$\\

    $\nFA$ & \ & \ & \ \\

    \ & $\Theta(m)$ & $\Theta(m)$ & $\Theta(m)$ \\

    \vspace{-1mm}

    \ & \ & \ & \ \\

    \hline

    \ & \ & \ & \ \\

    \vspace{-1mm}

    \ & \classc{PSPACE} & \classc{PSPACE} & $\class{DTIME}(\mathcal{O}(1))$\\

    $\nRE$ & \ & \ & \ \\

    \ & $\Theta(m)$ & $\Theta(m)$ & $\Theta(m)$ \\

    \vspace{-1mm}

    \ & \ & \ & \ \\

    \hline

    \ & \ & \ & \ \\

    \vspace{-1mm}

    \ & $\classc{PSPACE}$ & $\classc{PSPACE}$ & $\class{DTIME}(\mathcal{O}(1))$\\

    $\dFA$ & \ & \ & \ \\

    \ & $\Theta(2^m)$ & $\Theta(2^m)$ & $\Theta(m^2)$ \\

    \vspace{-1mm}

    \ & \ & \ & \ \\

    \hline

    \ & \ & \ & \ \\

    \vspace{-1mm}

    \ & \classh{PSPACE} $\leftrightsquigarrow$ $\class{EXPTIME}$ & \classh{PSPACE} $\leftrightsquigarrow$ $\class{EXPTIME}$ & $\class{DTIME}(\mathcal{O}(1))$ \\

    $\dRE$ & \ & \ & \ \\

    \ & $\Omega(2^m)$ $\leftrightsquigarrow$ $\mathcal{O}(2^{2^m})$ & $\Omega(2^m)$ $\leftrightsquigarrow$ $\mathcal{O}(2^{2^m})$ & $\Theta(m)$ \\

    \vspace{-2mm}

    \ & \ & \ & \ \\
\end{tabular}
\end{table}
\normalsize

\begin{example}\label{newTypeExample}
Consider the kernel $T = s_0(a \ \f_1 \ c \ \f_2)$ and the pair $\tau_1 = \langle \{s_1,b\},\pi_1,s_1\rangle$ and $\tau_2 = \langle \{s_2,d\},\pi_2,s_2\rangle$ of $\xDTD{\dRE}$-types, with $\pi_1(s_1) = b^*$ and $\pi_2(s_2) = d^*$. The activation of both $\f_1$ and $\f_2$ may return trees $s_1(bb)$ and $s_2(d)$, respectively. These trees can be plugged into $T$ producing the extension $s_0(abbcd)$. The tree language obtained by considering each possible extension of $T$ is $\ext_T(\tau_1,\tau_2) = \{s_0(ab^ncd^m) : n,m \geq 0\}$. Now, we have:
\[
    \textsf{type}_T(\tau_1, \tau_2) = \langle \{s_0,a,b,c,d\},\pi,s_0\rangle
\]
where $\pi(s_0) = a \ b^* c \ d^*$ and all the other element names other that $s_0$ are leaves. Finally, $(\tau_1,\tau_2)$ is $\xDTD{\dRE}$-consistent with $T$. \qed
\end{example}

We now define the top-down design problems. But before, we introduce some straightforward notation. Let $\tau$ and $\tau'$ be two types. We say that:
\begin{enumerate}[ \ $\centerdot$]
  \item $\tau \equiv \tau'$ (equivalent) iff $[\tau] = [\tau']$

  \item $\tau \leq \tau'$ (smaller or equivalent)  iff $[\tau] \subseteq [\tau']$

  \item $\tau < \tau'$ (smaller) iff $[\tau] \subset [\tau']$
\end{enumerate}
and also that, given two typings $(\tau_n)$ and $(\tau_n')$:
\begin{enumerate}[ \ $\centerdot$]
  \item $(\tau_n) \equiv (\tau_n')$ iff $\tau_i \equiv \tau_i'$ for each $i$

  \item $(\tau_n) \leq (\tau_n')$ iff $\tau_i \leq \tau_i'$ for each $i$

  \item $(\tau_n) < (\tau_n')$ iff $(\tau_n) \leq (\tau_n')$ and $\tau_i < \tau_i'$ for some $i$
\end{enumerate}

\begin{definition}\label{TypingProperties}
Given an $\mathcal{S}$-design $D = \langle \tau, T \rangle$, we say that a typing ($\tau_n$)
is:
\begin{enumerate}[ \ $\centerdot$]
  \item \textbf{sound} if $\extTn \subseteq \langOf{\tau}$;

  \item \textbf{maximal} if it is sound, and there is no other
      sound typing $(\tau_n')$ s.t. $(\tau_n) < (\tau_n')$;

  \item \textbf{complete} if $\extTn \supseteq \langOf{\tau}$;

  \item \textbf{local} if $\extTn = \langOf{\tau}$, namely if it is both sound and complete;

  \item \textbf{perfect} if it is local, and $(\tau_n') \leq (\tau_n)$ for each other
      sound typing $(\tau_n')$;

  \item $D$-\textbf{consistent} if it is an $\mathcal{S}$-typing which is $\mathcal{S}$-consistent as well.  \qed
\end{enumerate}
\end{definition}

\begin{remark}
It should be clear that for a given $\mathcal{S}$-design $D = \langle \tau, T \rangle$ we could have sound typings that are not $D$-consistent. But, note that, it is even possible to have a sound typing where $T(\tau_n)$ does not define a regular tree language. Consider the design $D$ where $T = s_0(\f_1)$ and $\tau = s_0(a^+b^+)$. Clearly, the typing $\langOf{\tau_1} = \{s_1(a^nb^n):n>0\}$ is sound but $\langOf{T(\tau_1)}$ is not regular. Anyway, we prove in Section \ref{additionalProperties} (for strings, but the results generalizes to trees due to our reductions) that if an $\mathcal{S}$-design admits a sound typing $(\tau_n)$, then it also admits a sound $\xEDTD{\nFA}$-typing $(\tau_n')$ such that $(\tau_n) \leq (\tau_n')$.

Also, by definition of maximality, note that for instance, for a given $\xDTD{\dRE}$-design $D$, a $\xDTD{\dRE}$-typing $(\tau_n)$
is not maximal even if there is a sound $\xDTD{\nFA}$-typing $(\tau_n')$ for $D$ such that $(\tau_n) < (\tau_n')$. One could have some objection to such a definition. Anyway, \citet{MartensNiewerthSchwentick10} proved that whenever the illustrated situation happens, then there is also a $\xDTD{\dRE}$-typing $(\tau_n'')$ such that $(\tau_n) < (\tau_n'')$.\qed
\end{remark}

Clearly, local typings present the advantage of allowing a local verification of document consistency (soundness and completeness by definition). Also, no consistent document is ruled out (completeness). Maximal locality guarantees that in some sense, no unnecessary constraints are imposed to the participants. Finally, perfect typings are somehow the ultimate one can expect in terms of not imposing constraints to the participants.  Many designs will not accept a perfect typing. However, there are maximal sound typings which are not local. This is not surprising as there are designs that have at least a sound typing but do not allow any local at all, and clearly, if there is a sound typing, then there must also exist a maximal sound one. We will see examples that separate these different classes further. But before, we make an observation on $D$-consistency and formally state the problems studied in the paper.

Let $\mathcal{S}$ be any schema language among $\rDTDs$, $\rSDTDs$, and $\rEDTDs$, where $\mathcal{R} \in \{\nFA, \dFA, \nRE, \dRE\}$. Whenever we consider a top-down $\mathcal{S}$-design $D = \langle \tau, T\rangle$, we require that a typing $(\tau_n)$ for $D$ has to be $D$-consistent, namely both $T(\tau_n)$ is $\mathcal{S}$-consistent (it has an equivalent $\mathcal{S}$-type) and each $\tau_i$ is an $\mathcal{S}$-type. In order to verify such a condition, we can exploit the techniques that we have developed for bottom-up design. In particular, it is not hard to see that if $(\tau_n)$ is not $\mathcal{S}$-consistent, then it can not be local. Thus, our approach aims at isolating problems concerning locality from those concerning consistency.

\begin{definition}
$\locS$, $\mlS$, $\perfS$ are the following decision problems. Given an $\mathcal{S}$-design $D = \langle \tau, T\rangle$ and a $D$-consistent typing $(\tau_n)$, is $(\tau_n)$ a local, or maximal local, or perfect typing for $D$, respectively? \qed
\end{definition}

\begin{definition}
$\elocS$, $\emlS$, $\eperfS$ are the following decision problems. Given an $\mathcal{S}$-design $\langle \tau, T \rangle$, does there exist a local, or maximal local, or perfect $D$-consistent typing for this design, respectively? \qed
\end{definition}

We similarly define the corresponding {\em word} problems ($\mathcal{S}$ is simply $\mathcal{R}$). We have $\locR$, $\mlR$, $\perfR$, $\elocR$, $\emlR$ and $\eperfR$. Finally, we will use in proofs box versions of the problems, $\elocRB$, $\emlRB$ and $\eperfRB$.

\begin{remark}
In this paper, although we analyze all the three defined schema languages ($\rDTDs$, $\rSDTDs$, and $\rEDTDs$) for top-down designs, after providing reductions from trees to strings, we specialize the analysis to the case of $\mathcal{R} = \nFA$. More tractable problems may be obtained by considering deterministic content models or restricted classes of regular expressions \citep{MartensNeven04,GhelliColazzo07} as made by \citet{MartensNiewerthSchwentick10}. Also, notice that we pay more attention to maximum locality rather than to maximality proper. In fact, for the latter notion, the existence problem is trivial. Moreover, the complexity of the verification problem essentially coincides for both notions.
Nevertheless, one could be interested in a maximal sound typing when, for some reason, the design can not be improved and does not admit any local typing. There could be even cases where a local typing does not exist but, there is a unique maximal sound typing comprising any other possible sound typing, a sort of quasi-perfect typing. For instance, the design $T = s(a \ \f_1)$ and $\tau = s(ab^* + d)$ has such a property. Our techniques can be easily adapted to these cases, too. \qed
\end{remark}

Table \ref{tableResults} gives an overview of the complexity results for the typings problems previously defined. We will see in Section \ref{TopDownDesign} that, for $\rDTDs$ and $\rSDTDs$, each problem on trees is logspace-reducible to a set of problems on strings (thus, it suffices to prove the results in Table \ref{tableResults} for words) and that, for $\rEDTDs$, the problems on trees depend on the problems on boxes in a more complex manner. In particular, row $D$ includes two problems that are actually the same (they only differ if $\mathcal{R} = \dREs$, as shown in \citet{MartensNiewerthSchwentick10}). Each number in brackets refers either to the corresponding statement/proof in the paper (if rounded) or the paper where the particular result has already been proved (if squared).

\begin{table}[h]
\caption{Complexity results in case of top-down design} \label{tableResults}\centering \scriptsize
\begin{tabular}{cc||c|c}
    \noalign{\smallskip} \noalign{\smallskip}\noalign{\smallskip}

    \vspace{-2mm}

    \ & \ & \ & \ \\

    \vspace{-1mm}

    \ & \ & [1] & [2] \\

    \vspace{-1mm}

    \ & \ & \ & \ \\

    \vspace{-1mm}

    \ & \ & $\nFAs$ / $\xDTDs{\nFA}$ / $\xSDTDs{\nFA}$ & $\xEDTDs{\nFA}$\\

    \vspace{-1mm}

    \ & \ & \ & \ \\

    \hline\hline

    \ & \ & \ & \ \\

    \vspace{-1mm}

    [A] & \textsc{loc} & $\classc{PSPACE}$ (\ref{LocNfaPspaceComplete}) & $\classc{EXPTIME}$ (\ref{LocEdtdNfaExpCompl})\\

    \ & \ & \ & \ \\

    \hline

    \ & \ & \ & \ \\

    \vspace{-2mm}

    \ & \ & in $\class{PSPACE}$ \cite{MartensNiewerthSchwentick10} $\vee$ (\ref{nfa-ml-pspace}) & $\classh{EXPTIME}$ (\ref{CorolLocExpHard}) \\

    \vspace{-1mm}

    [B] & \textsc{ml} & \ & \ \\

    \vspace{-1mm}

    \ & \ & $\classh{PSPACE}$ (\ref{check-nfa-hard}) & in 2-$\class{EXPSPACE}$ (\ref{MLedtdNFA2EXPSPACE})\\

    \vspace{-1mm}

    \ & \ & \ & \ \\

    \hline

    \ & \ & \ & \ \\

    \vspace{-2mm}

    \ & \ & \ & $\classh{EXPTIME}$ (\ref{CorolLocExpHard}) \\

    \vspace{-1mm}

    [C] & \textsc{perf} & $\classc{PSPACE}$ (\ref{PerfNfaPspaceComplete}) & \ \\

    \vspace{-1mm}

    \ & \ & \ & in $\class{coNEXPTIME}$ (\ref{PERFedtdNFAcoNEXP})\\

    \vspace{-1mm}

    \ & \ & \ & \ \\

    \hline\hline

    \ & \ & \ & \ \\

    \vspace{-2mm}

    \ & \ & $\classh{PSPACE}$ (\ref{nfa-locHard}) & $\classh{EXPTIME}$ (\ref{ElocEdtdNfaExpHard})\\

    \vspace{-1mm}

    [D] & $\exists$\textsc{-loc}/$\exists$\textsc{-ml} & \ & \ \\

    \vspace{-1mm}

    \ & \ & in $\class{EXPSPACE}$ \cite{MartensNiewerthSchwentick10} $\vee$ (\ref{ThmExpspace}) & in 2-$\class{EXPSPACE}$ (\ref{Thme-loc-ml-2Expspace}) \\

    \vspace{-1mm}

    \ & \ & \ & \ \\

    \hline

    \ & \ & \ & \ \\

    \vspace{-2mm}

    \ & \ & \ & $\classh{EXPTIME}$ (\ref{ElocEdtdNfaExpHard}) \\

    \vspace{-1mm}

    [E] & $\exists$\textsc{-perf} & $\classc{PSPACE}$ (\ref{nfa-pspace-complete}) & \ \\

    \vspace{-1mm}

    \ & \ & \ & in $\class{coNEXPTIME}$ (\ref{EPERFedtdNFAcoNEXP})\\

    \vspace{-2mm}

    \ & \ & \ & \ \\
\end{tabular}
\end{table}
\normalsize

We now present examples that separate the different design properties of typings.

\begin{example}
Let $\tau = \langle\{s,a,b,c\},\pi,s\rangle$ be an $\xDTD{\nRE}$ where $\pi(s) = a^*bc^*$, and $T=s(\f_1\f_2)$ be a kernel tree. It is easy to see that both $s_1(a^*bc^*), s_2(c^*)$ and $s_1(a^*), s_2(a^*bc^*)$ are \emph{local typings} as $a^*bc^* c^* \equiv a^* a^*bc^* \equiv a^*bc^*$. In fact, they are also \emph{maximal local typings}, and so there is no perfect typing for this design.  Observe that, for instance, $s_1(a?), s_2(a^*bc^*)$ is still a \emph{local typing} that, however, is not maximal because it imposes unnecessary constraints to the local sites. If desired, one could leave them more freedom, e.g., type the first function with $a^*$. \qed
\end{example}

\begin{example}
Let $\tau = s(a^* b c^*)$ be a type and $T=s(\f_1 b \f_2)$ be a kernel tree. The typing $s_1(a^*), s_2(c^*)$ is \emph{perfect}. This has to be an excellent typing since there is no alternative maximal local typing. \qed
\end{example}

\begin{example}\label{umlNotPerfect}
Let $\tau = (ab)^*$ be a type and $T=s(\f_1 \f_2)$ be a kernel tree. The typing $s_1((ab)^*), s_2((ab)^*)$ is a unique maximal local but it is not perfect. Consider, in fact, typing $s_1(a), s_2(b)$. It is sound but $(a, b) \leq ((ab)^*, (ab)^*)$ does not hold. Clearly, a perfect typing cannot exist. \qed
\end{example}

\begin{example}
Let $\tau = (ab)^+$ be a type and $T=s(\f_1 \f_2)$ be a kernel tree. There are three maximal local typings:
\[
    s_1((ab)^*), \ s_2((ab)^+) \ \ \ \ \
    s_1((ab)^*a), \ s_2(b(ab)^*) \ \ \ \ \
    s_1((ab)^+), \ s_2((ab)^*)
\]
according to whether either $s_1$, $s_1(a)$, or none of them may belong to each possible $\textsf{ext}(\f_1)$, respectively. \qed
\end{example}

The following theorem completes the comparison of the properties of typing we study. (Note that its converse is not true by Example \ref{umlNotPerfect}.)

\begin{theorem}
Every perfect typing is unique maximal local.
\end{theorem}
\begin{proof}\ONLINE{
Consider a perfect typing $(\tau_n)$ for $T(\f_n)$ and $\tau$. We observe that $(\tau_n)$ is local, by definition. Moreover, by definition, for each other sound (so also local) typing $(\tau_n')$, we have $(\tau_n') \leq (\tau_n)$. The typing, $(\tau_n)$ is maximal because there exists no other sound typing $(\tau_n'')$ such that $(\tau_n) < (\tau_n'')$, and it is unique because there is no another local typing $(\tau_n'')$ such that for some index $i$ and some string $w$, then $w \in \langOf{\tau_i''}$ but $w \notin \langOf{\tau_i}$.}
\end{proof}

\section{Bottom-up design}\label{BottomUpDesign}

In this section, we consider bottom-up design.

\subsection{$\rEDTDs$ typing}\label{rEDTDsTyping}

Let $T(\f_n)$ be a kernel and $(\tau_n)$ be an $\rEDTD$-typing where each $\tau_i = \langle \Sigma_i, \tilde{\Sigma}_i, \pi_i, \tilde{s}_i, \mu_i \rangle$. We next present the construction of $\Tn$, that (to be as general as possible) is an $\xEDTD{\nFA}$. We use the following notations:
\begin{enumerate}
  \item $\Sigma_0$ contains the element names in $T$ (the labels but not the functions);

  \item $\tilde{\Sigma}_0$ contains a specialized element name $\tilde{a}_0^x$, for each $a \in \Sigma_0$ and each node $x$ of $T$ with label $a$.

  \item $s_0$ is the root of $T$;

  \item $s_i$ is the root of trees in $\langOf{\tau_i}$ for each $i$.
\end{enumerate}
We also make without loss of generality the following assumptions:
\begin{enumerate}
  \item $\tilde{\Sigma}_i \cap \tilde{\Sigma}_{j} = \emptyset$, for each $i,j, i \neq j$. (Note that $\Sigma_i \cap \Sigma_{j \neq i}$ may be nonempty.)
\end{enumerate}
Consider the $\xEDTD{\nFA}$ $\Tn = \langle \Sigma, \tilde{\Sigma}, \pi, \tilde{s}_0, \mu \rangle$ defined as follows:
\begin{enumerate}
  \item $\Sigma = \Sigma_0 \cup (\Sigma_1 \mysetminus \{s_1\}) \cup \ldots \cup (\Sigma_n \mysetminus \{s_n\})$;

  \item $\tilde{\Sigma} = \tilde{\Sigma}_0 \cup (\tilde{\Sigma_1} \mysetminus \{\tilde{s}_1\}) \cup \ldots \cup (\tilde{\Sigma_n} \mysetminus \{\tilde{s}_n\})$;

  \item $\tilde{s}_0 = \tilde{s}_0^x$, where $x$ is the root of $T$;

  \item $\mu(\tilde{a}) = a$ for each $\tilde{a} \in \tilde{\Sigma}$;

  \item $\pi(\tilde{a}_0^x) = \nFA(\{\varepsilon\})$ for each leaf-node $x$ of $T$ with label $a \in \Sigma_0$;

  \item $\pi(\tilde{a}_i) = \nFA(\langOf{\pi_i(\tilde{a}_i)})$ for each $\tilde{a}_i \in \tilde{\Sigma}$ with $i$ in $[1..n]$

  \item for each node $x$ of $T$ with label $a$ and children $y_1 \ldots y_p$, we define $\pi(\tilde{a}_0^x) = \nFA(L_1 \ldots L_p)$ where each language $L_k$ is
    \begin{enumerate}[$\centerdot$]
      \item $\{\tilde{b}_0^{y_k}\}$ if $y_k$ has label $b \in \Sigma$;

      \item $\langOf{\pi_i(\tilde{s}_i)}$ if $y_k$ is labeled by $f_i$.
    \end{enumerate}
\end{enumerate}
The previous algorithm clearly runs in polynomial time by scanning the tree $T$ and preforming some easy regular language manipulation. Also, the size of $\Tn$ is linear in the size of the input pair $T$ and $(\tau_n)$. This is clearly true for $\mathcal{R} \in \{\nFA, \dFA\}$ where only a linear number of $\varepsilon$-transitions is required. If $\mathcal{R} \in \{\nRE, \dRE\}$, it is also true because the translation from regular expressions to $\nFAs$ produce at most an $n \log^2n$ blow up but because in these cases we might define $\Tn$ directly as an $\xEDTD{\nRE}$-type of actual linear size.
These considerations immediately yield the following proposition:

\begin{proposition}\label{TnPolyTimeESize}
Given $T(\f_n)$ and $(\tau_n)$, the $\xEDTD{\nFA}$-type $\Tn$ can be constructed in polynomial time, and its size is linear in the input pair.
\end{proposition}

Now we prove that our construction preserves the semantics of $\extTn$.

\begin{theorem}\label{TnIsWellDef}
Given a kernel $T(\f_n)$ and an $\rEDTD$-typing $(\tau_n)$, $\langOf{\Tn} = \extTn$ holds for each possible $\mathcal{R}$.
\end{theorem}
\begin{proof}
By construction of $\Tn$, we assume the specialized element names in each type $\tau_i$ of $(\tau_n)$ to be different (in fact, they could always be renamed appropriately before building $\Tn$). Also, the specialized element names added for giving witnesses to the nodes of $T$ labeled with an element name belong to a fresh set (it is $\tilde{\Sigma}_0$). This means that there is no ``competition'' among all of these witnesses. So we just create new content models that exactly allow all and only the trees being valid for each $\tau_i$ and the non-function nodes that are already in $T$. But this is exactly the semantic definition of $\extTn$.
\end{proof}

\begin{corollary}\label{EDTDconstTime}
All the problems $\consistent{\xEDTD{\nFA}}$, $\consistent{\xEDTD{\dFA}}$, $\consistent{\xEDTD{\nRE}}$, and $\consistent{\xEDTD{\dRE}}$ always have a yes answer.
Thus, they are decidable in constant time.
\end{corollary}
\begin{proof}
For $\consistent{\xEDTD{\nFA}}$, $\consistent{\xEDTD{\dFA}}$, and $\consistent{\xEDTD{\nRE}}$ the decision-answer is always \textbf{\emph{``yes''}} because each content model in $\Tn$ is, respectively, already an $\nFA$, expressible by a $\dFA$, and expressible by an $\nRE$.

\medskip

\noindent For $\consistent{\xEDTD{\dRE}}$ the decision-answer is always \textbf{\emph{``yes''}} as well, but the reason is less obvious. In general, there are regular languages not expressible by $\dREs$. Anyway, in our case, by considering how $\pi$ is built in $\Tn$, we are sure that each content model has an equivalent $\dRE$. In fact, $\pi(\tilde{a}_0^x) = \varepsilon$ (step 4) is already a $\dRE$; $\pi(\tilde{a}_i) = \nFA(\langOf{\pi_i(\tilde{a}_i)})$ (step 5) has an equivalent $\dRE$ because $\pi_i(\tilde{a}_i)$ is already a $\dRE$ by definition; $\pi(\tilde{a}_0^x) = \nFA(L_1 \ldots L_p)$ (step 6) is expressible by a $\dRE$ because each $L_k$ originates itself from a $\dRE$ and does not share any symbol with any $L_{j\neq k}$.
\end{proof}

By Corollary \ref{EDTDconstTime}, we now give a safe and easy construction of $\typeTn$ from $\Tn$ according to the schema language $\mathcal{S}$ used for $(\tau_n)$.
\begin{enumerate}[ \ $\centerdot$]
  \item For $\xEDTDs{\nFA}$, we choose $\typeTn = \Tn$;

  \item For $\xEDTDs{\dFA}$, we modify $\Tn$ by computing the $\varepsilon$-closure for each content model. Notice that this can be done in polynomial time and the size of $\typeTn$ is at most quadratic (and there are cases where this could really happen) in the size of $\Tn$ because each content model originates from $\dFAs$ that do not share any symbol.

  \item For $\xEDTDs{\nRE}$ or $\xEDTDs{\dRE}$, we modify the content models of $\Tn$ as follows: each $\pi(\tilde{a}_i) = \pi_i(\tilde{a}_i)$ and each $\pi(\tilde{a}_0^x) = R_1 \ldots R_p$ where the generic $R_k$ is either $\tilde{b}_0^{y_k}$ or $\pi_i(\tilde{s}_i)$ (compare with the $\Tn$ definition). Also here the size of $\typeTn$ is linear in the size of $\Tn$ due to the previous corollary.
\end{enumerate}

\subsection{$\rSDTDs$ typing}

For $\rSDTDs$ we also use $\Tn$ as defined for $\rEDTDs$ because any $\rSDTD$ can be seen as a special $\rEDTD$ and the algorithm for building $\Tn$ still works with no problem. At this point, it should be clear that $\Tn$ can easily not be an $\rSDTD$ because of our assumptions ($\tilde{\Sigma}_i \cap \tilde{\Sigma}_{j \neq i} = \emptyset$). But, in this case it is also possible that $\Tn$ does not have an equivalent $\rSDTD$. Indeed, $T(\f_n)$ may contain some pattern that already prohibits obtaining an $\rSDTD$ for any possible typing $(\tau_n)$, or it may contain a function-layout that prohibits obtaining an $\rSDTD$ for some $(\tau_n)$. So, we have to discriminate when this is possible or not. Such a problem (deciding whether an $\rEDTD$ has an equivalent $\rSDTD$) is in general (when $\mathcal{R}$ stands for $\nREs$ or $\nFAs$) an $\classc{EXPTIME}$ problem \citep{MartensNeven06}. Nevertheless, we will show that, in our case, it is almost always ``easier'' and in particular that, in general, it depends on the complexity of equivalence between tree-languages specified by $\xSDTDs{\nFA}$, which, in turn depends on the complexity of equivalence between string-languages specified by $\nFAs$. Before giving proofs of that, we illustrate the definition of distributed document using $\xSDTD{\nRE}$-types.
\begin{example}
Let $T = s_0(\f_1 \ a(b \ \f_2) \ c)$ be a kernel tree and $\tau_1, \tau_2$ be two $\xSDTD{\nRE}$-types describing respectively $b \cdot d^+ \cdot {a(b^+)}^*$ and $b^*$. In the $\xSDTDs{\nRE}$ syntax, $\tau_1 = \langle \{s_1, a, b, d \}, \{\tilde{s}_1, \tilde{a}_1, \tilde{b}_1, \tilde{d}_1\}, \pi_1, \tilde{s}_1, \mu_1 \rangle$ and $\tau_2 = \langle \{s_2, b \}$, $\{\tilde{s}_2, \tilde{b}_2\}$, $\pi_2, \tilde{s}_2, \mu_2 \rangle$ two types where:
\begin{enumerate}[ \ $\centerdot$]
  \item $\Sigma_0 = \{s_0, a, b, c\}$; $\tilde{\Sigma}_0 = \{\tilde{s}_0^1, \tilde{a}_0^3, \tilde{b}_0^4, \tilde{c}_0^6,\}$, where $\{1,3,4,6\}$ are the nodes of $T$ with label in $\Sigma_0$ (based on a preorder traversal of $T$;)

  \item $\pi_1(\tilde{s}_1) = \tilde{b}_1 \cdot \tilde{d}_1^+ \cdot \tilde{a}_1^*$; $\pi_1(\tilde{a}_1) = \tilde{b}_1^+$; $\pi_2(\tilde{s}_2) = \tilde{b}_2^*$; $\pi_1(\tilde{b}_1) = \pi_1(\tilde{d}_1) = \pi_2(\tilde{b}_2) = \varepsilon$;

  \item $\mu_1$ and $\mu_2$ are clear;
\end{enumerate}
For instance, the activation of both $\f_1$ and $\f_2$ may return trees $s_1(bda(bbb))$ and $s_2(bb)$, respectively. In general, the resulting type is $s_0(b \cdot d^+ \cdot {a(b^+)}^+ \cdot c)$. It can be described by an $\xSDTDs{\nRE}$. Thus, $(\tau_1,\tau_2)$ is an $\xSDTDs{\nRE}$-typing consistent with $T$.\qed
\end{example}

Now we need to introduce some definitions and mention previous results.

\begin{lemma}\label{singTypeLab1}
Let $\tau = \langle \Sigma, \tilde{\Sigma}, \pi, \tilde{s}, \mu \rangle$ be an $\rSDTD$. For each $\tilde{a} \in \tilde{\Sigma}$, also $\tau(\tilde{a}) = \langle \Sigma, \tilde{\Sigma}, \pi, \tilde{a}, \mu \rangle$ is.
\end{lemma}
\begin{proof}
By definition of $\rSDTD$ (in the worst case, if $\tau$ is reduced, then $\tau(\tilde{a})$ may be not.)
\end{proof}

\begin{definition}[\cite{MartensNeven06}]
A tree language $L$ is \textbf{closed under ancestor-guarded subtree exchange} if the following holds. For each $t_1,t_2 \in L$, and for each $x_1,x_2$ in $t_1,t_2$, respectively, with $\ancstrt{t_1}{x_1} = \ancstrt{t_2}{x_2}$, the trees obtained by exchanging $\treet{t_1}{x_1}$ and $\treet{t_2}{x_2}$ are still in $L$. \qed
\end{definition}

\begin{lemma}[\cite{MartensNeven06}]\label{singTypeLab3}
A tree language is definable by a $\rSDTD$ iff it is ``closed under ancestor-guarded subtree exchange'' and each content model is defined by an $\mathcal{R}$-type.
\end{lemma}

\begin{remark}
Intuitively, this means that the witness associated by an $\rSDTD$-type $\tau$ to a node $x$ of a tree $t \in \langOf{\tau}$ only depends on the string $\ancstrt{t}{x}$. This is consistent with the definition of $\dual{\tau}$ as a $\dFA$. In fact, the (unique) sequence of states that $\dual{\tau}$ scans for recognizing $\ancstrt{t}{x}$ (except the initial one) exactly gives the unique witness to each node of $t$ in the path from the root to $x$.\qed
\end{remark}

\begin{proposition}\label{oneUnambPropos}\emph{\cite{BruggemannWood98}}
\begin{enumerate}
  \item There is an equivalent $\dRE$ for each one-unambiguous regular language;

  \item Let $\autom{A}$ be a minimum $\dFA$. There is an algorithm, that runs in time quadratic in the size of $\autom{A}$, deciding whether $\langOfAutom{A}$ is one-unambiguous;

  \item There are one-unambiguous regular languages where the smallest equivalent $\dRE$ is exponential in the size of the minimum equivalent $\dFA$. (This is worst-case optimal;)

  \item There are one-unambiguous regular languages where some $\nRE$ is exponentially more succinct than the smallest equivalent $\dRE$. In particular, the language $\{(a+b)^m b (a+b)^n: m \leq n, \ n > 0\}$ has such a property;

  \item The set of all one-unambiguous regular languages is not closed under concatenation.
\end{enumerate}
\end{proposition}

\begin{corollary}\label{oneUnambCoroll}
\
\begin{enumerate}[$(1)$]
  \item Problem $\oneUnamb{\nRE}$ is in $\class{EXPTIME}$.

  \item For each $\nRE$ defining a one-unambiguous grammar, there exists an equivalent $\dRE$ which is, at most, doubly exponential in size. (An exact bound is still open.)

  \item There are pairs of $\dREs$ the concatenation of which, by a string separator, defines a one-unambiguous language such that the smallest equivalent $\dRE$ has an exponential size.
\end{enumerate}
\end{corollary}
\begin{proof}
$(1)$: Let $r$ be an $\nRE$. Build, in polynomial time from $r$, an equivalent $\nFA$ $\autom{A}$. Run the quadratic-time algorithm described in \cite{BruggemannWood98} on the minimum $\dFA$ (at most exponentially larger) equivalent to $\autom{A}$.

\medskip

\noindent $(2)$: By Proposition \ref{oneUnambPropos}, the $\dRE$ $r'$ that we construct from the $\dFA$ $\autom{A}$, introduced in (1), has at most size exponential in the size of $\autom{A}$. Thus, the size of $r'$ is at most doubly exponential in the size of $r$.

\medskip

\noindent $(3)$: Let $r_1 = (a+b)^m$ and $r_2 = (a+b)^n$ be to $\nREs$, with $m \leq n$. By definition, it is clear that they are also both $\dREs$ linear in $n$. Consider the new $\nRE$ $r = r_1 b r_2$. By the previous proposition, $r$ defines a one-unambiguous language but its smallest equivalent $\dRE$ is exponentially larger.
\end{proof}

\begin{lemma}[\cite{BexGeladeMartensNeven09}]\label{REvsDRELemma}
Problem $\oneUnamb{\nRE}$ is $\classh{PSPACE}$.
\end{lemma}

\begin{definition}
$\concatUnivR$ is the following decision problem. Let $\Sigma$ be an alphabet. Given two $\mathcal{R}$-types $\tau_1$ and $\tau_2$ over $\Sigma$, is $\langOf{\tau_1} \circ \langOf{\tau_2} = \Sigma^*$. \qed
\end{definition}

\begin{lemma}[\cite{StockmeyerMeyer73,JiangRavikumar93,MartensNiewerthSchwentick10}]\label{concatUnivLemma}
$\concatUnivR$ is $\classc{PSPACE}$ for each $\mathcal{R} \in \{\nFA, \nRE, \dFA, \dRE\}$.
\end{lemma}

After introducing some necessary definitions and results, we are ready for proving the following theorem. It is fundamental for pinpointing the complexity of $\consistent{\rSDTDs}$, for giving size-bounds about $\typeTn$ and the guidelines for constructing it.

\begin{theorem}\label{consSDTDReduct}
Let $T(\f_n)$ be a kernel and $(\tau_n)$ be an $\rSDTD$-typing.
\begin{enumerate}
  \item If $\mathcal{R} \in \{\nFA, \nRE\}$ (nondeterministic and closed under concatenation), then $\consistent{\rSDTD}$ is polynomial-time Turing reducible to $\equivalence{\rSDTD}$ and $\typeTn$ is not larger than $\Tn$;

  \item If $\mathcal{R} = \dFA$ (deterministic and closed under concatenation), then problem $\consistent{\rSDTD}$ is polynomial-time Turing reducible to $\equivalence{\xSDTD{\nFA}}$ and $\typeTn$ has unavoidably a single-exponential blow up w.r.t. $\Tn$ in the worst case;

  \item If $\mathcal{R} = \dRE$ (not closed under concatenation), then $\consistent{\rSDTD}$ is polynomial-space Turing reducible to $\oneUnamb{\nRE}$. There are cases where the size of $\typeTn$ is, at least, exponential in the size of $\Tn$. A doubly exponential size is sufficient in the worst case. (The exact bound is still open.)
\end{enumerate}
\end{theorem}
\begin{proof}\ONLINE{
First of all we observe that, by construction, the only content models of $\Tn$ that might not satisfy the single-type requirement are those related to the witnesses of the non-leaf nodes of $T$. More formally, let $x$ be any non-leaf node of $T$ the label of which is denoted by $a$, the content model $\pi(\tilde{a}_0^x)$ of its (unique) witness $\tilde{a}_0^x$ is the only one(s) that may contain some conflict. All other content models either refer to leaves ($\varepsilon$ is single-type) or come from some $\tau_i$ (that is already single-type.)

\medskip

\noindent \textbf{Case 1.} \textsc{Proof Idea}: Consider $(\tau_n)$ being simply an $\rEDTD$. Build $\Tn$ and (from it) $\typeTn = \langle \Sigma, \tilde{\Sigma}, \pi, \tilde{s}_0, \mu \rangle$ (both in polynomial time as described in Section \ref{rEDTDsTyping}) and try to ``simplify'' the latter (in a bottom-up way starting from the nodes of $T$ having only leaves as children and going on to the root) for satisfying the single-type requirement. If the algorithm does not fail during its run ($\consistent{\rSDTD}$ admits a \textbf{``yes''} answer), then the resulting $\typeTn$ is now an $\rSDTD$.
During the proof we only make use of the ``ancestor-guarded subtree exchange'' property, and so, by Lemma \ref{singTypeLab3}, we can conclude that if we cannot simplify $\typeTn$, then it does not have an equivalent $\rSDTD$. Moreover, due to the simplification process that does not change the structure of $\Tn$ but only \emph{merges} some specialized element names, then the resulting type is at most as large as the original one. Finally, to check the subtree exchange property we only use equivalence between $\rSDTDs$ and the number of performed steps is clearly polynomial in the size of $\Tn$ witch, by Proposition \ref{TnPolyTimeESize}, is polynomial in $T$ and $(\tau_n)$.

More formally, for each node $x$ of $T$ having only leaves as children and of course an element name as label, say $a$, and for each pair of children $y \neq z$ of $x$, do:
\begin{enumerate}
  \item If both $y$ and $z$ are not function nodes and have the same label, say $b$. As $\pi(\tilde{b}_0^y) = \pi(\tilde{b}_0^z) = \varepsilon$ (by definition) we can consider hereafter, by Lemma \ref{singTypeLab3}, $\tilde{b}_0^y$ and $\tilde{b}_0^z$ the same element.

  \item If only one of the two, say $y$, has an element name as label, say $b$, while $z$ has a function as label, say $\f_i$, and $\pi_i(\tilde{s}_i)$ contains in its specification an element $\tilde{b}_i$ (at most one, as $\tau_i$ is already an $\rSDTD$), by Lemma \ref{singTypeLab3}, if $\langOf{\textsf{type}_T(\tau_n, \tilde{b}_i)} = \{b()\}$, then we can consider hereafter, $\tilde{b}_i$ and $\tilde{b}_0^y$ the same element; otherwise we can conclude that $\typeTn$ does not have an equivalent $\rSDTD$.

  \item Finally, if both $y$ and $z$ are function nodes having label $\f_i$ and $\f_j$, respectively, for each element name in $\Sigma$, say $b$, if both $\pi_i(\tilde{s}_i)$ and $\pi_j(\tilde{s}_j)$ contain in their specifications the elements $\tilde{b}_i$ and $\tilde{b}_j$ (at most one for each of them, as $\tau_i$ and $\tau_j$ are already $\rSDTD$), by Lemma \ref{singTypeLab3}, if $\langOf{\textsf{type}_T(\tau_n, \tilde{b}_i)} = \langOf{\textsf{type}_T(\tau_n, \tilde{b}_j)}$ (by construction, this can be done by deciding whether the two $\rSDTDs$ $\tau_i(\tilde{b}_i)$  and $\tau_j(\tilde{b}_j)$ are equivalent), then we can consider hereafter, $\tilde{b}_i$ and $\tilde{b}_j$ the same element; otherwise if for some $\tilde{b}_i$ and $\tilde{b}_j$ this is not true, we can conclude that $\typeTn$ does not have an equivalent $\rSDTD$;
\end{enumerate}
If the corresponding condition is satisfied for each $y$ and $z$, then we can conclude that $\pi(\tilde{a}_0^x)$ complies with the single-type requirement, that $\textsf{type}_T(\tau_n, \tilde{a}_0^x)$ has an equivalent $\rSDTD$ (obtained by applying the previous steps), and that it can be used for checking equivalences when we consider the parent of $x$, its children and (some modifications of) the three previous steps (see further.)

If $\textsf{type}_T(\tau_n, \tilde{a}_0^x)$ has an equivalent $\rSDTD$ for each considered node $x$, then the next iteration considers each node $x'$ of $T$ having only leaves as children or a node already analyzed. We perform Step 3 exactly as above, while Step 1 or Step 2 with the following trivial changes. Let $y$ be, now, a non-leaf node (instead of a leaf one):
\begin{enumerate}[ \ $\centerdot$]
  \item[$1'.$] If both $y$ and $z$ are not function nodes and have the same label, say $b$. By Lemma \ref{singTypeLab3}, if $\langOf{\textsf{type}_T(\tau_n, \tilde{b}_0^y)} = \{b()\}$ we can consider hereafter, $\tilde{b}_0^y$ and $\tilde{b}_0^z$ the same element; otherwise we can conclude that $\typeTn$ does not have an equivalent $\rSDTD$;

  \item[$2'.$] If only one of the two, say $y$, has an element name as label, say $b$, while $z$ has a function as label, say $\f_i$, and $\pi_i(\tilde{s}_i)$ contains in its specification an element $\tilde{b}_i$ (at most one, as $\tau_i$ is already an $\rSDTD$), by Lemma \ref{singTypeLab3}, if $\langOf{\textsf{type}_T(\tau_n, \tilde{b}_i)} = \langOf{\textsf{type}_T(\tau_n, \tilde{b}_0^y)}$, then we can consider hereafter, $\tilde{b}_i$ and $\tilde{b}_0^y$ the same element; otherwise we can conclude that $\typeTn$ does not have an equivalent $\rSDTD$.
\end{enumerate}
Finally, if we reach the root of $T$ and after checking equivalences on its children we can conclude that $\pi(\tilde{s_0})$ complies with the single-type requirement, then $\textsf{type}_T(\tau_n, \tilde{s_0}) = \typeTn$ is now (after merging the prescribed specialized element names) an $\rSDTD$.

\medskip

\noindent \textbf{Case 2.} If $\mathcal{R} = \dFA$ then, when we merge some specialized element names in the same content model, we can obtain an $\nFA$. So, we can still invoke the $\equivalence{\xSDTD{\nFA}}$ problem but the size of $\typeTn$ may be exponential as we want it to be a $\xSDTD{\dFA}$, and there are cases for which this may happen already by concatenating two $\dFAs$ \cite{Yu01}.
Given that, the blowup cannot be larger than single-exponential, this bound is optimal.

\medskip

\noindent \textbf{Case 3.} If $\mathcal{R} = \dRE$ then, when we merge some specialized element names in the same content model, we can obtain (due to the concatenation and by Proposition \ref{oneUnambPropos}) an $\nRE$ that may not be expressible by a $\dRE$. We can still invoke the $\equivalence{\xSDTD{\nRE}}$ problem (as necessary condition) but
we also have to invoke the $\oneUnamb{\nRE}$ problem (at least as hard as the first one). Notice that this new check does not compromise the soundness of the algorithm. In fact, for each possible $\xSDTD{\dRE}$ (if any) equivalent to $\typeTn$ the unique witness that can be assigned to $x$, due to Lemma \ref{singTypeLab3}, must define the same language as $\pi(\tilde{a}_0^x)$ by applying $\mu$ to them. Finally, if both the two decision problems answer yes, then we can consider a new iteration of the previous algorithm. In case that each content model has an equivalent $\dRE$ specification and we reach the root of $T$, we can conclude that $\typeTn$ is now a $\xSDTD{\dRE}$. By Proposition \ref{oneUnambPropos}, there are cases where $\typeTn$ may require, at least, single-exponential size. By Corollary \ref{oneUnambCoroll}, a doubly exponential size is sufficient in the worst case.
}
\end{proof}

We now have the following result:

\begin{corollary}\label{SDTDpspacePtime}
\
\begin{enumerate}[$(1)$]
  \item Problems $\consistent{\xSDTD{\nRE}}$ and $\consistent{\xSDTD{\nFA}}$ are $\classc{PSPACE}$;

  \item Problem $\consistent{\xSDTD{\dFA}}$ is $\classc{PSPACE}$;

  \item Problem $\consistent{\xSDTD{\dRE}}$ is both $\classh{PSPACE}$ and in $\class{EXPTIME}$;
\end{enumerate}
\end{corollary}
\begin{proof}
\textbf{Membership}. For (1) and (2) consider that both \textsc{equiv}$_{[\scriptsize\xSDTD{\nRE} \normalsize]}$ and \textsc{equiv}$_{[\scriptsize\xSDTD{\nFA} \normalsize]}$ are feasible in $\class{PSPACE}$ \citep{MartensNeven06}. While for (3) we also consider that $\oneUnamb{\nRE}$ is doable in $\class{EXPTIME}$, by Corollary \ref{oneUnambCoroll}.

\medskip

\noindent \textbf{Hardness}. For (1) we know that both \textsc{equiv}$_{[\scriptsize\xSDTD{\nRE} \normalsize]}$ and \textsc{equiv}$_{[\scriptsize\xSDTD{\nFA} \normalsize]}$ are also $\classh{PSPACE}$ \citep{MartensNeven06}.

For (2) and (3) we directly consider a reduction from $\concatUnivR$ ($\classh{PSPACE}$, by Lemma \ref{concatUnivLemma}) to problem $\consistent{\rSDTD}$ ($\mathcal{R} \in \{\dFA, \dRE\}$). In particular, let $A_1$, $A_2$ be two $\mathcal{R}$-types, we consider the consistency problem for the kernel tree $T = s(a(\f_1 \f_2) \ a(\f_3))$ and the $\rSDTDs$ typing $(\tau_1, \tau_2, \tau_3)$ where the trees in $\tau_1,\tau_2$ have only one level other than the root, $\pi_1(\tilde{s}_1) = A_1$, $\pi_2(\tilde{s}_2) = A_2$, and $\langOf{\pi_3(\tilde{s}_3)} = \Sigma^*$. It is easy to see that $(\tau_1, \tau_2, \tau_3)$ is consistent with $T$ if and only if $\langOf{A_1} \circ \langOf{A_2} = \Sigma^*$.
\end{proof}

We conclude this section with a remark.

\begin{remark}
The exponential blow-up affecting $\typeTn$ may suggest that there are cases for which it may be better to store an XML document in a distributed manner keeping each part valid w.r.t. its local (and small) type $\tau_i$ rather than validate the whole document w.r.t. a very large type.\qed
\end{remark}

\subsection{$\rDTDs$ typing}

Even for $\rDTDs$ we use $\Tn$ as defined for $\rEDTDs$. But here the algorithm we introduced for $\rEDTDs$ does not work any more because an $\rDTD$-typing is structurally different from an $\rSDTD$ or an $\rEDTD$.

Let $T$ be a kernel and $(\tau_n)$ be an $\rDTD$-typing. Before building $\Tn$ we construct, from $(\tau_n)$, an equivalent $\rSDTD$-typing $(\tau_n')$ as follows. Let $\tau_i = \langle \Sigma_i, \pi_i, s_i \rangle$ be the $i^{th}$ type in $(\tau_n)$.
Consider the $\rSDTD$-type $\tau_i = \langle \Sigma_i, \tilde{\Sigma}_i, \pi_i', \tilde{s}_i, \mu_i \rangle$ defined as follows:
\begin{enumerate}[ \ $\centerdot$]
  \item $\tilde{a} \in \tilde{\Sigma}_i$ iff $a \in \Sigma_i$;
  \item $\mu_i$ is a bijection between $\tilde{\Sigma}_i$ and $\Sigma_i$;
  \item $\pi_i'(\tilde{a})$ = $\mu^{-1}(\pi_i(a))$.
\end{enumerate}

The two types are trivially equivalent. So, we can build the new $\xEDTD{\nFA}$-type (or $\xEDTD{\nRE}$-type), representing $\extTn$, by using $(\tau_n')$. But since the overhead of constructing $(\tau_n')$ is completely negligible, we still denote it by $\Tn$ instead of $T(\tau_n')$.

Also in this case we would like to decide whether $\Tn$ has an equivalent $\rDTD$-type or not, and even the general problem (when $\mathcal{R}$ stands for $\nREs$ or $\nFAs$) of deciding whether an $\rEDTD$ has an equivalent $\rDTD$ is $\classc{EXPTIME}$ \citep{MartensNeven06}. As for $\rSDTD$, we will show that in our settings we can do better.

\begin{definition}[\cite{PapakonstantinouVianu00}]
A tree language $L$ is \textbf{closed under subtree substitution} if the following holds. Whenever for two trees $t_1,t_2 \in L$ with nodes $x_1$ and $x_2$, respectively, $\labt{t_1}{x_1} = \labt{t_2}{x_2}$, then the trees obtained, from $t_1$ and $t_2$, by exchanging $\treet{t_1}{x_1}$ and $\treet{t_2}{x_2}$ are still still in $L$. \qed
\end{definition}

\begin{lemma}[\cite{PapakonstantinouVianu00}]\label{cloUndSubSubLem}
A tree language is definable by an $\rDTD$ iff it is ``closed under subtree substitution'' and each content model is defined by an $\mathcal{R}$-type.
\end{lemma}

The following theorem (with the related corollary) concludes the set of results for the bottom-up design problem, and gives the last guidelines for constructing $\typeTn$ or evaluating its size.

\begin{theorem}\label{consDTDReduct}
Let $T(\f_n)$ be a kernel and $(\tau_n)$ be an $\rDTD$-typing.
\begin{enumerate}
  \item If $\mathcal{R} \in \{\nFA, \nRE\}$, then $\consistent{\rDTD}$ is polynomial-time Turing reducible to $\equivalence{\rSDTD}$ and $\typeTn$ is linear in $\Tn$;

  \item If $\mathcal{R} = \dFA$, then $\consistent{\rDTD}$ is polynomial-time Turing reducible to $\equivalence{\xSDTD{\nFA}}$ and $\typeTn$ has unavoidably a single-exponential blow up w.r.t. $\Tn$ in the worst case;

  \item If $\mathcal{R} = \dRE$, then $\consistent{\rDTD}$ is polynomial-space Turing reducible to $\oneUnamb{\nRE}$ and there are case where $\typeTn$ is, at least, exponentially larger than $\Tn$. A doubly exponential size is sufficient in the worst case. (The exact bound is still open.)
\end{enumerate}
\end{theorem}
\begin{proof}\ONLINE{
Build the $\rSDTD$-typing $(\tau_n')$ from $(\tau_n)$ as said before.
Perform, from $T$ and $(\tau_n')$, the decision-algorithm defined in the proof of Theorem \ref{consSDTDReduct} by enforcing, due to Lemma \ref{cloUndSubSubLem}, the additional constraint at the end of each macro-step when we assert that $\textsf{type}_T(\tau_n, \tilde{a}_0^x)$ has an equivalent $\rSDTD$:
\begin{enumerate}[ \ $\centerdot$]
  \item $\langOf{\mu(\pi(\tilde{a}'))} = \langOf{\mu(\pi(\tilde{a}''))}$ for each $\tilde{a}', \tilde{a}''$ already considered;
\end{enumerate}
Finally, notice that the (polynomial number) additional steps are special cases of calls to $\equivalence{\rSDTD}$, and that
the same observations made for $\rSDTDs$ hold for $\typeTn$ as well.}
\end{proof}

\begin{corollary}\label{DTDpspacePtime}
We have the following results:
\begin{enumerate}[$(1)$]
  \item Problems $\consistent{\xDTD{\nRE}}$ and $\consistent{\xDTD{\nFA}}$ are $\classc{PSPACE}$;

  \item Problem $\consistent{\xDTD{\dFA}}$ is $\classc{PSPACE}$;

  \item Problem $\consistent{\xDTD{\dRE}}$ is both $\classh{PSPACE}$ and in $\class{EXPTIME}$;
\end{enumerate}
\end{corollary}
\begin{proof}\ONLINE{
\textbf{Membership}. As for $\rSDTDs$ (see Corollary \ref{SDTDpspacePtime}).

\medskip

\noindent \textbf{Hardness}. For (1) we know that both \textsc{equiv}$_{[\scriptsize\xDTD{\nRE} \normalsize]}$ and \textsc{equiv}$_{[\scriptsize\xDTD{\nFA} \normalsize]}$ are also $\classh{PSPACE}$ \citep{StockmeyerMeyer73,MartensNeven06}.

For (2) and (3) we use the same reduction (from $\concatUnivR$) that we have used in Corollary \ref{SDTDpspacePtime}, where the problem $\consistent{\rSDTD}$ is replaced now by $\consistent{\rDTD}$. In particular, we just notice that also in this case $(\tau_1, \tau_2, \tau_3)$ is consistent with $T = s(a(\f_1 \f_2) \ a(\f_3))$ if and only if $\langOf{\pi_1(\tilde{s}_1)} \circ \langOf{\pi_2(\tilde{s}_2)} = \langOf{\pi_3(\tilde{s}_3)} = \Sigma^*$.}
\end{proof}

\section{Top-down design}\label{TopDownDesign}

In this section, we consider design problems where we start from a kernel and a given global type, and we show how to reduce each of these problems on trees to a set of typing problems on strings. In the next section, we will show how to solve the problems for strings.

\subsection{$\rDTDs$}

We briefly present some obvious results on equivalence of $\rDTDs$. The proof of the next result is obvious and thus omitted.

\begin{proposition}\label{equivDTDs}
Two reduced $\rDTDs$ $\tau_1$ and $\tau_2$ are equivalent if and only if the following are true:
\begin{enumerate}
  \item They have the same root;

  \item They use the same element names;

  \item For each element name $a$, the content models of $a$ in both are equivalent.
\end{enumerate}
\end{proposition}

\begin{theorem}\label{reductionThmDTD}
Let $D = \langle \tau, T(\f_n) \rangle$ be a distributed design where $\tau = \langle \calL,\pi,s\rangle$ is an $\rDTD$. The following are equivalent:
\begin{enumerate}[$(1)$]
  \item $D$ admits a local $\rDTD$-typing;

  \item The $\mathcal{R}$-design $D^x = \langle \pi(\lab{x}), \childstr{x}\rangle$ admits a local $\mathcal{R}$-typing for each node $x$ in $T$ where $\lab{x} \in \calL$.
\end{enumerate}
\end{theorem}
\begin{proof}
$(1) \Rightarrow (2)$: Let $(\tau_n)$ be a local typing for $D$, then $\typeTn \equiv \tau$ holds. This means (by Proposition \ref{equivDTDs}) that for each node $x$ in $T$ such that $\lab{x}\in \calL$, the content model $\pi(\lab{x})$ of $x$ has an equivalent specification in $\typeTn$. But this means that the subset of types in $(\tau_n)$ in bijection with the functions of $\childstr{x}$ represents a local typing for $D^x$ as well.

\medskip

\noindent $(2) \Rightarrow (1)$: Also in this case, by Proposition \ref{equivDTDs}, since each node $x$ such that $\lab{x} \in \calL$ has a local typing, this means that such a typing allows describing exactly the content model $\pi(\lab{x})$. Thus, by combining all the local typings of the various string-designs with the content models of $\tau$ we obtain a $D$-consistent typing also local for $D$. To be more precise, we now show how to exploit the local string-typings for building a local typing for $D$. First of all we observe that, for each $i$ in $[1..n]$, there exists only one node $x$ of $T$ such that $\f_i$ is in the kernel string $\childstr{x}$ of $D^x$. Since each $D^x$ admits a local typing, then there is a sequence, say $(\tau_1^{str}, \ldots, \tau_n^{str})$, of string-types (one for each function) allowing that. In particular, if for some $x$, $\childstr{x}$ has no function, then this necessarily means that $D^x$ admits a trivial local typing, namely $\langOf{\pi(\lab{x})} = \{\childstr{x}\}$ must hold. Let $i$ be an index in $[1..n]$, and $x$ be the parent of $\f_i$. The new type (not necessarily reduced) $\tau_i = \langle \Sigma_i, \pi_i, s_i \rangle$ is defined as follows:
\begin{enumerate}[ \ $\centerdot$]
  \item $\Sigma_i = \Sigma \cup \{s_i\}$;
  \item $\pi_i$ contains all the rules of $\pi$ and the extra rule $\pi_i(s_i) = \tau_i^{str}$.
\end{enumerate}
Finally, it is very easy to see that, $T(\tau_n)$ is structurally equivalent to  $\tau$.
\end{proof}

\begin{corollary}\label{dtd2strReduction}
The problems $\loc{\rDTD}$, $\ml{\rDTD}$ $\perf{\rDTD}$, $\eloc{\rDTD}$, $\eml{\rDTD}$ and $\eperf{\rDTD}$ are logspace Turing reducible to $\locR$, $\mlR$, $\perfR$, $\elocR$, $\emlR$  and $\eperfR$, respectively.
\end{corollary}
\begin{proof}
Let $D = \langle \langle \calL,\pi,s\rangle, T(\f_n) \rangle$ be a top-down $\rDTD$-design.

Consider, firstly, the $\eloc{\rDTD}$ problem. Scan $T$ in document order, which is well known to be feasible in logarithmic space \citep{CookMcKenzie87}. For each node $x$ in $T$ such that $\lab{x} \in \calL$, solve the problem $\elocR$ for the design $D^x$.

If we consider, instead, the problem $\loc{\rDTD}$, as $(\tau_n)$ is $D$-consistent, then $\typeTn$ exists and there are no different content models for the same element name. So it is also enough to scan $T$ in document order, and for each node $x$ in $T$ such that $\lab{x} \in \calL$, solve the problem $\locR$ for the design $D^x$ and the subset of types from $(\tau_n)$ in bijection with the functions in $D^x$.

For the maximal and perfect requirements, as they are specializations of the local requirement, it is enough to observe that, by Theorem \ref{reductionThmDTD}, they only depend on the structure of the various $D^x$.
\end{proof}

\subsection{$\rSDTDs$}

Before proving that a similar reduction still holds for $\rSDTDs$, we need a proposition and a new definition.

\begin{proposition}\label{equivSDTDs}
Let $\tau_1 = \langle \Sigma_1, \tilde{\Sigma_1}, \pi_1, \tilde{s}_1, \mu_1 \rangle$ and $\tau_2 = \langle \Sigma_2, \tilde{\Sigma_2}, \pi_2, \tilde{s}_2, \mu_2 \rangle$ be two reduced $\rSDTDs$. If they are equivalent, then for each $i,j$ in $[1..2]$, and each $\tilde{a}_i \in \tilde{\Sigma_i}$ there is $\tilde{a}_j \in \tilde{\Sigma_j}$ such that $\mu_i(\pi_i(\tilde{a}_i)) = \mu_j(\pi_j(\tilde{a}_j))$.
\end{proposition}
\begin{proof}\ONLINE{
As $\tau_1$ is reduced, for each specialized element name $\tilde{a}_1 \in \tilde{\Sigma_1}$ there is a tree $t \in \langOf{\tau_1}$ such that its unique witness $t'$ contains at least one node having $\tilde{a}_1$ as label. So, let us fix $\tilde{a}_1$, $t$, and a node $x$ of $t$ such that $\ancstr{x}$ ends with the element name $a$. As $\tau_1$ and $\tau_2$ are equivalent, there must exist also a unique witness $t''$ for $t$ produced by $\tau_2$. Let us denote by $\tilde{a}_2$ the specialized element name associated to $x$ in $t''$. As $\langOf{\tau_1} = \langOf{\tau_2}$, if $\mu_1(\pi_1(\tilde{a}_1)) \neq \mu_2(\pi_2(\tilde{a}_2))$, then we would violate the ancestor-guarded subtree exchange property.}
\end{proof}

\begin{definition}
Let $D = \langle \tau, T(\f_n) \rangle$ be a distributed design where the type $\tau = \langle \Sigma, \tilde{\Sigma}, \pi, \tilde{s}, \mu \rangle$ is an \rSDTD. For each node $x$ in $T$ such that $\lab{x} \in \calL$ we denote by $D^x = \langle \pi(\tilde{a}),w_x\rangle$ the the unique string-design induced by $D$, where $\tilde{a}$ is the (unique) witness assigned by $\tau$ to $x$. Moreover, $w_x = \varepsilon$ if $x$ is a leaf, and it is the string obtained from $\textsf{children}(x)$ by changing each non-function node with the corresponding (unique) witness assigned by $\tau$, otherwise.\qed
\end{definition}

\begin{theorem}\label{reductionThmSDTD}
Let $D = \langle \tau, T(\f_n) \rangle$ be a distributed design where the type $\tau = \langle \Sigma, \tilde{\Sigma}, \pi, \tilde{s}, \mu \rangle$ is an $\rSDTD$. The following are equivalent:
\begin{enumerate}[$(1)$]
  \item $D$ admits a local $\rSDTD$-typing;

  \item Each $\mathcal{R}$-design induced by $D$ admits a local $\mathcal{R}$-typing.
\end{enumerate}

\end{theorem}
\begin{proof}\ONLINE{
$(1) \Rightarrow (2)$: Since $D$ admits a local $\rSDTD$-typing, say $(\tau_n)$, then $(\tau_n)$ is $\rSDTD$-consistent with $T$, and $\typeTn \equiv \tau$ holds. For each node $x$ of $T$ having an element name as label, say $a$, consider the unique witness associated by $\tau$ and $\typeTn$ to $x$, say $\tilde{a}_{\tau}$ and $\tilde{a}_{T}$, respectively. By hypothesis, both $\pi_{\tau}(\tilde{a}_{\tau})$ and $\pi_{T}(\tilde{a}_{T})$ satisfy the single-type requirement, and by Proposition \ref{equivSDTDs}, $\mu_{\tau}(\pi_{\tau}(\tilde{a}_{\tau})) \equiv \mu_{T}(\pi_{T}(\tilde{a}_{T}))$. Thus, as $\pi_{\tau}(\tilde{a}_{\tau})$ and the children of $x$ have a local decomposition induced by $(\tau_n)$, then also $\pi_{T}(\tilde{a}_{T})$ and the children of $x$ (namely $D^x$) have.

\medskip

\noindent $(2) \Rightarrow (1)$: We show, if the premise is true, how to build $(\tau_n)$ in such a way that $\typeTn \equiv \tau$ holds and in particular that, from a structural point of view, $\typeTn$ is equivalent to $\tau$. We denote by $(\tau_1^{str}, \ldots, \tau_n^{str})$ one possible sequence of string-types satisfying contemporarily all the string-designs. In particular, if for some node $x$ of $T$ the kernel string of $D^x$ has no function call this necessarily means that $D^x$ admits a trivial local typing, namely $\mu(\pi(\tilde{a})) = \childstr{x}$, where $\tilde{a}$ is the unique witness assigned by $\tau$ to $x$. In particular, for each function $\f_i$ consider its parent node in $T$, say $x$. The new type $\tau_i = \langle \Sigma_i, \tilde{\Sigma_i}, \pi_i, \tilde{s}_i, \mu_i \rangle$ is defined as follows:
\begin{enumerate}[ \ $\centerdot$]
  \item $\Sigma_i = \Sigma \cup \{s_i\}$;
  \item $\tilde{\Sigma}_i = \tilde{\Sigma} \cup \{\tilde{s_i}\}$;
  \item $\pi_i$ contains all the rule of $\pi$ and the extra rule $\pi_i(\tilde{s_i}) = \tau_i^{str}$;
  \item $\tilde{s}_i$ is the usual extra witness for the root of any tree in $\langOf{\tau_i}$;
  \item $\mu_i$ is defined as $\mu$ and also $\mu_i(\tilde{s_i}) = s_i$.
\end{enumerate}
It is very easy to see that, if we build $T(\tau_n)$ without renaming the specialized element names we obtain exactly $\tau$, and so $\typeTn \equiv \tau$. In particular, when we assign the witnesses to the non-function nodes of $T$ we choose exactly those assigned by $\tau$. The only difference may be in the specification of the content models because the recomposition after a decomposition may produce a different structure (for instance a different $\nFA$) being, anyway, equivalent to the original one. Notice that, the ``ancestor-guarded subtree exchange'' property is guarantied because we also require that all the designs without any function call admit local typings. For instance, if $D^x = \langle \pi(\tilde{a}), \varepsilon \rangle$ admits a local typing, where $x$ is a leaf node of $T$ and $\tilde{a}$ is its witness assigned by $\tau$, this necessarily means that $\pi(\tilde{a}) = \{\varepsilon\}$, and we automatically take it into account when we build each $\tau_i$.}
\end{proof}
\begin{corollary}\label{sdtd2strReduction}
The problems $\loc{\rSDTD}$, $\ml{\rSDTD}$, $\perf{\rSDTD}$, $\eloc{\rSDTD}$, $\eml{\rSDTD}$ and $\eperf{\rSDTD}$ are logspace Turing reducible to $\locR$, $\mlR$, $\perfR$, $\elocR$, $\emlR$  and $\eperfR$, respectively.
\end{corollary}
\begin{proof}
Exactly the same as for $\rDTDs$.
\end{proof}

\subsection{$\rEDTDs$}
Although $\rEDTDs$ have nice properties simplifying the $\consistent{\rEDTD}$ problem and the construction of $\typeTn$ (when we start from a kernel and an $\rEDTD$-typing), things dramatically change when we consider the problems concerning locality. The freedom of using, in the same content model, various specialized element names for the same element name has a price. Consider the following example.

\begin{example}
Let $D = \langle \tau, T \rangle$ be a $\xEDTD{\dRE}$-design where $T = s_0 (\f_1 \f_2)$ and $\tau = \langle \Sigma, \tilde{\Sigma}, \pi, \tilde{s}_0, \mu \rangle$. In particular, $\pi(\tilde{s}_0) = \tilde{a}^1 (\tilde{b}^1)^* + \tilde{a}^2 (\tilde{b}^2)^*$; $\pi(\tilde{a}^1) = \tilde{c}^1$; $\pi(\tilde{a}^2) = \tilde{d}^1$; $\pi(\tilde{b}^1) = \tilde{e}^1 + \tilde{g}^1$; $\pi(\tilde{b}^2) = \tilde{g}^1 + \tilde{h}^1$.
It is not hard to see that the string-design $\langle \pi(\tilde{s}_0), \f_1 \f_2 \rangle$ admits only two maximal local typings:
\[
  (\varepsilon, \ \tilde{a}^1 (\tilde{b}^1)^* + \tilde{a}^2 (\tilde{b}^2)^*)\ \ \ \ \ \ \
  (\tilde{a}^1 (\tilde{b}^1)^* + \tilde{a}^2 (\tilde{b}^2)^*, \ \varepsilon)
\]
But, only the first one is also maximal for $D$, while the actual second one is $(\tilde{a}^1 (\tilde{b}^1)^* + \tilde{a}^2 (\tilde{b}^2)^*, \ (\tilde{b}^3)^*)$ where $\langOf{\tau_2(\tilde{b}^3)} = b(g)$. \qed
\end{example}

The problem highlighted by the previous example originates from the fact that $\tilde{b}^1$ and $\tilde{b}^2$ can not be considered completely distinct as $\tilde{a}^1$ and $\tilde{a}^2$ (notice that $\langOf{\tau(\tilde{a}^1)} \cap \langOf{\tau(\tilde{a}^2)} = \emptyset$), and as we naturally do for two different symbols of an alphabet in string languages, yet they are witnesses for two sets of trees with a nonempty intersection. In fact, $\langOf{\tau(\tilde{b}^1)} \cap \langOf{\tau(\tilde{b}^2)} = b(g)$ can be part of $\tau_2$ in the second maximal local typing for $D$.

From this, it is unclear whether, by only analyzing content models (such as $\pi(\tilde{s}_0)$, in the previous example), we can decide whether a given design admits at least a local typing. Clearly, if we apply $\mu$ to both $(\tilde{a}^1 (\tilde{b}^1)^* + \tilde{a}^2 (\tilde{b}^2)^*) \cdot (\tilde{b}^3)^*$ and to $\pi(\tilde{s}_0)$ we obtain the same string-language, namely $ab^*$, but unfortunately, this is only a necessary condition and even if $(ab^*,b^*)$ is a maximal local typing for $ab^*$, it is not clear how to assign the witnesses for obtaining $(\tilde{a}^1 (\tilde{b}^1)^* + \tilde{a}^2 (\tilde{b}^2)^*, (\tilde{b}^3)^*)$.

The following theorems, give a further idea of the higher complexity of locality when we consider $\rEDTD$-designs.

\begin{theorem}[\citep{Seidl90,Suciu02}]\label{EquivEDTDs}
Problems \ $\equivalence{\xEDTD{\nFA}}$ \ and \ $\equivalence{\xEDTD{\nRE}}$ \ are \ $\classc{EXPTIME}$.
\end{theorem}

\begin{theorem}
Problems $\eloc{\rEDTD}$, $\eml{\rEDTD}$, and $\eperf{\rEDTD}$ are at least as hard as $\equivalence{\rEDTD}$.
\end{theorem}
\begin{proof}
We define a logspace transformation $\varphi$ from $\equivalence{\rEDTD}$ to $\eloc{\rEDTD}$.
Afterwards, we show that the statement also holds for the other two problems by using exactly the same reduction. Let $\tau',\tau''$ be two arbitrary $\rEDTDs$. The application of $\varphi$ to this pair produces the design $D = \langle \tau, T \rangle$, where
\begin{enumerate}[ \ $\centerdot$]
  \item $T = s_0(\f_1 \ c \ \f_2)$

  \item $\pi(\tilde{s}_0) = \mathcal{R}(\tilde{a}_1 \tilde{c}_1 \tilde{d}_1 + \tilde{b}_1 \tilde{c}_1 \tilde{d}_2)$

  \item $\pi(\tilde{d}_1) = \mathcal{R}(\tilde{s}_0')$, where $s_0'$ is the root of the trees in $\langOf{\tau'}$

  \item $\pi(\tilde{d}_2) = \mathcal{R}(\tilde{s}_0'')$, where $s_0''$ is the root of the trees in $\langOf{\tau''}$

  \item $\pi(\tilde{a}_1) = \pi(\tilde{b}_1) = \pi(\tilde{c}_1) = \mathcal{R}(\varepsilon)$

  \item $\tilde{c}_1$ does not appear in any other content model of $\tau$ and $c$ appears exactly once in any tree in $\langOf{\tau}$
\end{enumerate}
Informally, $\langOf{\tau} = s_0(acd(\langOf{\tau'}) + bcd(\langOf{\tau''}))$. First of all, we observe that all the new content models (other than those being already in $\tau'$ and $\tau''$) can be represented by $\mathcal{R}$-types, even $\dREs$. Now, it is easy to see that $D$ admits a local typing iff $\langOf{\tau(\tilde{d}_1)} = \langOf{\tau(\tilde{d}_2)}$ iff $\tau' \equiv \tau''$. It is $\langOf{\tau_1} = s_1(a + b)$ and $\langOf{\tau_2} =  s_2(d(\langOf{\tau'}))$. Finally, we just notice that if $\tau' \equiv \tau''$ holds, then $(\tau_1,\tau_2)$ is the unique maximal local typing for $D$ which is even perfect.
\end{proof}

\begin{corollary}\label{ElocEdtdNfaExpHard}
Problems $\eloc{\rEDTD}$, $\eml{\rEDTD}$, and $\eperf{\rEDTD}$ are $\classh{EXPTIME}$ if $\mathcal{R} \in \{\nFA, \nRE\}$.
\end{corollary}

The equivalence between $\langOf{\tau(\tilde{d}_1)} = \langOf{\tau(\tilde{d}_2)}$, in the previous reduction, is necessary because we do not know, a priori, whether $\f_1$ is imposing a constraint on $\f_2$ or not.
In particular, this is an extreme case of the fact that $\langOf{\tau(\tilde{d}_1)} \cap \langOf{\tau(\tilde{d}_2)} \neq \emptyset$.

What we really need is to be able to consider \emph{completely distinct}, in the same content model, each pair of different specialized element names of the form $\tilde{a}$ and $\tilde{a}'$, namely $\langOf{\tau(\tilde{a})} \cap \langOf{\tau(\tilde{a}')} = \emptyset$. To do that, given an $\rEDTD$, we construct an equivalent $\nUTA$ \cite{MartensNiehren07}, we transform it into an $\dUTA$ \cite{ComonDauchetGilleron07}, and finally we try to derive a new $\rEDTD$ satisfying our requirement. If $\mathcal{R} = \dRE$ the last step could not be always possible.

Given an $\rEDTD$ $\tau = \langle \Sigma, \tilde{\Sigma}, \pi, \tilde{s}, \mu \rangle$, an equivalent $\nUTA$ $\autom{A} = \langle K, \Sigma, \Delta, F \rangle$ can be constructed as follows: $K = \tilde{\Sigma}$; $\Delta(\tilde{a},a) = \nFA(\pi(\tilde{a}))$, for each $\tilde{a} \in \tilde{\Sigma}$; $F = \{\tilde{s}_0\}$.
Now we want to transform $\autom{A}$ into an equivalent $\dUTA$ $\autom{A}^d$ (that may be exponential in size). Notice that $\autom{A}^d$ will have only one final state as well.
Finally, we convert again $\autom{A}^d$ (whenever it is possible) into an $\rEDTD$ $\tau^d$ as follows: $\tilde{\Sigma} = K$; $\pi(\tilde{a}) = \mathcal{R}(\Delta(\tilde{a},a))$, for each $\tilde{a} \in K$.

\begin{lemma}
Let $\tau^d$  be an $\rEDTD$ built as above. For each element name, say $a$, and each pair $\tilde{a},\tilde{a}'$ of different specialized element names in $\tilde{\Sigma}^d(a)$, then $\langOf{\tau^d(\tilde{a})} \cap \langOf{\tau^d(\tilde{a}')} = \emptyset$.
\end{lemma}
\begin{proof}
It is easy to see that by a (bottom-up) run of $\autom{A}^d$ over each tree $t \in \langOf{\tau^d(\tilde{a})} \cup \langOf{\tau^d(\tilde{a}')}$, there is only one possible state (between $\tilde{a}$ and $\tilde{a}'$) that can be associated to the root of $t$, and the states of $\autom{A}^d$ coincide with the specialized element names of $\tau^d$.
\end{proof}

Now we are ready for handling $\rEDTDs$ (that we call \textbf{normalized}) satisfying the above property. But before we introduce a general property of $\rEDTDs$.

\begin{proposition}\label{generalEDTD}\emph{\cite{MartensNeven06}}
Let $\tau$ be an $\rEDTD$. Whenever for two trees $t_1,t_2 \in \langOf{\tau}$ with nodes $x_1$ and $x_2$, respectively,
there are witnesses $\tilde{t}_1$ and $\tilde{t}_2$ assigning the same specialized element name to both $x_1$ and $x_2$, then the trees obtained, from $t_1$ and $t_2$, by exchanging $\treet{t_1}{x_1}$ and $\treet{t_2}{x_2}$ are still in $\langOf{\tau}$.
\end{proposition}

The following lemma holds for general $\rEDTDs$ but is it also useful for normalized $\rEDTDs$. Consider the design $D = \langle \tau^d,T\rangle $ where $T = s_0 (a(\f_1) \ \f_2)$ and $\tau^d$ is a normalized $\xEDTD{\nRE}$ having $\pi(\tilde{s}_0) = (\tilde{a}^1 + \tilde{a}^2)^+$ (we ignore the other content models). As $\langOf{\tau^d(\tilde{a}^1)} \cap \langOf{\tau^d(\tilde{a}^2)} = \emptyset$, it is clear that the unique maximal local typing $(\tau_1,\tau_2)$ for $D$ has $\pi_1(\tilde{s}_1) = \pi(\tilde{a}^1) + \pi(\tilde{a}^2)$ and $\pi_2(\tilde{s}_2) = (\tilde{a}^1 + \tilde{a}^2)^*$. Thus the node under the root labeled by $a$ may have either $\tilde{a}^1$ or $\tilde{a}^2$ as witness depending on the tree replacing $\f_1$.

\begin{lemma}\label{edtdImpLemma}
Let $D = \langle \tau, T \rangle$ be an $\rEDTD$-design and $(\tau_n)$ be a local typing for $D$. For each node $x$ of $T$ having an element name as label, say $a$, there is a set of specialized element names $\tilde{\Sigma}^x \subseteq \tilde{\Sigma}(a)$ such that $\bigcup_{\tilde{a} \in \tilde{\Sigma}^x} \langOf{\tau(\tilde{a})} = \langOf{T(\tau_n, \tilde{a}_0^x)}$.
\end{lemma}
\begin{proof}\ONLINE{
By hypothesis, $\langOf{\tau} = \langOf{T(\tau_n)}$. We recall that this equivalence is obtained when we consider all the possible extensions of $T$. Let $x$ be a node of $T$ having an element name as label, say $a$, and $k$ be the cardinality of $\tilde{\Sigma}(a)$. First of all, we observe that if $\tilde{\Sigma}^x = \tilde{\Sigma}(a)$, then
\[
\langOf{T(\tau_n, \tilde{a}_0^x)} \subseteq \langOf{\tau(\tilde{a}^1)} \cup \ldots \cup \langOf{\tau(\tilde{a}^k)}
\]
is trivially true. In fact, in such a case, the first member has to be necessarily a subset of the second one because, otherwise, in some extension of $T$ there would be a subtree rooted in $x$ which  $\tau$ cannot produce any witness for.

Starting from $\tilde{\Sigma}^x = \tilde{\Sigma}(a)$ we \textbf{claim} that each $\tilde{a}^j$, with $1 \leq j \leq k$, is either a ``friend'' to keep in $\tilde{\Sigma}^x$ or an ``intruder'' to remove from $\tilde{\Sigma}^x$. Finally, the resulting $\tilde{\Sigma}^x$ will prove the statement. Consider now each $\tilde{a}^j$. We distinguish two cases:
\begin{enumerate}
  \item We say that $\tilde{a}^j$ is a \emph{friend} if and only if $\langOf{\tau(\tilde{a}^j)} \subseteq \langOf{T(\tau_n, \tilde{a}_0^x)}$ because it then clearly contributes to prove the statement. We leave it in $\tilde{\Sigma}^x$.

  \item We say that $\tilde{a}^j$ is an \emph{intruder} (and we remove it from $\tilde{\Sigma}^x$) if and only if one of the following is true:

  \begin{enumerate}[$\centerdot$]
      \item $\langOf{\tau(\tilde{a}^j)} \cap \langOf{T(\tau_n, \tilde{a}_0^x)} = \emptyset$ because even if we remove it from $\tilde{\Sigma}^x$, then $\langOf{T(\tau_n, \tilde{a}_0^x)} \subseteq \bigcup_{\tilde{a} \in \tilde{\Sigma}^x} \langOf{\tau(\tilde{a})}$ is still true.

      \item $\langOf{\tau(\tilde{a}^j)} \cap \langOf{T(\tau_n, \tilde{a}_0^x)} \neq \emptyset$, $\langOf{\tau(\tilde{a}^j)} \nsubseteq \langOf{T(\tau_n, \tilde{a}_0^x)}$, and for each (sub)tree $st_1$ in the intersection there is not a tree $t_1 \in \langOf{T(\tau_n)}$ having $st_1$ as subtree rooted in $x$, such that at least a witness $\tilde{t}_1$ from $\tau$ associates $\tilde{a}^j$ to $x$. It is still an intruder because for each (sub)tree $st_1$ in the intersection and for each possible trees $t_1 \in \langOf{T(\tau_n)}$ having $st_1$ as subtree rooted in $x$, since tree $t_1$ must necessarily have a witness $\tilde{t}_1$ from $\tau$, then $\tilde{t}_1$ associates $\tilde{a}^{i\neq j}$ to $x$ entailing that each $st_1$ is contained also in  \langOf{\tau(\tilde{a}^i)}. Finally, even if we remove $\tilde{a}^j$ from $\tilde{\Sigma}^x$, then $\langOf{T(\tau_n, \tilde{a}_0^x)} \subseteq \bigcup_{\tilde{a} \in \tilde{\Sigma}^x} \langOf{\tau(\tilde{a})}$ is still true.
      \item $\langOf{\tau(\tilde{a}^j)} \cap \langOf{T(\tau_n, \tilde{a}_0^x)} \neq \emptyset$, $\langOf{\tau(\tilde{a}^j)} \nsubseteq \langOf{T(\tau_n, \tilde{a}_0^x)}$, and for some (sub)tree $st_1$ in the intersection there is a tree $t_1 \in \langOf{T(\tau_n)}$ having $st_1$ as subtree rooted in $x$ and there is a witness $\tilde{t}_1$ from $\tau$ associating $\tilde{a}^j$ to $x$. Anyway, this last case is not possible because it would contradict the hypothesis that $(\tau_n)$ is local. In fact, consider a tree $t_2 \in \tau$ where some node $y$ has also witness $\tilde{a}^j$ and the subtree rooted in $y$ is $st_2 \in \langOf{\tau(\tilde{a}^j)} \mysetminus \langOf{T(\tau_n, \tilde{a}_0^x)}$. For each node $z$ in $t_1$ such that $\ancstrt{t_1}{z} = \ancstrt{t_1}{x}$ and the witness for $z$ in $\tilde{t}_1$ is $\tilde{a}^j$, by Proposition \ref{generalEDTD}, we may replace (in $t_1$) each subtree rooted in $z$ with $st_2$. The new tree is still in $\tau$ but cannot be obtained by any possible extension of $T$ as $\langOf{T(\tau_n, \tilde{a}_0^x)}$ contains all the possible trees obtainable by all the extensions of the functions under $x$ and $st_2$ is not among them.
  \end{enumerate}
\end{enumerate}
Thus, if we consider all the feasible cases, then the claim is true as well as the theorem: $\tilde{\Sigma}^x$ is exactly the set of all friends (or a possible subset if some $\langOf{\tau(\tilde{a}^j)}$ can be obtained by the union of some $\langOf{\tau(\tilde{a}^i)}$ with $\tilde{a}^i$ still in $\tilde{\Sigma}^x$).}
\end{proof}

\begin{definition}
Let $D = \langle \tau, T(\f_n) \rangle$ be an $\rEDTD$-design where the type $\tau = \langle \Sigma, \tilde{\Sigma}, \pi, s, \mu \rangle$ is normalized. We denote by
\begin{enumerate}[ \ $\centerdot$]
  \item $\kappa$ any function associating to each node $x$ of $T$ either a set $\tilde{\Sigma}^x \subseteq \tilde{\Sigma}(a)$ if $a$ is the label of $x$, or the set $\{\f\}$ if $\f$ is the label of $x$.
  \item $D^x_{\kappa} = \langle \pi(\kappa(x)),B^x\rangle$, for each node $x$ in $T$ with $\lab{x} \in \calL$, the box-design induced by $D$ and $\kappa$ where either $B^x = \{\varepsilon\}$ if $x$ is a leaf node, or $B^x = \kappa(y_1) \ldots \kappa(y_k)$ if $\textsf{children}(x) = y_1 \ldots y_k$.
\end{enumerate}
Given a sound typing $(\tau_n)$ for $D$, we say that
\begin{enumerate}[ \ $\centerdot$]
  \item $\kappa$ is induced by the pair $(\tau_n)$ and $T$ if, for each non-function node $x$ of $T$, $\kappa(x)$ contains exactly all the specialized element names associated to $x$ by validating each possible tree in $\extTn$.

  \item $\kappa' \leq \kappa$ iff $\kappa'(x) \subseteq \kappa(x)$, for each $x$. \qed
\end{enumerate}
\end{definition}

The intention is to relate locality properties about $D$ with locality properties about each $D^x_{\kappa}$ similarly as we made for $\rSDTDs$, with the difference that here $D^x_{\kappa}$ depends on the choice of $\kappa$. Unfortunately, although $\tau$ is normalized, if $D$ admits local typings, then $\kappa$ may not be unique. Consider the following example.

\begin{example}
Let $D = \langle \tau, T \rangle$ be a normalized $\xEDTD{\dRE}$-design where $T = s_0 (\f_1 a(\f_2) \f_3)$, $\tau = \langle \Sigma, \tilde{\Sigma}, \pi, \tilde{s}_0, \mu \rangle$, $\pi(\tilde{s}_0) = (\tilde{a}^1 \tilde{a}^2)^+$, $\pi(\tilde{a}^1) = \tilde{b}^1$, and $\pi(\tilde{a}^2) = \tilde{c}^1$. We have two successfully mappings $\kappa^1$, $\kappa^2$ such that
\begin{enumerate}[ \ $\centerdot$]
  \item $\kappa^1(x_1) = \tilde{s}_0$, $\kappa^1(x_3) = \tilde{a}^1$, $D^{x_1}_{\kappa^1} = \langle (\tilde{a}^1 \tilde{a}^2)^+, \f_1 \tilde{a}^1 \f_3\rangle$, and $D^{x_3}_{\kappa^1} = \langle \tilde{b}^1, \f_2\rangle$

  \item $\kappa^2(x_1) = \tilde{s}_0$, $\kappa^2(x_3) = \tilde{a}^2$, $D^{x_1}_{\kappa^2} = \langle (\tilde{a}^1 \tilde{a}^2)^+, \f_1 \tilde{a}^2 \f_3\rangle$, and $D^{x_3}_{\kappa^2} = \langle \tilde{c}^1, \f_2\rangle$
\end{enumerate}
From them we have two different maximal local typings for $D$:
\[
  ((\tilde{a}^1 \tilde{a}^2)^*, \ \tilde{b}^1, \ \tilde{a}^2(\tilde{a}^1 \tilde{a}^2)^*)\ \ \ \ \ \ \
  ((\tilde{a}^1 \tilde{a}^2)^*\tilde{a}^1, \ \tilde{c}^1, \ (\tilde{a}^1 \tilde{a}^2)^*)
\]
\noindent Notice that they are substantially different and also that from the other possible mapping $\kappa^3$, where $\kappa^3(x_3) = \{\tilde{a}^1,\tilde{a}^2\}$, we cannot derive any local typing because if $\f_2$ is replaced by $b$, then $\f_3$ must start with $a(c)$, and if $\f_2$ is replaced by $c$, then $\f_1$ must start with $a(b)$. But $((\tilde{a}^1 \tilde{a}^2)^*\tilde{a}^1, \ \tilde{b}^1+\tilde{c}^1, \ \tilde{a}^2(\tilde{a}^1 \tilde{a}^2)^*)$ is neither local (even) nor sound. \qed
\end{example}

Now we prove the main results of this section.

\begin{theorem}\label{reductionThmEDTD}
Let $D = \langle \tau, T(\f_n) \rangle$ be a distributed design where $\tau$ is a \textbf{normalized} $\rEDTD$. The following are equivalent:
\begin{enumerate}[$(1)$]
  \item $D$ admits a local typing;

  \item There is a function $\kappa$, as defined above, such that each box-design $D^x_{\kappa}$ admits a local typing.
\end{enumerate}
\end{theorem}
\begin{proof}
$(1) \Rightarrow (2)$: Let $(\tau_n)$ be a local typing for $D$, then $T(\tau_n) \equiv \tau$ holds. Consider the function $\kappa$ induced by $(\tau_n)$ and $T$ (the choice is consistent with Lemma \ref{edtdImpLemma}). As $\tau$ is normalized, there is only one possibility for validating (in a bottom-up way) each tree in $\extTn$. If for some node $x$ of $T$ the box-design $D^x_{\kappa}$ did not admit any local typing, then there would be no possibility of generating all the strings in $\pi(\kappa(x))$. Contradiction.

\medskip

\noindent $(2) \Rightarrow (1)$: If for some $\kappa$ each box-design $D^x_{\kappa}$ admits a local typing, then we can construct each type $\tau_i$ as made for $\rSDTDs$, in such a way that $T(\tau_n)$ is structurally equivalent to $\tau$.
\end{proof}

\begin{corollary}\label{edtd2strReduction}
Problem $\eloc{\rEDTD}$ (or $\eml{\rEDTD}$ but $\mathcal{R} \neq \dRE$) for \textbf{normalized} $\rEDTDs$ is decidable by an oracle machine in $\class{NP}^\mathcal{C}$ where $\mathcal{C}$ is the complexity class of solving $\elocRB$ (or $\emlRB$).
\end{corollary}
\begin{proof}
Let $\tau = \langle \calL,\pi,s\rangle$ be a type and $T(!f_n)$ be a kernel. Consider the $\eloc{\rEDTD}$ problem and the following algorithm:
\begin{enumerate}
  \item \textbf{Guess}: the function $\kappa$;

  \item \textbf{Check}: call $\elocRB$ over $D^x_{\kappa}$ for each node $x$ of $T$ with $\lab{x} \in \calL$.
\end{enumerate}
For $\eml{\rEDTD}$ we use the same algorithm since, in general ($\mathcal{R} \neq \dRE$), a maximal local typing always exists if there is a local one.
\end{proof}

Problem $\eml{\xEDTD{\dRE}}$ will be discussed in Section \ref{ComplexityForTrees}.

\begin{theorem}
Let $D = \langle \tau, T(\f_n) \rangle$ be a distributed design where $\tau$ is a \textbf{normalized} $\rEDTD$. The following are equivalent:
\begin{enumerate}[$(1)$]
  \item $D$ admits a perfect typing;

  \item There is a function $\kappa$ such that each $D^x_{\kappa}$ admits a perfect typing, and for each sound typing $(\tau_n')$ for $D$, $\kappa' \leq \kappa$ where $\kappa'$ is induced by $(\tau_n')$.
\end{enumerate}
\end{theorem}
\begin{proof}
$(1) \Rightarrow (2)$: Let $(\tau_n)$ be the perfect typing of $D$ and $\kappa$ be the function induced by $(\tau_n)$. By Theorem \ref{edtd2strReduction}, each box-design $D^x_{\kappa}$ admits a local typing, and clearly it is perfect as $(\tau_n)$ is. Finally, we observe that since $(\tau_n') \leq (\tau_n)$, then $(\tau_n')$ can not induce in $\kappa'$ more elements than $(\tau_n)$.

\medskip

\noindent $(2) \Rightarrow (1)$: As we made for $\rSDTDs$, the typing $(\tau_n)$ that we can construct by the local typings of the various $D^x_{\kappa}$ (without renaming the specialized element names) together with the needful content models already in $\tau$ produces a type $\Tn$ structurally equivalent to $\tau$.
\end{proof}

\begin{corollary}\label{edtd2strPerfReduction}
Problem $\eperf{\rEDTD}$ for \textbf{normalized} $\rEDTDs$ is polynomial time reducible to $\eperfRB$.
\end{corollary}
\begin{proof}\ONLINE{
Let $\tau = \langle \calL, \pi, s \rangle$ be a type and $T(!f_n)$ be a kernel.

\medskip

\noindent \textsc{Proof Idea}: Build $\kappa$ in polynomial time and in a top-down style (this is the technical core of the proof) by assuming that a perfect typing exists. Thus, call $\eperfRB$ over $D^x_{\kappa}$ for each node $x$ of $T$ with $\lab{x} \in \calL$.

\medskip

\noindent More formally, let $x$ be the root of $T$ and $m$ be the number of its children. Consider the following steps:
\begin{enumerate}
  \item Build from $\childstr{x}$ a $\dRE$, that we call $r(x)$, as follows. For each $j$ in $[1..m]$,
      \begin{enumerate}[$\centerdot$]
        \item if \childstr{x}[j] is an element name, say $a$, then replace it with the set $\tilde{\Sigma}_j(a)$ (where the subscript means that all the specialized element names are renamed with $j$ as subscript);

        \item else, if \childstr{x}[j] is a function, then replace it with $\tilde{\Sigma}_j^*$ ($j$ has the same meaning as above);
      \end{enumerate}
  \item Build from $\pi(\tilde{s}_0)$ an $\mathcal{R}$-type $\tau(x)$ by replacing, in the alphabet of $\pi(\tilde{s}_0)$, each symbol of the form $\tilde{a}^{\ell}$ with $\tilde{a}^{\ell}_1, \ldots, \tilde{a}^{\ell}_m$.

  \item Perform the intersection $L = \langOf{r(x)} \cap \langOf{\tau(x)}$.

  \item For each child $y$ of $x$ having $a$ as label and position $j$, then $\kappa(y)$ contains all the elements of the form $\tilde{a}$ such that $\tilde{a}_j$ is in the alphabet of $L$.
\end{enumerate}
As we know $\kappa(y)$ for each child $y$ of $x$, then we repeat the previous steps for the children of $y$ by considering $\pi(\kappa(y))$ instead of $\pi(\tilde{s}_0)$. We will stop when we reach the leaves of $T$.

The algorithm is correct because if we have a look at the alphabet of $L$, we see that it contains, for each $j$ in $[1..m]$, exactly the specialized element names that we need to associate to the $j^{th}$ child of $x$ because are induced by all possible local typings. Intuitively, if the alphabet of $L$ contains, for instance, $\tilde{a}^{\ell}_1$ this means that there is a sound typing for $D$ that induces $\tilde{a}^{\ell}$ for the first child of $x$, and if the alphabet of $L$ does not contain, for instance, $\tilde{b}^{\ell}_3$ there is no sound typing for $D$ inducing $\tilde{b}^{\ell}$ for the third child of $x$.}
\end{proof}

Now we consider the remaining complexity result that does not require any reduction to strings.

\begin{theorem}
Problems $\loc{\rEDTD}$, $\ml{\rEDTD}$, and $\perf{\rEDTD}$ are at least as hard as $\equivalence{\rEDTD}$.
\end{theorem}
\begin{proof}
We define a logspace transformation $\varphi$ from $\equivalence{\rEDTD}$ to $\loc{\rEDTD}$. Afterwards, we show that the statement also holds for the other two problems. Let $\tau',\tau''$ be two arbitrary $\rEDTDs$. The application of $\varphi$ to the this pair produces the design $D = \langle \tau, T \rangle$ and the typing $\tau_1$, where $\langOf{\tau} = s_0(\langOf{\tau'})$,  $T = s_0(\f_1)$, and $\langOf{\tau_1} = s_1(\langOf{\tau''})$. Since $T(\tau_1)$ is exactly $s_0(\langOf{\tau''})$, it is clear that $\tau \equiv T(\tau_1) $ if and only if $\tau' \equiv \tau''$. Finally, we just notice that $\tau' \equiv \tau''$ iff $\tau_1$ is both perfect and maximal local as $T$ consists of just a function node other than the root.
\end{proof}

\begin{corollary}\label{CorolLocExpHard}
Problems $\loc{\xEDTD{\nFA}}$, $\ml{\xEDTD{\nFA}}$, and $\perf{\xEDTD{\nFA}}$ are $\classh{EXPTIME}$.
\end{corollary}

\begin{theorem}\label{LocEdtdNfaExpCompl}
Problem $\loc{\xEDTD{\nFA}}$ is $\classc{EXPTIME}$.
\end{theorem}
\begin{proof}
\textbf{(Membership)} Let $D = \langle \tau, T \rangle$ be an $\xEDTD{\nFA}$-design and $(\tau_n)$ be a $D$-consistent typing. Build $T(\tau_n)$ in polynomial time (by Proposition \ref{TnPolyTimeESize}) and check in exponential time if $T(\tau_n) \equiv \tau$ (by Theorem \ref{EquivEDTDs}).

\medskip

\noindent \textbf{(Hardness)} By Corollary \ref{CorolLocExpHard}.
\end{proof}

\section{The typing problems for words}\label{word}

We study in this section the typing problems for words. (Recall that most of our problems for trees has been reduced to problems for words.) We present a number of complexity results. We leave for the next section, two issues, namely $\perf{\nFA}$ and $\eperf{\nFA}$, for which we will need a rather complicated automata construction.  We start by recalling a definition and a result that we will use further.

\begin{theorem}[\citep{StockmeyerMeyer73}]
$\equivalence{\nFA}$ is $\classc{PSPACE}$.
\end{theorem}

The hardness of the $\equivalence{\nFA}$ problem is used to show some hardness results of our problems.

\begin{theorem}\label{check-nfa-hard}
Problems $\loc{\nFA}$, $\ml{\nFA}$, $\perf{\nFA}$ are $\classh{PSPACE}$.
\end{theorem}
\begin{proof}\ONLINE{
We define a logspace transformation $\varphi$, in such a way that
\[
    \equivalence{\nFA} \leq^\mathbf{\mathbf{L}}_m  \loc{\nFA}
\]
Afterwards, we show that the statement also holds for the other two problems. Let $\autom{A},$ $\autom{A}_1$ be two arbitrary $\nFAs$. The application of $\varphi$ to the pair $\autom{A},$ $\autom{A}_1$ produces the design $\langle \tau, w \rangle$ and the typing $\tau_1$, where $\tau = \autom{A}$, $w = \f_1$ and $\tau_1 = \autom{A}_1$. Since $w(\tau_1) = \autom{A}_1$, it is clear that $\tau \equiv w(\tau_1) $ if and only if $\autom{A} \equiv \autom{A}_1$. Finally, we just notice that $\autom{A} \equiv \autom{A}_1$ if and only if $\tau_1$ is both perfect and maximal local as $w$ consists of just a function.}
\end{proof}

We now consider upper bounds. Section \ref{perfect} will show that $\perf{\nFA}$ is in $\class{PSPACE}$. We next show that $\loc{\nFA}$ is.

\begin{theorem}\label{LocNfaPspaceComplete}
$\loc{\nFA}$ is in $\class{PSPACE}$ (so it is $\classc{PSPACE}$).
\end{theorem}
\begin{proof}
Let $w(\f_n)$ be a kernel string, $\tau$ be an $\nFA$, and $(\tau_n)$ be a typing. Since the new automaton $w(\tau_n)$ has size $\mathcal{O}(\|w\| + |(\tau_n)|)$, we can check in polynomial space if $w(\tau_n) \equiv \tau$.
\end{proof}

The proof that also $\ml{\nFA}$ is in $\class{PSPACE}$ requires more technical insights and it is deferred to Section \ref{ComplexityForTrees}.

Let us turn to the hardness of the $\exists$-versions of the problems.

\begin{theorem}\label{nfa-locHard}
$\eloc{nFA}$, $\eml{nFA}$, and $\eperf{nFA}$ are $\classh{PSPACE}$.
\end{theorem}
\begin{proof}\ONLINE{
We define a logspace transformation $\varphi$, in such a way that the following relations hold:
\begin{enumerate}[$(1)$]
  \item $\equivalence{\nFA} \leq^\mathbf{\mathbf{L}}_m \eloc{\nFA}$;

  \item $\equivalence{\nFA} \leq^\mathbf{\mathbf{L}}_m \eml{\nFA}$;

  \item $\equivalence{\nFA} \leq^\mathbf{\mathbf{L}}_m \eperf{\nFA}$.
\end{enumerate}
Let $\autom{A}_1 = \langle K_1, \Sigma_1, \Delta_1, s_1, F_1\rangle$, $\autom{A}_2 = \langle K_2, \Sigma_2, \Delta_2, s_2, F_2\rangle$ be two $\nFAs$. The application of $\varphi$ to the pair $(\autom{A}_1,\autom{A}_2)$ produces the design $D = \langle \autom{A}, w \rangle$ where
\begin{enumerate}[ \ $\centerdot$]
  \item $w = \f_1 \ c \ \f_2$, with $c$ being a fresh terminal symbol which does not belong to $(\Sigma_1 \cup \Sigma_2)$;

  \item  while automaton $\autom{A} = \langle K, \Sigma, \Delta, s, F\rangle$ is defined as follows: (i) $K =$ $K_1$ $\cup$ $K_2$ $\cup$ $\{s, p_c, q_c\}$; (ii) $\Sigma =$ $\Sigma_1$ $\cup$ $\Sigma_2$ $\cup$ $\{a,b,c\}$; (iii) $\Delta =$ $\Delta_1$ $\cup$ $\Delta_2$ $\cup$ $\{(s,a,p_c),$ $(s,b,p_c),$ $(p_c,c,q_c),$ $(q_c, \varepsilon, s_1),$ $(q_c, \varepsilon, s_2)\}$; (iv) $F = F_1 \cup F_2$.
\end{enumerate}
Intuitively, if we consider $\autom{A}_1$ and $\autom{A}_2$ as $\nREs$, then $\autom{A}$ is $(ac\autom{A}_1 + bc\autom{A}_2)$.

We claim that there is a local typing (similarly, maximal local, or perfect) for $D$ if and only if $\autom{A}_1 \equiv \autom{A}_2$. First of all, we observe that transformation $\varphi$ is extremely simple and it is clearly in logspace. In fact, string $w$ is a constant, while the choice of a terminal symbol which does not appear in $\autom{A}_1$ nor in $\autom{A}_2$ can be done in logspace, and also $\autom{A}$ can be obtained by merging $\autom{A}_1$ and $\autom{A}_2$ with a constant number of transitions. We prove the statement for (1) and we just notice that whenever there is a local typing for $D$, then the typing $((a+b),\autom{A}_1)$ is perfect (thus, also maximal).

\medskip

\noindent $(\Rightarrow$) \emph{If there is a local typing for $D$ then $\autom{A}_1 \equiv \autom{A}_2$}. Since $\autom{A} = (ac\autom{A}_1 + bc\autom{A}_2)$, then \langOf{ac\autom{A}_1} and \langOf{bc\autom{A}_2} form a partition of \langOf{\autom{A}}. In this case, any local typing must have the following form $((aX_1+bX_2), Y)$ where $X_1, X_2, Y$ are $\nFAs$. Clearly, all the strings accepted by $w$ are obtained by $aX_1cY$ and $bX_2cY$. Then $c\autom{A}_1 \equiv X_1cY$ and $c\autom{A}_2 \equiv X_2cY$ must hold. But since any string in $\langOf{\autom{A}_1}$ or $\langOf{\autom{A}_2}$ does not start with $c$, then necessarily $\langOf{X_1} = \langOf{X_2} = \varepsilon$. This way, $\autom{A}_1 \equiv Y$ and $\autom{A}_2 \equiv Y$ and then $\autom{A}_1 \equiv \autom{A}_2$.

\medskip

\noindent ($\Leftarrow$) \emph{If $\autom{A}_1 \equiv \autom{A}_2$, there is a local typing for $D$}. This part of the proof is trivial because $((a+b), \autom{A}_1)$ always represents a local typing for $D$.}
\end{proof}

We now have lower bounds for all these problems and some upper bounds. We will derive missing upper bounds using the construction of automata that we call ``perfect'' for given design problems.

\section{Perfect automaton for words}\label{perfect}
We next present the construction of the {\em perfect automaton} for a design word problem. The perfect automaton has the property that if a perfect typing exists for this problem, it is ``highlighted'' by the automaton.  This will provide a $\class{PSPACE}$ procedure for finding this perfect typing if it exists.

Let $\autom{A} = \langle K, \Sigma, \Delta, s, F\rangle$ be an $\nFA$. We can assume w.l.o.g. that it has no $\varepsilon$-transition. Given two states $q_i, q_f$ in $K$, a string $w$ in $\Sigma^*$ is
said to be \emph{delimited} in $\autom{A}$ by $q_i$ and $q_f$ if
$(q_i,w,q_f) \in \Delta^*$. By exploiting this notion, the sets of
all the states delimiting $w$ in $\autom{A}$ are defined as follows:
\begin{displaymath}
     Ini(\autom{A},w) = \{q_i \in K: \ \exists q_f \in K \ \textrm{s.t.}
\ (q_i,w,q_f) \in
     \Delta^*\}
\end{displaymath}
\begin{displaymath}
     Fin(\autom{A},w) = \{q_f \in K: \ \exists q_i \in K \ \textrm{s.t.}
\ (q_i,w,q_f) \in \Delta^*\}
\end{displaymath}
\noindent In particular, if $w = \varepsilon$, these two sets are
$Ini(\autom{A},\varepsilon)$ $=$ $Fin(\autom{A},\varepsilon)$ $=$
$K$. $Ini(\autom{A},w)$ is called the set of \emph{initial states}
while $Fin(\autom{A},w)$ is the set of \emph{final states} for the
word $w$. Given two states $q_i, q_f$ in $K$, the \emph{local
automaton} $\autom{A}(q_i,q_f)$ $=$ $\langle K' \subseteq K,$
$\Sigma,$ $\Delta',$ $q_i,$ $\{q_f\}\rangle$ induced from
$\autom{A}$ by $q_i, q_f$ is a portion of $\autom{A}$ containing all
those transitions of $\autom{A}$ leading from $q_i$ to $q_f$. More
precisely, for each pair of states $q,q'$ in $K$ and for
each symbol $a$ in $\Sigma$, $(q,a,q') \in \Delta'$ if and only if there are two strings $u,v$ in $\Sigma^*$ such that: $(q_i,u,q) \in
\Delta^*,$ $(q,a,q') \in \Delta,$ and $(q',v,q_f) \in \Delta^*$.
Finally, given two strings $w_1, w_2$ in $\Sigma^+$, then
$\autom{A}(w_1,w_2)$ is the \emph{set of all local automata} induced
by $w_1$ and $w_2$. It is formally defined as $\autom{A}(w_1,w_2) = \{\autom{A}(q_i,q_f) : q_i \in Fin(\autom{A},w_1), \ q_f \in Ini(\autom{A},w_2)\}.$ In particular, if
$w_i = \varepsilon$ for some $i$ in $[1..n]$, the kernel
string contains consecutive functions. In particular for the
previous definitions we have:
\begin{displaymath}
\autom{A}(w_1,\varepsilon) = \{\autom{A}(q_i,q_f) : q_i \in
Fin(\autom{A},w_1) \ \textrm{and} \ q_f \in K\}
\end{displaymath}
\begin{displaymath}
\autom{A}(\varepsilon,w_2) = \{\autom{A}(q_i,q_f) : q_i \in K \
\textrm{and} \ q_f \in Ini(\autom{A},w_2)\}
\end{displaymath}
\begin{displaymath}
\autom{A}(\varepsilon,\varepsilon) = \{\autom{A}(q_i,q_f) : q_i,q_f
\in K\}.
\end{displaymath}

Similarly, given a string $w$ in $\Sigma^*$, $\autom{A}(w)$ is the
\emph{set of all local automata} induced by $w$. It is defined as
$\autom{A}(w) = \{\autom{A}(q_i,q_f) : (q_i , w, q_f) \in \Delta^*
\}$ and in particular $\autom{A}(\varepsilon) = \{\autom{A}(q,q) : q
\in K\}$ is a set of $|K|$ automata, one for each state in $K$.

\begin{figure*}[h!] \centering
\fbox{\includegraphics[scale=0.56]{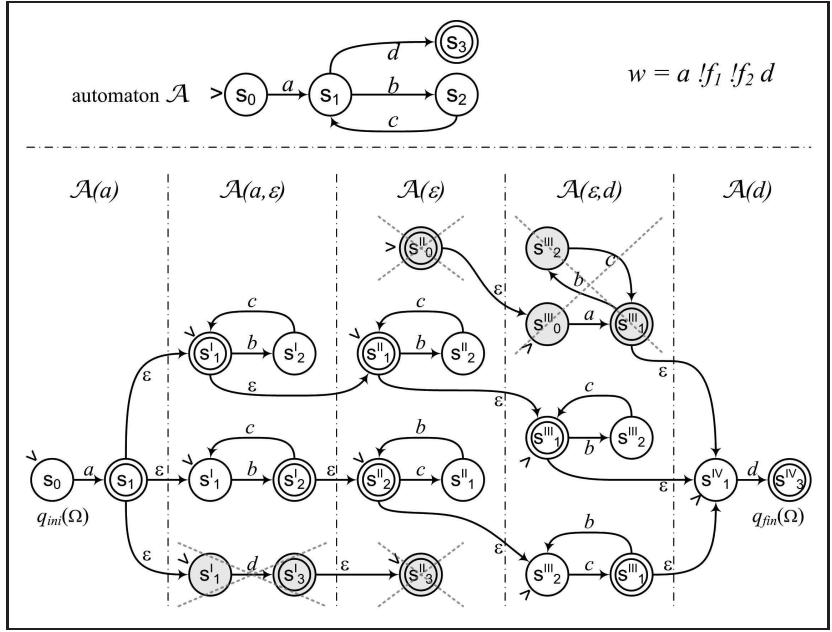}}
\caption{A perfect automaton (construction)}
\label{figPerfectAutomExample}
\end{figure*}

Let $w(\f_n)$ be a kernel string and $\autom{A}$ be an $\nFA$. The
\emph{perfect automaton} w.r.t. $\autom{A}$ and $w$ consists of
several local automata suitably joined together by
$\varepsilon$-transitions. It is denoted by $\Omega(\autom{A},w)$ (or $\Omega$ when it is clear from the context who are $\autom{A}$ and $w$). Algorithm $1$ describes how to build the perfect automaton (assume that any pair of local automata have disjoint sets of states labeled as in $\autom{A}$), while Figure~\ref{figPerfectAutomExample} shows the perfect automaton obtained by a given finite state machine and a kernel string. We say that $\autom{A}$ is \emph{compatible}
with $w$ if the set of all (legal) local automata in $\Omega$ is not
empty after correction steps, or equivalently, if there exists at
least a sound typing. Moreover,
\begin{enumerate}[ \ $\centerdot$]
   \item $Seq(\Omega)$ denotes the set of all the sequences $W_0, X_1, W_1, \ldots, X_n, W_n$ of connected automata in $\Omega$ such that: $W_0$ is an automaton in $\autom{A}(w_0)$, while $W_i$ and $X_i$ are, respectively, in $\autom{A}(w_i)$ and $\autom{A}(w_{i-1},w_i)$ for any $i$ in $[1..n]$;

   \item $Typ(\Omega) = \{(X_i) : W_0,$ $X_1,$ $W_1,$ $\ldots,$ $X_n,$ $W_n \in Seq(\Omega)\}$ is the set containing all different typings $(X_1,$ $\ldots,$ $X_n)$ from any sequence in $Seq(\Omega)$;

   \item $Aut(\Omega_i) = \{X_i : (X_1,\ldots,X_n) \in Typ(\Omega)\}$ is the set of all legal automata in $\autom{A}(w_{i-1}, w_i)$;

   \item $\Omega_i = \cup Aut(\Omega_i)$ is the type obtained by the union of all automata $Aut(\Omega_i)$;

   \item $(\Omega_n)$ is \emph{the typing} for $w$ and $\autom{A}$ obtained from $\Omega$.
\end{enumerate}

Let $(\autom{A}_n)$ be a sequence of automata. We define the \emph{direct extension} of $(\autom{A}_n)$ as the set of string defined as $\langOf{(\autom{A}_n)} = \{ u_1 \ldots u_n \ | ~\mbox{for each $i$}~ \ u_i \in \langOf{\autom{A}_i} \}$.

\begin{figure}[t!]
\vspace{0.5mm} \hrule \vspace{0.3mm} \hrule \vspace{1.75mm}
\textbf{Algorithm 1 } \textsc{PerfectAutomaton}$(w,\mathcal{A})$
\vspace{1.2mm} \hrule \vspace{0.3mm} \hrule \vspace{-4.0mm} \small
\begin{tabbing}
\ \ \= \ \ \ \= \ \ \ \= \ \ \ \= \ \ \ \= \ \ \ \= \ \ \ \= \ \ \ \= \\
\> 1. \> \textbf{Input}: $w(\f_n) = w_0 \f_1 w_1 \ldots \f_n w_n$,
\ $\autom{A}=\langle K,\Sigma, \Delta, s, F\rangle$ \\
\> 2. \> \textbf{Output}: $\Omega(\autom{A},w) \texttt{:=} \emptyset$\\
\> 3. \> \textbf{for each} automaton $W \in \autom{A}(w_{0})$ \textbf{do}\\
\> \> $\triangleright$ add $W$ to $\Omega$\\
\> 4. \> \textbf{for each} $i$ in $[1..n]$ \textbf{do}\\
\> \> $\triangleright$ \> \textbf{for each} automaton $X \in
\autom{A}(w_{i-1}, w_i)$
\textbf{do}\\
\> \> \> a. add $X$ to $\Omega$\\
\> \> \> b. \textbf{for each} automaton $W \in \autom{A}(w_{i-1})$
\textbf{do}\\
\> \> \> \> -- \> \textbf{if} $\textsf{label}(q_{\emph{fin}}(W)) =
\textsf{label}(q_{\emph{ini}}(X))$\\
\> \> \> \> \> $\cdot$ add the transition $(q_{\emph{fin}}(W),
\varepsilon,
q_{\emph{ini}}(X))$ to $\Omega$\\
\> \> \> c. \textbf{for each} automaton $W \in \autom{A}(w_{i})$ \textbf{do}\\
\> \> \> \> -- \> add $W$ to $\Omega$\\
\> \> \> \> -- \> \textbf{if} $\textsf{label}(q_{\emph{fin}}(X)) =
\textsf{label}(q_{\emph{ini}}(W))$\\
\> \> \> \> \> $\cdot$ add the transition $(q_{\emph{fin}}(X),
\varepsilon,
q_{\emph{ini}}(W))$ to $\Omega$\\
\> \texttt{//Correction steps:}\\
\> 5. \> \textbf{for each} automaton $W \in \autom{A}(w_{0})$ \textbf{do}\\
\> \> -- \> \textbf{if} $\textsf{label}(q_{\emph{ini}}(W)) \neq s$ \
\texttt{//if $w_0 = \varepsilon$}\\
\> \> \> $\cdot$ remove $W$ from $\Omega$ \ \texttt{//it is illegal}\\
\> 6.\> merge all automata in $\Omega$ being in $\autom{A}(w_{0})$
according to their\\
\> \> labels and use the (unique) initial state as initial state for $\Omega$\\
\> 7. \> \textbf{for each} automaton $W \in \autom{A}(w_{n})$ \textbf{do}\\
\> \> -- \> \textbf{if} $\textsf{label}(q_{\emph{fin}}(W)) \in F$\\
\> \> \> $\cdot$ $F(\Omega) = F(\Omega) \cup \{q_{\emph{fin}}(W)\}$\\
\> \> \> \textbf{else} \ \texttt{//if $w_n = \varepsilon$}\\
\> \> \> $\cdot$ remove $W$ from $\Omega$ \ \texttt{//it is illegal}\\
\> 8. \> \textbf{for each} automaton $A \in \Omega$ \textbf{do}\\
\> \> -- \> \textbf{if} \> \texttt{(}there is no path from
$q_{\emph{ini}}(\Omega)$ to $A$ \textbf{or}\\
\> \>             \> \> \texttt{\ }there is no path from $A$ to any
final
state of $\Omega$\texttt{)}\\
\> \> \> $\cdot$ remove $A$ from $\Omega$ \ \texttt{//it is illegal}\\
\end{tabbing} \normalsize
\vspace{-5.0mm}
\hrule
\vspace{0.3mm}
\hrule
\vspace{0.5mm}
\end{figure}

\begin{lemma}\label{OmegaLeqA}
For any $\nFA$ $\autom{A}$, then $\Omega \leq \autom{A}$ holds. On the other hand, $\autom{A} \leq \Omega$ does not hold in general.
\end{lemma}
\begin{proof}\ONLINE{
Given a string $u$ in $\langOf{\Omega}$, then there exists a
sequence $(\tau_{2n+1})$ of automata in $Seq(\Omega)$ accepting $u$
and expressible as $\autom{A}(s,q_0),$ $\autom{A}(q_0,s_1),$
$\autom{A}(s_1,q_1),$  $\ldots,$ $\autom{A}(q_{n-1},s_n),$
$\autom{A}(s_n,q_n)$ for some states $q_0, s_1, q_1 \ldots, s_n,
q_n$. Moreover, by definition of direct extension, for each
string $u_0 \sigma_1 u_1 \dots \sigma_n u_n$ in
$\langOf{(\tau_{2n+1})}$ we have that $u_0 \in
\langOf{\autom{A}(s,q_0)}$, $\sigma_i \in
\langOf{\autom{A}(q_{i-1},s_i)}$ and $u_i \in
\langOf{\autom{A}(s_i,q_i)}$, for each $i$ in $[1..n]$. But, by definition of local automata, the following sequence of transitions (each of which belongs to $\Delta^*$):
$(s,w_0,q_0), \ (q_0,\sigma_1,s_1), \
(s_1,w_1,q_1), \ \ldots,$ \ $(q_{n-1},\sigma_n,s_n), \ (s_n,w_n,q_n)$,
\noindent where $q_n \in F $, is also derivable by $\autom{A}$.

For the second part of the proof consider the string $w = a \f c$ and the $\dRE$ $abc+d$.}
\end{proof}

\begin{lemma}\label{XdiOmegaSound}
Let $w(\f_n)$ be a string compatible with an $\nFA$ $\autom{A}$. Any typing in $Typ(\Omega)$ is sound for $w$ and $\autom{A}$.
\end{lemma}
\begin{proof}\ONLINE{
Given any typing $(X_n)$ in $Typ(\Omega)$, by definition, there is a
sequence $(\tau_{2n+1})$ of automata such that $X_i = \tau_{2i}$ for
each $i$ in $[1..n]$. By Lemma \ref{OmegaLeqA}
$(\tau_{2n+1}) \leq \autom{A}$ holds. Moreover as, by definition,
the extension of $w(X_n)$ is $ \langOf{w(X_n)} = \{w_0 \sigma_1 w_1
\ldots \sigma_n w_n : \sigma_i \in \langOf{X_i},  \ 1 \leq i \leq
n\}.$ Then $w(X_n) \leq (\tau_{2n+1})$ as well since all strings
$w_0,\ldots,w_n$ are accepted by $\tau_1,\tau_3\ldots,\tau_{2n+1}$,
respectively, by definition of local automata induced by a single
string. Therefore, $w(X_n) \leq \autom{A}$.}
\end{proof}

\begin{theorem}\label{AcceptedByAAndAi}
Let $w(\f_n)$ be a kernel string compatible with a given $\nFA$ $\autom{A}$, and $(\tau_n)$ be a sound typing for them. Then, both $w(\tau_n) \leq \Omega$ and $(\tau_n) \leq (\Omega_n)$ hold.
\end{theorem}
\begin{proof}
Since $(\tau_n)$ is sound for $w$ and $\autom{A}$, then $w(\tau_n)
\leq \autom{A}$ holds. In particular, for each string $\chi =
w_0 \sigma_1 w_1 \ldots \sigma_n w_n$ in $\langOf{w(\tau_n)}$,
where each $\sigma_i \in \langOf{\tau_i}$, there is a
sequence of states $q_0, s_1, q_1 \ldots, s_n, q_n$ proving the
membership of $\chi$ in $\langOfAutom{A}$ by the following sequence
of transitions
$(s,w_0,q_0) \in \Delta^*, \ (q_0,\sigma_1,s_1) \in \Delta^*, \
(s_1,w_1,q_1) \in \Delta^*, \
\ldots, \ (q_{n-1},\sigma_n,s_n) \in \Delta^*, \ (s_n,w_n,q_n) \in
\Delta^*$
\noindent where $q_n \in F$ holds as well. But, this means that the
sequence $\autom{A}(s,q_0)$, $\autom{A}(q_0,s_1)$,
$\autom{A}(s_1,q_1)$, \ldots, $\autom{A}(q_{n-1},$ $s_n)$,
$\autom{A}(s_n,$ $q_n)$ of automata belongs to $Seq(\Omega)$, so
$w(\tau_n) \leq \Omega$ holds. Moreover, since each
$\autom{A}(q_{i-1},s_i) \in Aut(\Omega_i)$, it follows that
$\tau_i \leq \Omega_i$ for each $i$, that is $(\tau_n) \leq
(\Omega_n)$.
\end{proof}

\begin{corollary}
Let $w(\f_n)$ be a kernel string compatible with a given $\nFA$ $\autom{A}$, and $(\tau_n)$ be a local typing for them. Then, $w(\tau_n) \equiv \Omega \equiv  \autom{A}$ holds.
\end{corollary}
\begin{proof}
By Lemma \ref{OmegaLeqA} and Theorem \ref{AcceptedByAAndAi}.
\end{proof}

\begin{theorem}\label{OmegaPerfect}
Let $w(\f_n)$ be a kernel string and $\autom{A}$ be an $\nFA$ compatible with $w$. There is a perfect typing for $w$ and $\autom{A}$ if and only if $w(\Omega_n) \equiv \autom{A}$. If so, the perfect typing is exactly $(\Omega_n)$.
\end{theorem}
\begin{proof}
($\Rightarrow$) \emph{if there is a perfect typing for $w$ and
$\autom{A}$ then $w(\Omega_n) \equiv \autom{A}$}. If $w$ and
$\autom{A}$ admit a perfect typing, say $(\tau_n)$, then (as it is
also sound), by Theorem \ref{AcceptedByAAndAi}, $(\tau_n) \leq
(\Omega_n)$. Suppose that $(\tau_n) < (\Omega_n)$ held. There would
be (at least) an $i$ in $[1..n]$ such that $\tau_i <
\Omega_i$. In other words, there would be an automaton $\tau_i' \in
Aut(\Omega_i)$ accepting some strings rejected by $\tau_i$. Consider
the typing $(\tau_n') \in Typ(\Omega)$ containing $\tau_i'$ in
position $i$. By Lemma \ref{XdiOmegaSound}, $(\tau_n')$ is sound and
then $\tau_i' \leq \tau_i$, by definition. But this is a
contradiction. Therefore $(\tau_n) \equiv (\Omega_n)$ and then
$w(\Omega_n) \equiv \autom{A}$, as $(\tau_n)$ is also local.

\medskip

\noindent ($\Leftarrow$)\emph{if $w(\Omega_n) \equiv \autom{A}$ then there is a perfect typing for $w$ and $\autom{A}$.} This is true since
$(\Omega_n)$ is local and because, by Theorem
\ref{AcceptedByAAndAi}, $(\tau_n) \leq (\Omega_n)$ for any sound
typing $(\tau_n)$.
\end{proof}

The following two examples show that if there exists a local typing
$(\tau_n)$ for $w$ and $\autom{A}$, then $(\tau_n) < (\Omega_n)$
might hold. This can happen even if $(\tau_n)$ is a unique maximal
local.
\begin{example}
Consider the string $w = a \ \f_1 \ c \ \f_2 \ e$, and the regular
expression $\tau = abccde$ compatible with $w$. Clearly, the typing
$(b,cd)$ is local (sound and complete) for $w$ and $\tau$ because
$w(b,cd) \equiv \tau$. Nevertheless, $(\Omega_2)$ = $ (bc?$, $c?d)$
is (strictly) greater then $(b,cd)$ since $\langOf{bc?} = \{b,$
$bc\} \supset \{b\}$ and $\langOf{c?d} =\{d,$ $cd\}\supset \{cd\}$. \qed
\end{example}
\begin{example}
Let $w = a \ \f_1 \ \f_2 \ d$ be a kernel string and $\tau$ be
the regular expression $a(bc)^*d$. Clearly, the typing
$((bc)^*,(bc)^*)$ is \emph{local} (also unique maximal local but not
\emph{perfect}). But, as consequence of the construction of perfect automaton, we have: $Aut(\Omega_1) = \{(bc)^*, (bc)^*b\}$ and
$Aut(\Omega_2) = \{(bc)^*, c(bc)^*\}$. Consequently, $\Omega_1
\equiv ((bc)^*b?)$ and $\Omega_2 \equiv (c?(bc)^*)$ do not
represent a sound (and hence local) typing since they allow strings
such as $abccbcd$ or $abcbbcd$ that are not accepted by $\tau$. \qed
\end{example}

The following example shows that even if there is no local typing
for $w$ and $\tau$, then $\Omega \equiv \tau$ may hold.
\begin{example}
Let a $\tau$ be the regular expression $ab+ba$ and $w = \f_1 \f_2$.
There are two \emph{sound typings}: $(a, b)$ and $(b, a)$, but there
is no local typing. However, $\Omega \equiv \tau$. \qed
\end{example}

We can now use the perfect automata construction to characterize the complexity of $\perf{\nFA}$.  We use the next lemma:

\begin{lemma}\label{omegaPolLemma}
Let $w(\f_n)$ be a kernel string and $\autom{A}$ be a $k$-state $\nFA$. The algorithm for building the perfect automaton $\Omega(\autom{A},w)$ works in polynomial time.
\end{lemma}
\begin{proof}\ONLINE{
Any set $\autom{A}(w_i)$ or $\autom{A}(w_{i-1},w_i)$ contains at
most $k^2$ automata each of which having size $\mathcal{O}(k)$.
Therefore, the number of macro-iterations of the algorithm are
$\mathcal{O}(nk^2)$, while the size of $\Omega$ is
$\mathcal{O}(nk^3)$. For each $w_i$, the sets
$Ini(\autom{A},w_i)$ and $Fin(\autom{A},w_i)$ can be obtained in
nondeterministic logarithmic space (thus in polynomial time) because
for any pair of states $q_1, q_2$ in $\autom{A}$, we check if the
string $w_i$ is in the language $\langOf{\autom{A}(q_1,q_2)}$.
Finally, all the automata in $\autom{A}(w_i)$ and
$\autom{A}(w_{i-1},w_i)$ are nothing else but different copies of
$\autom{A}$ having different initial and finial states.}
\end{proof}

Now, we have:

\begin{theorem}\label{PerfNfaPspaceComplete}
$\perf{\nFA}$ is in $\class{PSPACE}$. So it is also
$\classc{PSPACE}$ by Theorem \ref{check-nfa-hard}.
\end{theorem}
\begin{proof}
Let $w(\f_n)$ be a kernel string, $\tau$ be an $\nFA$, and $(\tau_n)$ be a typing. Construct the perfect automaton $\Omega(\tau,w)$. By Lemma \ref{omegaPolLemma}, $\Omega$ can be built in polynomial time w.r.t. $|\tau| + \|w\|$. Then, check in polynomial space if $w(\Omega_n) \equiv \tau \equiv w(\tau_n)$.
\end{proof}

And w.r.t. finding a perfect typing (if it exists), we have:

\begin{theorem}\label{nfa-pspace-complete}
$\eperf{\nFA}$ is in $\class{PSPACE}$. So it is also
$\classc{PSPACE}$ by Theorem \ref{nfa-locHard}.
\end{theorem}
\begin{proof}
Let $\langle \tau, w(\f_n)\rangle$ be a (string) design. Construct
the perfect automaton $\Omega(\tau,w)$. By Lemma
\ref{omegaPolLemma}, $\Omega$ can be built in polynomial time w.r.t.
$|\tau| + \|w\|$. Then, check if $w(\Omega_n) \equiv \tau$, which is
feasible in polynomial space.
\end{proof}

\subsection{Additional properties}\label{additionalProperties}

We now show how to exploit perfect automaton properties to find (maximal) sound typings when a design does not allow any perfect. Clearly, this technique can be used for seeking (maximal) local typings as well.
Let $w(\f_n)$ be a kernel string and $\autom{A}$ be an $\nFA$-type compatible with $w$. All the automata belonging to $Aut(\Omega_i)$ can be decomposed in at most $2^{|Aut(\Omega_i)|} - 1$ different automata such that there are no two of them accepting the same string. In particular, this new set is denoted by $Dec(\Omega_i)$ and defined as follows:
\[
Dec(\Omega_i) = \{\automInters{\automSet{A}_1} \mymid
\automUnion{\automSet{A}_2}  \ : \ \emptyset \neq \automSet{A}_1
\subseteq Aut(\Omega_i), \ \automSet{A}_2 =
Aut(\Omega_i) \mysetminus \automSet{A}_1 \}
\]
\noindent An example for three automata is given in
Figure~\ref{figSetPartition}. Finally, $Dec(\Omega) = \{(D_1,\ldots,D_n): D_i \in Dec(\Omega_i)\}$ is the set of all different typings from $Dec(\Omega_1)
\times \ldots \times Dec(\Omega_n)$. Given a typing $(\tau_n)$, we
say that $(\tau_n) \in Dec(\Omega)$ if there exists a sequence
$(D_n) \in Dec(\Omega)$ such that $\tau_i \equiv D_i$, for
each $i$.
\begin{figure}[htbp] \centering
\fbox{\includegraphics[scale=1]{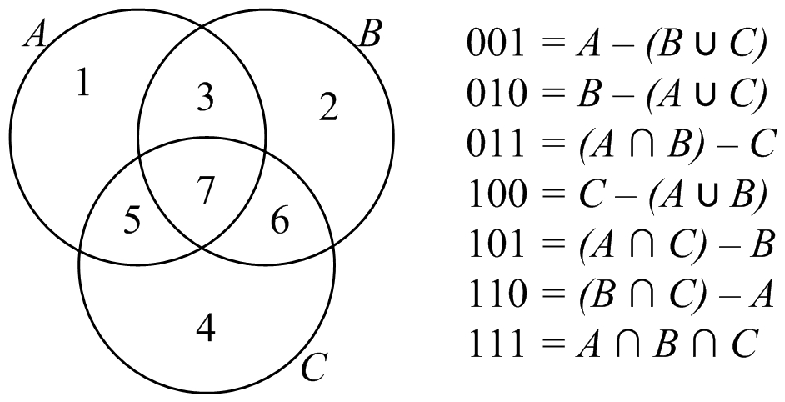}}
\caption{Partitioning of (three) sets and enumeration of the parts}
\label{figSetPartition}
\end{figure}

Given a type $\tau \leq \Omega_i$ for some $i$ in $[1..n]$, $Dec(\tau,i) = \{\tau \cap \tau' : \tau' \in
Dec(\Omega_i)\}$ denotes the partition of $\tau$, namely
$\automUnion{Dec(\tau,i)} \equiv \tau$, obtained by its projection
on $Dec(\Omega_i)$. Let $(\tau_n)$ be any typing for a kernel
string $w(\f_n)$. Given a string $u \in \calL^*$ and an $i$ in $[1..n]$, then $(\tau_n)_{[\tau_i | u]}$ denotes the new
typing obtained from $(\tau_n)$ by replacing $\tau_i$ with the
minimal $\dFA$ accepting only the string $u$. In particular
$\langOf{w(\tau_n)_{[\tau_i | u]}}$ is defined as $\{w_0 \sigma_1
w_1 \ldots \sigma_n w_n \ : \ \sigma_i = u, \sigma_j \in
\langOf{\tau_j} \ \forall j \neq i\}$ and clearly,
\begin{displaymath}
w(\tau_n) \ \ \equiv \ \bigcup_{u \in \langOf{\tau_i}}
w(\tau_n)_{[\tau_i | u]}
\end{displaymath}
We now define an extension of $(\tau_n)$ as the new typing
obtained from $(\tau_n)$ by replacing $\tau_i$ with the new type
$(\tau_i \cup \tau)$, and denoted by $(\tau_n)_{[\tau_i \cup
\tau]}$. In particular,
\begin{displaymath}
w(\tau_n)_{[\tau_i \cup \tau]} \ \ \ \equiv \bigcup_{u \in
\langOf{\tau_i \cup \tau}} w(\tau_n)_{[\tau_i | u]}
\end{displaymath}
\noindent Clearly, if $\tau \leq \tau_i$, then $(\tau_n) \equiv
(\tau_n)_{[\tau_i \cup \tau]}$. Otherwise $(\tau_n) <
(\tau_n)_{[\tau_i \cup \tau]}$.

\begin{definition}\label{DefGradExten}
A type $\tau$ \textbf{extends} another type $\tau'$ if $\langOf{\tau} \mysetminus \langOf{\tau'} \neq \emptyset$ holds. Moreover, the extension is called \textbf{partial} or \textbf{total} depending on whether $\langOf{\tau} \cap \langOf{\tau'} \neq \emptyset$ or not, respectively. \qed
\end{definition}

\begin{lemma}\label{SoundExtension}
Let $D = \langle \autom{A}, w \rangle$ be an $\nFA$-design,
$(\tau_n)$ be a consistent sound typing for $D$, and $\tau \in Dec(\Omega_i)$ be an $\nFA$ belonging to the decomposition of $\Omega_i$, for some $i$ in $[1..n]$. If $\tau$ partially extends $\tau_i$, then the extension $(\tau_n)_{[\tau_i \cup \tau]}$ of $(\tau_n)$ is still sound.
\end{lemma}
\begin{proof}\ONLINE{
By definition \ref{DefGradExten}, $\langOf{\tau}$ contains at least a string that does not belong to $\langOf{\tau_i}$ but also a string, say  $u'$, accepted by both $\tau_i$ and $\tau$. In order to prove the statement, we show that $w(\tau_n)_{[\tau_i | u]} \leq \Omega$ holds for each $u \in \langOf{\tau} \mysetminus \langOf{\tau_i}$ (recall that, by Lemma \ref{OmegaLeqA}, $\Omega \leq \autom{A}$).

Since, by Theorem \ref{AcceptedByAAndAi}, $\tau_i \leq \Omega_i$, then there is a nonempty set $\automSet{A} \subseteq Aut(\Omega_i)$ containing all-and-only the automata accepting $u'$. Clearly, since $\tau \in Dec(\Omega_i)$ and $u' \in \langOf{\tau}$, then $\tau$ is also in $Dec(\tau',i)$ for each $\tau' \in \automSet{A}$. This means that each string $u \in \langOf{\tau} \mysetminus \langOf{\tau_i}$ is accepted by all-and-only the automata in $\automSet{A}$ as well. By Theorem \ref{AcceptedByAAndAi}, $w(\tau_n) \leq \Omega$, and in particular $w(\tau_n)_{[\tau_i | u']} \leq \Omega$, as $u' \in \langOf{\tau_i \cap \tau}$. In other words, any string in $\langOf{w(\tau_n)_{[\tau_i | u']}}$ is accepted by (at least) a sequence of automata in $Seq(\Omega)$. Finally, as both $u$ and $u'$ are recognized by all-and-only the automata in $\automSet{A}$, then each string $w_0 \sigma_1 w_1 \ldots \sigma_n w_n$ in $\langOf{w(\tau_n)_{[\tau_i | u]}}$ (with $\sigma_i = u$) has a twin in $\langOf{w(\tau_n)_{[\tau_i | u']}}$ (with $\sigma_i = u'$) and both of them are accepted by exactly the same sequences in $Seq(\Omega)$.}
\end{proof}

\begin{theorem}\label{NEXPGuess}
Let $(\tau_n)$ be a maximal typing for a kernel string $w(\f_n)$ and an $\nFA$ $\autom{A}$ compatible with $w$. Then for each $i$, $Dec(\tau_i,i) \subseteq Dec(\Omega_i)$.
\end{theorem}
\begin{proof}
Let $i$ be an index arbitrarily fixed in $[1..n]$. As $(\tau_n)$ is maximal then, by definition, it is sound and, by Theorem \ref{AcceptedByAAndAi}, $\tau_i \leq \Omega_i$. Let $D_i$ be a copy of $Dec(\Omega_i)$. Then $\tau_i \leq \cup D_i$. Remove now, from $D_i$, each automata $\tau_{D_i}$ (if any) such that $\langOf{\tau_{D_i}} \cap \langOf{\tau_i} = \emptyset$. Still, $\tau_i \leq \cup D_i$ holds. Hence, consider the two possible (and alternative) cases: $(1)$ $\tau_i \equiv \cup D_i$, or $(2)$ $\tau_i < \cup D_i$.
\noindent In the first case the theorem is already proved. While, in the latter case, there is (at least) an  automaton $\tau \in D_i$ that partially extends $\tau_i$ entailing relation $(\tau_n) < (\tau_n)_{[\tau_i \cup \tau]}$. But since $(\tau_n)_{[\tau_i \cup \tau]}$ is still sound (see Lemma \ref{SoundExtension}), then there is a contradiction because $(\tau_n)$ is assumed to be maximal.
\end{proof}

We are now ready to prove a main results of the section. In our original paper, we showed a 2-$\class{EXPSPACE}$ upper bound for $\eloc{\nFA}$ and $\eml{\nFA}$. This was improved to $\class{EXPSPACE}$ in \cite{MartensNiewerthSchwentick10}. We present here an alternative proof of that results using the previous decomposition.

\begin{theorem}\label{ThmExpspace}
Problems $\eloc{\nFA}$ and $\eml{\nFA}$ are in $\class{EXPSPACE}$.
\end{theorem}
\begin{proof}
By Lemma \ref{SoundExtension} and Theorem \ref{NEXPGuess}, if an $\nFA$-design $D = \langle \tau, w \rangle$ admits a (maximal) local typing, say $(\tau_n)$, then for each $\tau_i$ there exists a subset of $Dec(\Omega_i)$, say $D_i$, such that $\cup D_i \equiv \tau_i$.

Let $m$ be the number of states of $\tau$, and $\nu + n$ be the length of $w$ where $n$ is clearly the number of functions and $\nu$ is the length of the non-function symbols in $w$. By definition, for each $i$ in $[1..n]$, each automaton in $Aut(\Omega_i)$ has size at most $m$ and the cardinality of $Aut(\Omega_i)$ is at most $m^2$. Thus, the cardinality of $Dec(\Omega_i)$ is no more than $2^{m^2}$, as well as the cardinality of $D_i$. In the worst case, an automaton in $Dec(\Omega_i)$ is obtained as $\automInters{\automSet{A}_1} \mymid
\automUnion{\automSet{A}_2}$ where both $|\automSet{A}_1| = |\automSet{A}_2| = \mathcal{O}(m^2)$. So, the size of $\automInters{\automSet{A}_1}$ is no more than $(m^2)^{m^2}$ \cite{HolzerKutrib02}, that is clearly lower than $2^{m^3}$. The size of $\automUnion{\automSet{A}_2}$ is at most $m^3$ \cite{HolzerKutrib02}. Now, for computing $\automInters{\automSet{A}_1} \mymid \automUnion{\automSet{A}_2}$ we perform the following intersection $(\automInters{\automSet{A}_1}) \cap
(\overline{\automUnion{\automSet{A}_2}})$. The complement of $\automUnion{\automSet{A}_2}$ may have $2^{m^3}$ states \cite{HolzerKutrib02}. Finally $\automInters{\automSet{A}_1} \mymid \automUnion{\automSet{A}_2}$ require no more than $2^{2m^3}$ states, and the size of $\cup D_i \equiv \tau_i$ is at most $2^{2m^3} * 2^{m^2}$ being clearly $2^{\mathcal{O}(m^3)}$.

Now, we are ready for computing the size of the $\nFA$ $w(\tau_n)$. It is exactly $\nu + n*2^{\mathcal{O}(m^3)}$. So, for deciding whether $w(\tau_n) \equiv \tau$ we need no more than exponential space w.r.t. the input size $\nu + n + m$. The only problem we still have is that we do not know a priori how to choose $D_i$. There are $2^{2^{m^2}}$ possible subsets. But as $\class{NEXPSPACE} = \class{EXPSPACE}$ (by Savitch's theorem), then we can simply guess each $D_i$.

\medskip

About $\eml{\nFA}$, we must find a maximal $D_i$. But in $\class{EXPSPACE}$ we can still guess the sequence $D_1,\ldots,D_n$ and prove (by Theorem \ref{NEXPGuess}), for each $D_i$, that none of the automata in $Dec(\Omega_i) \mysetminus D_i$ can be added to $D_i$ because the resulting typing would loose its soundness. After the guess, the number of checks (each of which may require exponential space) is at most $n*2^{m^2}$.
\end{proof}

\section{Complexity for trees}\label{ComplexityForTrees}

Based on Theorem \ref{ThmExpspace}, we now obtain complexity bounds for the tree problems. This completes results obtained in \cite{AbiteboulGottlobManna09,MartensNiewerthSchwentick10} on this topic. The next result first appeared in \cite{AbiteboulGottlobManna09}. However, the sketch of proof given there was not correct.  A proof was then presented in \cite{MartensNiewerthSchwentick10}. We next present a new proof based on perfect automata.

\begin{theorem}\label{nfa-ml-pspace}
$\ml{\nFA}$ is in $\class{PSPACE}$ (so the problem is $\classc{PSPACE}$).
\end{theorem}
\begin{proof}
Let $D = \langle \tau, w \rangle$ be an $\nFA$-design, and
$(\tau_n)$ be a $D$-consistent typing. First of all, we check if $(\tau_n)$ is local (and we have already proved that $\loc{\nFA}$ is doable in $\class{PSPACE}$). If so, then $\bar{\tau} \cap w(\tau_n) \equiv \emptyset$ (where $\bar{\tau}$ is the $\nFA$ of possibly exponential size accepting the complement of language $\langOf{\tau}$). Subsequently, we check if $(\tau_n)$ is \emph{not} maximal. In particular, by Lemma \ref{SoundExtension} and Theorem \ref{NEXPGuess}, $(\tau_n)$ is not maximal if there is an $\nFA$ $\autom{A} \in Dec(\Omega_i)$ for some $i$ in $[1..n]$ such that at least one of the following is true:
\begin{enumerate}[ \ $\centerdot$]
  \item $\autom{A}$ totally extends $\tau_i$ and $w(\tau_n)_{[\tau_i \cup \autom{A}]}$ is still sound, namely $\langOfAutom{A} \cap  \langOf{\tau_i} = \emptyset$ and $\bar{\tau} \cap w(\tau_n)_{[\tau_i \cup \autom{A}]} \equiv \emptyset$.

  \item $\autom{A}$ partially extends $\tau_i$, namely $\langOf{\autom{A}} \mysetminus \langOf{\tau_i} \neq \emptyset$ and $\langOf{\autom{A}} \cap \langOf{\tau_i} \neq \emptyset$;
\end{enumerate}
So we proceed as follows:
\begin{enumerate}
  \item Guess an index $i$, and a nonempty set of automata $\automSet{A}_1 \subseteq Aut(\Omega_i)$

  \item Compute $\automSet{A}_2 = Aut(\Omega_i) \mysetminus \automSet{A}_1$

  \item Let $\autom{A}$ denote the automaton $\automInters{\automSet{A}_1} \mymid \automUnion{\automSet{A}_2}$ (we do not really build it);

  \item If $\langOfAutom{A} \cap  \langOf{\tau_i} = \emptyset$ then,
    \begin{enumerate}[$\centerdot$]
      \item if $\bar{\tau} \cap w(\tau_n)_{[\tau_i \cup \autom{A}]} \equiv \emptyset$, then $(\tau_n)$ is not maximal
      \item else if $\langOf{\autom{A}} \mysetminus \langOf{\tau_i} =  \langOf{\autom{A}} \cap \langOf{\overline{\tau_i}}  \neq \emptyset$, then $(\tau_n)$ is not maximal
    \end{enumerate}
\end{enumerate}
Observe that even if $\autom{A}$, $\bar{\tau}$, or $\bar{\tau_i}$ may be exponential in size, we only use them for
\emph{intersection nonemptiness} or \emph{intersection emptiness} problems that are both $\classc{NL}$ problems \cite{Jones75}. Intuitively, we could avoid the materialization of such automata with ``on-the-fly'' constructions. Hence, an $\class{NL}$ algorithm on a non-materialized (single) exponential automaton leads to $\class{PSPACE}$. More formally, we consider alternating finite state machines $\aFAs$ (for more details see \cite{RozenbergSalomaaYu97,FellahJurgensenYu90}). We do not completely define them but we just recall what we need:
\begin{enumerate}[ \ $\centerdot$]
  \item given an $\aFA$ $\autom{A}$, deciding whether $\langOfAutom{A} = \emptyset$ is $\classc{PSPACE}$;

  \item Any $\nFA$ is trivially a special kind of $\aFA$;

  \item Given two $\aFAs$ $\autom{A}$ and $\autom{A}'$, a new $\aFA$ for $\overline{\autom{A}}$, $\autom{A} \cup \autom{A}'$, and $\autom{A} \cap \autom{A}'$, can be constructed in polynomial time and its size is linear.
\end{enumerate}
Finally, we observe that all the above emptiness decisions deal with $\nFAs$ of polynomial size and can be checked in $\class{PSPACE}$ as well as the nonemptiness decisions as $\class{PSPACE}$ is closed under complement.
\end{proof}

Now, we show how to reduce locality problems on boxes to locality problems on strings.

\begin{definition}
Let $D = \langle \tau, B \rangle$ be an $\mathcal{R}$-design where $B = B_0 \f_1 B_1 \ldots \f_n B_n$ is a kernel box. Consider the $k^{th}$ sequence of strings $w_0, \ldots, w_n$ (with $1 \leq k \leq |B_0| * \ldots * |B_n|$) built from $B_0, \ldots, B_n$ by varying $w_i$ among the strings in $\langOf{B_i}$ (in some fixed order) for each $i \in [0..n]$. We denote by $D^k = \langle \tau, w^k \rangle$ the $k^{th}$ $\mathcal{R}$-design built from $D$ where $w^k(\f_n)$ is the kernel string $w_0 \f_1 w_1 \ldots \f_n w_n$. \qed
\end{definition}

\begin{lemma}
Let $D = \langle \tau, B \rangle$ be an $\mathcal{R}$-design and $(\tau_n)$ be a D-consistent typing. We have that:
\begin{enumerate}[$(1)$]
  \item If $(\tau_n)$ is local for $D$, then it is sound for each $D^k$;

  \item If $(\tau_n)$ is sound for each $D^k$, then it is sound for $D$ as well.
\end{enumerate}
\end{lemma}
\begin{proof}
$(1)$: If $(\tau_n)$ is local for $D$, then $B_0 \tau_1 B_1 \ldots \tau_n B_n \equiv \tau$. This means that $w_0 \tau_1 w_1 \ldots \tau_n w_n \leq \tau$ for each $w_i \in \langOf{B_i}$. Thus, $(\tau_n)$ is sound for each $D^k$.

\medskip

\noindent $(2)$: If for each design $D^k$ we have that $w_0 \tau_1 w_1 \ldots \tau_n w_n \leq \tau$ holds, then $B_0 \tau_1 B_1 \ldots \tau_n B_n \leq \tau$ as well.
\end{proof}

A direct consequence of the above theorem is that if a typing is not sound for some $D^k$, then it can not be local for $D$. So a local typing candidate for $D$ is a typing being sound for each $D^k$. Now suppose that $(\tau_n)$ is a maximal sound typing for $D^{k_1}$ but it is not sound for $D^{k_2}$. This means that at least one $\langOf{\tau_i}$ contains some extra string such that $\langOf{w^{k_2}(\tau_n)}$ is not fully contained in $\langOf{\tau}$. So we could remove such strings to obtain a typing sound for both $D^{k_1}$ and $D^{k_2}$ but not maximal for $D^{k_1}$ any more. So we can guess a maximal sound typing for each $D^{k}$ and then, remove the exceeding strings. This is equivalent to keeping the componentwise intersection of these maximal typings. Let $\beta = |B_0| * \ldots * |B_n|$, we should build $\beta$ (it is an exponential number) perfect automata. For each $i$ in $[1..n]$ we should consider the sets of automata $Aut(\Omega_i^1), \ldots, Aut(\Omega_i^\beta)$ and from these the respective decompositions $Dec(\Omega_i^1), \ldots, Dec(\Omega_i^\beta)$. Now we can guess $\beta$ subsets $D_i^1, \ldots, D_i^\beta$ and finally compute $\tau_i$ as $(\cup D_i^1) \cap  \ldots \cap (\cup D_i^\beta)$. But this is equivalent to consider directly $Aut(\Omega_i) = Aut(\Omega_i^1) \cup \ldots \cup Aut(\Omega_i^\beta)$, compute the decomposition $Dec(\Omega_i)$ and guess a subset $D_i$ from $Dec(\Omega_i)$. This is much more convenient because $Aut(\Omega_i)$
contains at most a quadratic number of automata w.r.t. the states of $\tau$. Now, we show how to extend the construction of $\Omega$ to a box-design for obtaining this new $Aut(\Omega_i)$. Let $\autom{A}$ be an $\nFA$ and $B(\f_n)$ a kernel box, we have that:
\begin{enumerate}[ \ $\centerdot$]
  \item $Ini(\autom{A},B_i) = \{q_i \in K: \ \exists q_f \in K \ \textrm{s.t.} \ (q_i,w,q_f) \in \Delta^*, w \in \langOf{B_i} \}$
  \item $Fin(\autom{A},B_i) = \{q_f \in K: \ \exists q_i \in K \ \textrm{s.t.} \ (q_i,w,q_f) \in \Delta^*, w \in \langOf{B_i} \}$
  \item $\autom{A}(B_{i-1},B_{i}) = \{\autom{A}(q_i,q_f) : q_i \in Fin(\autom{A},B_{i-1}), \ q_f \in Ini(\autom{A},B_{i})\}$
\end{enumerate}
$Aut(\Omega_i)$ is the set of all legal automata in $\autom{A}(B_{i-1},B_{i})$ as for string. Note that, due to the structure of each $B_i$, it is very easy to build $Ini(\autom{A},B_i)$ and $Fin(\autom{A},B_i)$ without enumerating all the strings in $\langOf{B_i}$.

\begin{theorem}\label{Thme-loc-ml-B-Expspace}
Problems $\elocB{\nFA}$ and $\emlB{\nFA}$ are in $\class{EXPSPACE}$.
\end{theorem}
\begin{proof}
Let $D = \langle \tau, B \rangle$ be an $\nFA$-design where $B$ is a kernel box. We guess, for each $i$, a subset of automata in $Dec(\Omega_i)$, the decomposition of the new set $Aut(\Omega_i)$ built as above. Thus, we check if it is a (maximal) local typing for $D$ as made in the proof of Theorem \ref{ThmExpspace}.
\end{proof}

\begin{corollary}\label{Thme-loc-ml-2Expspace}
$\eloc{\xEDTD{\nFA}}$ and $\eml{\xEDTD{\nFA}}$ are in 2-$\class{EXPSPACE}$.
\end{corollary}

\begin{proof}
Let $D = \langle \tau, T \rangle$ be an $\xEDTD{\nFA}$-design. We build from $\tau$ its equivalent normalized version $\tau^d$ that, after all, is a $\xEDTD{\dFA}$ of exponential size. So the oracle machine discussed in Corollary \ref{edtd2strReduction} actually works in $\class{NEXPTIME}^\mathcal{C}$ where $\mathcal{C}$ is the complexity class of solving $\elocB{\dFA}$ (or $\emlB{\dFA}$). By Theorem \ref{Thme-loc-ml-B-Expspace}, both of these problems are in $\class{EXPSPACE}$. Thus, the whole algorithm works in 2-$\class{EXPSPACE}$. (Note that, $\class{EXPSPACE}$ is the best known upper bound even for $\eloc{\dFA}$ \cite{MartensNiewerthSchwentick10}.)
\end{proof}

The following analysis makes use of a technique introduced in \cite{MartensNiewerthSchwentick10} for building the perfect automaton for $\dFA$-designs.

\begin{definition}
Let $D = \langle \tau, B \rangle$ be an $\mathcal{R}$-design where $B = B_0 \f_1 B_1 \ldots \f_n B_n$ is a kernel box. Together with $D^k$ we consider the string design $\hat{D}^k$ defined as follows. Let $\hat{\Sigma} = \Sigma \uplus \{\sigma_0,\ldots,\sigma_n\}$ be an extension of $\Sigma$ and $\sigma(\f_n) = \sigma_0 \f_1 \sigma_1 \ldots \f_n \sigma_n$ be the kernel string built by combining the new symbols with the functions of $B$. We denote by $\hat{D}^k = \langle \hat{\Omega}^k, \sigma \rangle$ the $k^{th}$ $\dFA$-design built from $D$ where $\hat{\Omega}^k = \hat{\Omega}^k(\tau,w^k)$ is the perfect automaton built as described in \cite{MartensNiewerthSchwentick10}. \qed
\end{definition}

The following lemma is a direct consequence of the definition of $\hat{\Omega}$ in \cite{MartensNiewerthSchwentick10}.

\begin{lemma}\label{lemPerfDfa}
A typing $(\tau_n)$ is sound for $D^k$ iff it is sound for $\hat{D}^k$.
\end{lemma}

\begin{theorem}
Let $D = \langle \tau, B \rangle$ be a $\dFA$-design and $(\tau_n)$ be a D-consistent typing. The following are equivalent:
\begin{enumerate}[$(1)$]
 \item $(\tau_n)$ is perfect for $D$;

 \item $(\tau_n)$ is both local for $D$ and perfect for each $\hat{D}^k$.
\end{enumerate}
\end{theorem}
\begin{proof}\ONLINE{
$(1) \Rightarrow (2)$: If $(\tau_n)$ is perfect for $D$, then it is sound for each $D^k$, and by Lemma \ref{lemPerfDfa}, sound for $\hat{D}^k$ as well. Suppose that $(\tau_n)$ is not local for some $\hat{D}^k$, there is a string $\sigma_0 u_1 \sigma_1 \ldots u_n \sigma_n \in \langOf{\hat{\Omega}^k}$ (all the stings have this form by definition) not captured by $\sigma(\tau_n)$. By Lemma \ref{lemPerfDfa}, the string $w_0 u_1 w_1 \ldots u_n w_n$ belongs to $\langOf{\tau}$ and as $(\tau_n)$ is perfect, then each $u_i \in \langOf{\tau_i}$: contradiction. Suppose that $(\tau_n)$ is not perfect for some $\hat{D}^k$. There is a sound typing $(\tau_n')$ for $\hat{D}^k$ not contained in $(\tau_n)$, but by Lemma \ref{lemPerfDfa}, $(\tau_n')$ is also sound for $D^k$, so $w(\tau_n') \leq \tau$: again a contradiction because $(\tau_n)$ is perfect.

\medskip

\noindent $(2) \Rightarrow (1)$: If $(\tau_n)$ is perfect for each $\hat{D}^k$ then, by Lemma \ref{lemPerfDfa}, it is sound for each $D^k$. Suppose that it is not perfect for $D$, then there is a sound typing $(\tau_n')$ not contained in $(\tau_n)$ such that, for some $k$, $w(\tau_n') \leq \langOf{\tau}$ for the $k^{th}$ string $w_0,\ldots,w_n$. So $(\tau_n')$ is sound for $D^k$ and also for $\hat{D}^k$. Contradiction.}
\end{proof}

\begin{lemma}\label{CorEperfBdFAco-NP}
$\eperfB{\dFA}$ is in $\class{coNP}$.
\end{lemma}
\begin{proof}\ONLINE{
Let $D = \langle \tau, B \rangle$ be a $\dFA$-design where $B = B_0 \f_1 B_1 \ldots \f_n B_n$ is a kernel box. We can decide in $\class{NP}$ whether $D$ does not admit any perfect typing by preforming the following steps:
\begin{enumerate}
  \item \textbf{Guess}: four string-designs $D^{k_1}$, $D^{k_2}$, $D^{k_3}$, and $D^{k_4}$.

  \item \textbf{Check}: answer \textbf{``yes''} ($D$ does not admit any perfect typing) if at least one of the following holds
  \begin{enumerate}
    \item $D^{k_1}$ does not admit any perfect typing;

    \item $D^{k_2},D^{k_3}$ have different perfect typings;

    \item $D^{k_4}$ admits a perfect typing, say $(\tau_n)$, but it is not local for $D$.
  \end{enumerate}
\end{enumerate}
Each of check $(a)$, $(b)$, and the first part of $(c)$ require polynomial time \cite{MartensNiewerthSchwentick10}. For the second part of check $(c)$ we build $B(\tau_n)$ and prove that $B(\tau_n) < \tau$, namely $B(\tau_n) \cap \overline{\tau} = \emptyset$. Notice that if the yes answer only depends on step $(c)$ this means that $(\tau_n)$ is sound for each $D^{k}$, and so it is not possible that $B(\tau_n) > \tau$. Thus, as $\tau$ is a $\dFA$, its complement has the same size and the intersection emptiness can be done in polynomial time as well.}
\end{proof}

\begin{corollary}\label{EPERFedtdNFAcoNEXP}
$\eperf{\xEDTD{\nFA}}$ is in $\class{coNEXPTIME}$.
\end{corollary}
\begin{proof}\ONLINE{
Let $D = \langle \tau, T \rangle$ be an $\xEDTD{\nFA}$-design. We build from $\tau$ its equivalent normalized version $\tau^d$ that, after all, is a $\xEDTD{\dFA}$ of exponential size. By Corollary \ref{edtd2strPerfReduction} we polynomial-time reduce $\eperf{\rEDTD}$ (for \textbf{normalized} $\rEDTDs$) to $\eperfRB$. So in our case we call $\eperfB{\dFA}$. But, as  $\tau^d$ may be exponentially larger, then, by adapting the upper bound of Lemma \ref{CorEperfBdFAco-NP}, the whole algorithm works in $\class{coNEXPTIME}$.}
\end{proof}

\begin{theorem}\label{PERFedtdNFAcoNEXP}
$\perf{\xEDTD{\nFA}}$ is in $\class{coNEXPTIME}$.
\end{theorem}
\begin{proof}
Let $D = \langle \tau, T \rangle$ be an $\xEDTD{\nFA}$-design, and $(\tau_n)$ be a $D$-consistent typing. Compute in $\class{coNEXPTIME}$ a perfect typing $(\tau_n')$ if there is one. Transform $(\tau_n)$ into a $\xEDTD{\dFA}$-typing of exponential size. As $\equivalence{\dUTAs}$ is in $\class{PTIME}$ then we can decide in $\class{EXPTIME}$ whether $(\tau_n)$ and $(\tau_n')$ are equivalent.
\end{proof}

Unfortunately, for $\ml{\xEDTD{\nFA}}$ we do not have any good algorithm. Let $D = \langle \tau, T \rangle$ be an $\xEDTD{\nFA}$-design, $(\tau_n)$ be a maximal local typing for $D$, and $\kappa$ be de function induced by $(\tau_n)$ and $T$. At the moment, we do not even know whether there could be a (non-maximal) local typing $(\tau_n') < (\tau_n)$ such that $\kappa'  < \kappa$. If there is none, given a local typing $(\tau_n)$ and its induced function $\kappa$, then each maximal local typing that extends $(\tau_n)$ has to induce the same $\kappa$ as well. So we could compare the various $D^x_{\kappa}$ with $(\tau_n)$. But, the only know upper bound is given by the following theorem.

\begin{theorem}\label{MLedtdNFA2EXPSPACE}
$\ml{\xEDTD{\nFA}}$ is in 2-$\class{EXPSPACE}$.
\end{theorem}
\begin{proof}
Let $D = \langle \tau, T \rangle$ be an $\xEDTD{\nFA}$-design and $(\tau_n)$ be a $D$-consistent typing. We can check whether it is not maximal. Check in $\class{EXPTIME}$ whether it is local or not. So, build the normalized type $\tau^d$ from $\tau$. Guess a function $\kappa$ and check whether each $D^x_{\kappa}$ admits a local typing. This is in 2-$\class{EXPSPACE}$ by Corollary \ref{Thme-loc-ml-2Expspace}. So, build the typing $(\tau_n')$ induced by the box-designs. It may be an $\xEDTD{\nFA}$ typing exponentially larger. Check whether $(\tau_n) < (\tau_n')$. This can be done in 2-$\class{EXPTIME}$. So the algorithm works in 2-$\class{EXPSPACE}$ and as this class is closed under complementation we also can decide $\ml{\xEDTD{\nFA}}$ in it.
\end{proof}

Finally, we consider the reduction from trees to boxes for $\eml{\xEDTD{\dRE}}$. The difficulties affecting $\ml{\xEDTD{\nFA}}$ (as we do not know whether there could be a local typing $(\tau_n') < (\tau_n)$ such that $\kappa'  < \kappa$) concern also the existential problem in case of $\dREs$.

\begin{theorem}\label{eMLedtdNFA2EXPSPACE}
$\eml{\xEDTD{\dRE}}$ for \textbf{normalized} $\rEDTDs$ is decidable by an oracle machine in $\class{PSPACE}^\mathcal{C}$ where $\mathcal{C}$ is the complexity class of solving the most difficult problem between $\emlB{\dRE}$ $\elocB{\dFA}$.
\end{theorem}
\begin{proof}
In this case we have to check two sources of maximality depending on the choice of $\kappa$ and on the related box-designs. To do that, we guess a function $\kappa$ (the candidate for a maximal local typing) and we check whether each induced box-design (i) admits a local typing, (ii) is maximal and (iii) is $\dRE$-definable. So, we have to prove that each $\kappa' > \kappa$ does not lead to any local typing. In particular:
\begin{enumerate}
  \item Guess a functions $\kappa$;

  \item Prove that, for each node $x$ of $T$ with $\lab{x} \in \calL$, the answer of $\emlB{\dRE}$ over $D^x_{\kappa}$ is \textbf{``yes''};

  \item Prove that, for each $\kappa' > \kappa$, there is at least a node $x$ of $T$ with $\lab{x} \in \calL$ such that the answer of $\elocB{\dFA}$ over $D^x_{\kappa'}$ is \textbf{``no''}.
\end{enumerate}
We just notice that there could be an exponential number of $\kappa'$ to be enumerated and checked, as well as the number of calls to $\elocB{\dFA}$.
\end{proof}

\section{Conclusion}\label{conclusion}

As explained in the introduction, this work can serve as a basis for designing the distribution of a document. It would be interesting to extend our definitions and methods to richer types of web data. First, this would involve graph data and not just tree data. Then one should consider unordered collections and functional dependencies as in the relational model \cite{VincentLiu03,ArenasLibkin04}. Other dependencies and in particular inclusion dependencies would also clearly make sense in this setting \cite{VincentSchreflLiu04}. More specific design methodology
would also extend the techniques presented in this paper by considering concrete network configurations; this is left for future research.

Database design has a long history, see most database text book. Distributed database design has also been studied since the early days of databases, but much less, because distributed data management was limited by the difficulty to deploy distributed databases. The techniques that were developed, e.g., vertical and horizontal partitioning, are very different from the ones presented here because we focus on ordered trees and collections are not ordered in relational databases. We believe that traditional database studies even on mainly theoretical topics such as normal forms are also relevant in a Web setting. An interesting direction of research is to introduce some of these techniques in our setting.

In the paper, the focus was on local typing that forces verification to be purely local. More generally, it would be interesting to consider typings of the resources that would minimize the communications needed for type checking (and not completely avoid them).
Moreover, it would be interesting to analyze cases where a kernel document may change from time to time by adhering to some global type which uses function symbols in the specification itself. We are investigating this direction. Let us give a short example exhibiting some of the new difficulties that would arise in case kernel document changes were taken into account. Consider the kernel string $w = a\f$ and the type $\tau = a\f?ba^+$. By directly applying the techniques proposed in this paper, it seems clear that $\f?ba^+$ would be the perfect typing for this design. So, one extension of $w$ may be $a\f ba$ (by attaching the tree $\f ba$ complying with the perfect typing) which, in turn, represents a new kernel. But, this extension might still be extended, by attaching again tree $\f ba$, to form $a\f baba$, since the first extension still contains a function call and the perfect typing defined for the remote resource should not vary. This last step could be performed several time. The language obtained by all possible extensions is defined by the type $a\f?(ba^+)^+$, being clearly different from $\tau$. The problem here is that $\tau$ does not express directly a set of trees without taking into account a specific typing. New interesting questions might be: \emph{How to look for typings that are, in a sense, fixpoints w.r.t. the original type with functions?} or \emph{How to avoid irregularities?} or even \emph{Is the perfect typing still unique?} Finally, interesting issues may also come from studying the impact of distributed typing (as studied here) on query optimization.

\section{Acknowledgments}\label{ackn}

This work is a co-operation in the context of the ICT-2007.8.0 FET Open project Nr. 233599:  Foundations of XML - safe processing of dynamic data over the Internet (FOX). Serge Abiteboul's work was supported by the European Research Council Advanced Grant Webdam and by the French ANR Grant Docflow. Georg Gottlob's work was supported by EPSRC grant EP/E010865/1 ``Schema Mappings and Automated Services for Data Integration and Exchange''. Gottlob also gratefully acknowledges a Royal Society Wolfson Research Merit Award. Marco Manna acknowledges the support and hospitality of the Oxford-Man Institute of Quantitative Finance, where he worked as an Academic Visitor. Finally, the authors also want to thank Thomas Schwentick, Wim Martens, and Matthias Niewerth for useful discussions on the problem.

\section*{References}
\bibliographystyle{elsarticle-harv}
\bibliography{biblio}

\end{document}